\documentclass[a4paper,12pt]{article}
\usepackage{amsfonts,amsmath,amssymb,amsthm}
\numberwithin{equation}{section}
\usepackage{txfonts}
\usepackage[utf8]{inputenc}
\usepackage{amsmath}
\usepackage{amssymb}
\usepackage{mathtools}
\usepackage{bbm}
\usepackage{algorithm}
\usepackage{algpseudocode}

\RequirePackage[colorlinks,citecolor=blue,urlcolor=red, linkcolor=blue]{hyperref}
\usepackage[margin=0.8in]{geometry}

\usepackage{graphicx}
\usepackage{epsfig,rotating}
\usepackage{tabularx}
\usepackage{bm}
\usepackage{bbm}
\usepackage{array}
\newcolumntype{P}[1]{>{\raggedright\arraybackslash}p{#1}}

\usepackage{mdframed}
\usepackage{lipsum}
\usepackage{enumitem}

\usepackage[authoryear,square]{natbib}

\usepackage[nottoc,numbib]{tocbibind} 
\usepackage[toc,page]{appendix}

\newtheorem{thm}{Theorem}
\newtheorem{lemma}{Lemma}
\newtheorem{remark}{Remark}
\newtheorem{prop}{Proposition}
\newtheorem{co}{Corollary}

\newcommand{\beaa}{\begin{eqnarray*}}
\newcommand{\eeaa}{\end{eqnarray*}}
\newcommand{\bea}{\begin{eqnarray}}
\newcommand{\eea}{\end{eqnarray}}

\newcommand{\la}{\left\{}
\newcommand{\ra}{\right\}}
\newcommand{\lb}{\left(}
\newcommand{\rb}{\right)}
\newcommand{\mb}{\mathbb}
\newcommand{\ve}{\varepsilon}
\newcommand{\mc}{\mathcal}

\newcommand{\mfB}{\mathsf{\bm B}}
\newcommand{\mfH}{\mathsf{\bm H}}
\newcommand{\mfG}{\mathsf{\bm G}}

\newcommand{\lft}{\left[}
\newcommand{\rht}{\right]}

\begin{document}

\author{
Tuan Pham \\ \small{Department of Statistics and Data Science, University of Texas, Austin} \\  \small{\href{mailto:tuan.pham@utexas.edu}{tuan.pham@utexas.edu}}
\and 
Alessandro Rinaldo \\ \small{Department of Statistics and Data Science, University of Texas, Austin} \\  \small{\href{mailto:alessandro.rinaldo@austin.utexas.edu}{alessandro.rinaldo@austin.utexas.edu}}
\and
Purnamrita Sarkar \\ \small{Department of Statistics and Data Science, University of Texas, Austin} \\  \small{\href{mailto:purna.sarkar@austin.utexas.edu}{purna.sarkar@austin.utexas.edu}}
}

\title{Streaming PCA: averaging from a geometric perspective}

\maketitle

\begin{abstract}

We study principal component analysis (PCA) under memory constraints, a
setting that is increasingly important in large-scale data analysis. Our
focus is on Oja's algorithm, which is a one-pass,  memory-efficient algorithm requiring only
$O(p)$ storage in the rank-one case and $O(pk)$ storage for $k$-
PCA. The main goal is to develop a procedure based on Oja's algorithm that is optimal without prior knowledge of the eigengap, which is otherwise needed to tune the learning rates effectively. We do this by first introducing a geometric perspective on the convergence of
$k$-PCA: by embedding the Oja iterates into the exterior space, we
establish an exact equivalence between $k$-PCA in the ambient space and
$1$-PCA in the exterior space. This geometric viewpoint enables us to
establish several convergence guarantees, including one based only on a
single learning-rate schedule.

Building on this geometric perspective, we develop an averaging theory
for $k$-PCA and show that the resulting averaged estimator is adaptive:
it achieves the nearly optimal convergence rate without prior knowledge of the
eigengap. As an application, we construct memory-efficient estimators
for elliptical component analysis (ECA)
\cite{han2014scale,han2018eca}. Simulation studies and real data analysis are conducted to demonstrate the benefits of our proposed algorithms. 
\end{abstract}

\tableofcontents

\section{Introduction}

Principal component analysis (PCA) is arguably one of the most important methods for extracting useful features and performing dimensionality reduction on large datasets. PCA is known to perform particularly well when the data concentrate around a low-dimensional subspace \cite{jolliffe2016principal}. Over the past several decades, a substantial theoretical understanding of PCA has been developed in a wide range of settings, from classical fixed-dimensional asymptotics \cite{anderson1963asymptotic} to dependent and high-dimensional data. For a recent survey of the subject, see \cite{greenacre2022principal}.

Over the past twenty years, much of the PCA literature has focused on developing theoretical guarantees and methodologies for high-dimensional data. An incomplete list of results under various settings includes \cite{johnstone2001distribution,johnstone2009consistency,vu2013minimax,cai2018rate,baik2005phase}; see also the references therein. It is now well understood that high dimensionality can lead to inconsistency, but that consistency can still be achieved in high-dimensional settings when the eigengap separating the target subspace from its orthogonal complement is sufficiently large. The information-theoretic limits, detection thresholds, and estimators that attain the corresponding lower bounds are also well understood \cite{vu2013minimax}. It is worth noting that almost all methods for performing PCA in the statistics literature are based on the sample covariance matrix or a regularized version of it \cite{ledoit2004well,ledoit2012nonlinear}, and therefore require storing the full matrix in memory.

In recent years, datasets have grown so large that storing and repeatedly processing all observations at once has become costly or impractical. Moreover, a sufficient amount of random-access memory (RAM) to store a full sample covariance matrix can be prohibitively expensive, especially when the ambient dimension is large. As a result, researchers have increasingly adopted low-memory streaming PCA methods that do not require storing the full covariance matrix. A popular approach to performing PCA under memory constraints is to formulate it as a stochastic optimization problem, leading to what is known in the computer science and statistics literature as Oja's algorithm \cite{oja1982simplified,oja1985stochastic}.

As a stochastic-gradient-descent-based method, Oja's algorithm can be unstable during its early phase and requires careful tuning of the learning-rate sequence to achieve the optimal convergence rate. In fact, many prior works \cite{jain2016streaming,huang2021streaming,allen2017first,kumar2023streaming,kumar2024oja} have shown that some prior knowledge of the eigengap is required to achieve optimal convergence rates. From a statistical perspective, such a requirement is undesirable. This raises a natural question: without any prior knowledge of the population eigengap, how accurate can the output of an Oja-type algorithm be? This question is the main motivation for the present article.

In this paper, we address the above question using the idea of iterate averaging \cite{polyak1992acceleration,ruppert1988efficient}. Our main result is a general averaging framework for $k$-PCA that combines a sequence of individually suboptimal Oja iterates and produces a nearly optimal estimator at the terminal time. Although the idea of averaging is not new to the literature, developing a rigorous theory for non-convex problems  is highly nontrivial. Our framework requires new ideas and insights involving time-uniform concentration bounds, multilinear algebra, and several technical extensions. As a consequence of this framework, we establish a convergence result for $k$-PCA that uses only a single learning-rate schedule, which is of independent interest. As a concrete application, we develop a memory-efficient procedure for elliptical component analysis \cite{han2018eca,han2014scale,zhao2024spatial}. The proposed procedure uses only $O(p)$ memory to estimate the leading eigenvector under an elliptical model, potentially in the presence of heavy-tailed data.

We emphasize that it has long been known that the analysis of stochastic approximation methods, and of Oja's algorithm for streaming
PCA in particular, is considerably more difficult for $k$-PCA with $k>1$ than for the rank-one case $k=1$. Some remarkable results that provide theoretical analysis for the case $k>1$ include \cite{de2015global}, \cite{allen2017first} and \cite{huang2021streaming}. The latter two papers use quite different techniques: \cite{allen2017first} employs a very careful analysis of the recursion generated by the one-step update, while \cite{huang2021streaming} relies on matrix concentration inequalities. The approach developed in the present paper is intuitive and avoids the need for using heavy matrix concentration inequalities. 
In summary, our contributions are threefold.
\begin{itemize}
    \item We introduce a geometric perspective on one-pass algorithms for
    $k$-PCA. By enlarging the ambient space, we establish a one-to-one
    correspondence between $k$-PCA algorithms in the original ambient space
    and $1$-PCA algorithms in the enlarged space. Consequently, convergence
    and concentration properties established for $1$-PCA can be transferred
    seamlessly to $k$-PCA.

    \item Using this geometric perspective, we develop an averaging theory
    for Oja's algorithm that assumes only a random initialization. We show
    that, by using a two-phase algorithm, the averaged estimator achieves the
    information-theoretic lower bound, up to first order. One of the main
    advantages of the averaged estimator is that it does not require
    knowledge of the unknown eigengap.

    \item As an application, we apply the developed averaging theory to
    elliptical component analysis
    \cite{han2018eca,han2014scale} and construct an efficient estimator for
    $k$-PCA that requires only $O(pk)$ memory, in contrast to the $O(p^2)$
    memory required by the methods in
    \cite{han2018eca,han2014scale}.
\end{itemize}

\subsection*{Notation}
For two matrices  $\bm A, \bm B \in \mb R^{p \times k}$ with orthonormal columns, we denote by $\bm \Theta  \lb \bm A, \bm B \rb$ the $k$-by-$k$ diagonal matrix consisting of the principal angles between $\bm A$ and $\bm B$, namely
\[
\bm \Theta  \lb \bm A, \bm B \rb := \mbox{diag} \la \theta_1,\dots,\theta_k \ra
\]
where $\cos \theta_1\geq \cos \theta_2\geq \dots \geq \cos \theta_k \in [0,1]$ are the singular values of $\bm A^\top \bm B$.

We write $\sin \bm \Theta \lb \bm A, \bm B \rb$ to denote the sine matrix 
\[
\sin \bm \Theta \lb \bm A, \bm B \rb:= \mbox{diag} \la \sin \theta_1, \sin \theta_2,\dots,\sin \theta_k \ra.
\]
The norms $\|.\|_{\rm  op}$ and $\| . \|_{\rm F}$ are the operator norm and Frobenius norm of a matrix, respectively, i.e.
\begin{align*}
\| \bm A_{n \times p} \|_{\rm op}:&= \sup_{\| \bm x \|_2=1} \| \bm A \bm x \|_2;  \\
\| \bm A_{n \times p} \|_{\rm  F}:&=  \sqrt{\mbox{Tr} \lb \bm A \bm A^\top \rb}.
\end{align*}
Throughout the paper, the notation $A_n \lesssim B_n$ indicates that $A_n \leq C B_n$ for some universal constant $C>0$.  We will also write $A_n \lesssim_{a,b} B_n$ to indicate that the constant $C$ may depend on $a$ and $b$,  but is still free of $n$. 
For a full-column-rank matrix $\bm A\in\mathbb R^{p\times k}$, define the map
\[
\operatorname{polar}(\bm A)
:=
\bm A(\bm A^\top\bm A)^{-1/2}
\]
which is the unitary component in the polar decomposition of a matrix $\bm A$.


\section{Background on Oja's algorithm}

Suppose we want to perform PCA on the data $\bm{\mc{X}}_1,\dots,\bm{\mc{X}}_n$. 
The most common approach, as in \cite{han2014scale,han2018eca,zhao2024spatial}, is to compute the unnormalized second-moment matrix, or a variant of it:
\[
\hat{\bm{\Sigma}}_n:= \sum_{k=1}^n \bm{\mc{X}}_k \bm{\mc{X}}_k^\top.
\]
The drawback of this approach is that it requires $O(p^2)$ memory to store such matrices. 
Oja's algorithm \cite{oja1982simplified,oja1985stochastic} was proposed as a memory-efficient iterative algorithm for estimating principal components. 
Let us explain the idea behind Oja's algorithm. 
One can view the task of estimating the leading principal component as the optimization problem
\[
\mbox{argmin}_{\bm{x} \in \mb{S}^{p-1}} \la -\bm{x}^\top \mb{E} \lb \bm{\mc{X}}_1 \bm{\mc{X}}_1^\top  \rb \bm{x} \ra.
\]
This optimization problem is nonconvex, although its landscape has no spurious local minima. 
Oja's algorithm \cite{oja1982simplified,oja1985stochastic} solves this problem through a projected stochastic approximation method: it draws an initialization $\bm{x}_0$ uniformly from the hypersphere and then iteratively updates
\[
\bm{x}_{t}:= \frac{\lb \bm{I}_p + \eta_t \cdot \bm{\mc{X}}_t \bm{\mc{X}}_t^\top \rb \bm{x}_{t-1} }{\Big\| \lb \bm{I}_p + \eta_t \cdot \bm{\mc{X}}_t \bm{\mc{X}}_t^\top \rb \bm{x}_{t-1} \Big\|}.
\]
Here $\la \eta_t \ra$ is a sequence of learning rates that must be chosen appropriately. The above formulation is the simplest form of Oja's algorithm and there are other variants of it that are asymptotically equivalent; see, for example, \cite{chou2020ode,de2015global} and the references therein.
It is known that, in order to obtain a sequence of consistent estimators $\la \bm{x}_t \ra$, one needs
\[
\sum_{t=1}^\infty \eta_t = \infty 
\qquad \text{and} \qquad 
\sum_{t=1}^\infty \eta_t^2< \infty.
\]
The first condition guarantees that any limit point of $\la \bm{x}_t \ra$ is the true eigenvector, while the second condition guarantees that the sequence $\la \bm{x}_t \ra$ converges to a limit. 
The best possible choice of learning rates is therefore $\eta_t=C/t$ for some sufficiently large constant $C$. 
Theoretically, it is known that a choice of the form $C \gtrsim 1/(\text{eigengap})$ is required to obtain a $1/t$ convergence rate. 
This is an undesirable drawback in practice, since we rarely have prior information about the eigengap. 
We address this issue in the present paper.

Oja's algorithm has been analyzed and extended in several directions over the last few years. 
An incomplete list of early results includes \cite{jain2016streaming,de2015global,hardt2014noisy,balsubramani2013fast,shamir2016convergence,mitliagkas2013memory,balcan2016improved}. 
More recent works, which establish optimal rates and extensions to various settings, can be found in \cite{huang2021streaming,lunde2021bootstrapping,kumar2024oja,allen2017first,lowprecPCA2025,Kumar25entry}; see also the references therein.

\section{A general convergence theorem for $k$-PCA} \label{sec:k-PCA-convergence}


Let $\bm X_1, \cdots, \bm X_n \in \mb R^{p}$ be a sequence of centered i.i.d. random vectors such that for some constant $M>0$,
\[
\mb E \bm X_1 \bm X_1^\top = \bm \Sigma, \qquad \| \bm X_1 \bm X_1^\top - \bm \Sigma \| \leq M \, \text{a.s.}
\]
Suppose we are interested in estimating the $k$-principal subspace $\bm U^\star_{k}$ of $\bm \Sigma$ based on observing $\bm X_i$'s. Let $\lambda_1,\dots,\lambda_p$ be the eigenvalues of $\bm \Sigma$ sorted in descending order. For identifiability, we assume that the eigengap between the $k$-th and $(k+1)$-th eigenvalues is positive
\[
\Delta_k:= \lambda_k - \lambda_{k+1} >0.
\]
In the context of $k$-PCA, Oja's algorithm can be formulated as in Algorithm \ref{alg:oja-k-pca}. Here $\rm QR$ refers to the QR decomposition that outputs an orthogonal basis from a given matrix. 

\begin{algorithm}[t]
\caption{Oja's algorithm for \(k\)-PCA}
\label{alg:oja-k-pca}
\begin{algorithmic}[1]
\Require Data \(\bm X_1,\ldots,\bm X_n\in\mb R^p\), target rank
\(k\), and learning rates \(\eta_1,\ldots,\eta_n>0\).
\State Draw \(\bm G_0\in\mb R^{p\times k}\) with iid \(N(0,1)\) entries.
\State Set \(\bm Q_0\gets\operatorname{QR}(\bm G_0)\).
\For{\(t=1,\ldots,n\)}
    \State
    \[
        \bm Q_t
        \gets
        \operatorname{QR}\!\left[
            \bm Q_{t-1}
            +
            \eta_t\bm X_t
            \bigl(\bm X_t^\top\bm Q_{t-1}\bigr)
        \right].
    \]
\EndFor
\State \Return \(\bm Q_n\).
\end{algorithmic}
\end{algorithm}

This is precisely the streaming PCA setting in \cite{huang2021streaming,allen2017first}. We note that the results in \cite{huang2021streaming} use specialized matrix concentration inequalities while the technique in \cite{allen2017first} is based on a very technical argument involving recursion bounds. Both \cite{huang2021streaming} and \cite{allen2017first} require a two-phase analysis for a uniformly distributed initialization: they first prove that with constant learning rates, the algorithm can find a neighborhood of the principal eigenvector, and with a new learning-rate schedule of order $\Tilde{\Theta}(1/n)$, it can achieve the optimal convergence rate.

Before presenting the convergence results, we first make some remarks on the assumptions posed here. First, the assumption that the data are centered is only for the sake of presentation. For non-centered data, one can replace the update by using two data points $\bm X_i, \bm X_{i+1}$ at a time and change the update rule to 
\[
\frac{1}{2} \lb \bm X_i - \bm X_{i+1} \rb \cdot  \lb \bm X_i - \bm X_{i+1} \rb^\top.
\]
Furthermore, the boundedness assumption we impose here is also not essential. In fact, the convergence result we present below extends easily to the case of the data having sub-Gaussian tails as well. We keep the boundedness assumption here in order to have an easy, direct comparison with existing results \cite{huang2021streaming,allen2017first}. 

To state our convergence result, define 
\[
\bm A_t:= \bm X_t \bm X_t^\top, \qquad L:= \lambda_1 + M, \qquad \mu_1:= \sum_{i=1}^k \lambda_i,
\]
and 
\[
 r_k(\delta) = k  \cdot
    \exp\Bigg\{
        C
        \sqrt{
            \log(k) \cdot 
            \log\left(\frac{C}{\delta}\right)
        }
        +
        C
        \log\left(\frac{C}{\delta}\right)
    \Bigg\}
\]
where $C>0$ is a universal constant that is sufficiently large. 

\begin{thm} \label{k-PCA convergence}
    Fix $\delta>0$ and let $\la \eta_t; \, 1 \leq t \leq n \ra$ be an arbitrary sequence of positive learning rates such that 
    \begin{align} \label{step size cond}
    \max_{t \leq n} \eta_t \leq \frac{1}{4L}; \qquad L\mu_1 \cdot  \sum_{i=1}^n \eta_t^2 \leq c\delta 
    \end{align}
    Then, with probability at least $1-\delta$, we have 
\begin{align} \label{k-PCA-F-bound}
    \| \sin \bm \Theta \lb  \bm Q_n, \bm U^\star_k \rb \|_{\rm F}^2 \lesssim \frac{r_k(\delta)}{\delta} \cdot  \left[  \binom{p}{k} \cdot \exp \lb  -2\Delta_k  \sum_{i=1}^n \eta_i  \rb + \sigma_k^2 \sum_{i=1}^n \eta_i^2 \cdot \exp \lb -2\Delta_k \cdot \sum_{j=i+1}^n \eta_j \rb \right]
\end{align}
where
\[
\sigma_k^2:= \sum_{i=1}^k\sum_{j=k+1}^p \mathbb E \left[  \left(  \bm u_j^\top (\bm A_t-\bm\Sigma) \bm u_i \right)^2 \right]
.
\]
\end{thm}

Theorem \ref{k-PCA convergence} is an extension of \cite{jain2016streaming} to the case $k>1$. It is easy to check that our result recovers Theorem 3.1 in \cite{jain2016streaming} by setting $k \equiv 1$ in Theorem \ref{k-PCA convergence}, albeit with a slightly different polynomial dependence on $\delta$. We note that any polynomial dependence on $1/\delta$ can be improved to $\mbox{polylog}(1/\delta)$ by employing the bag-of-medians trick, see Algorithm \ref{alg:grassmann-median} below for more details. To the best of our knowledge, Theorem \ref{k-PCA convergence} is the first $k$-PCA analog of \cite{jain2016streaming} that gives an upper bound for an {\it arbitrary} choice of learning rates. Providing finite sample guarantees for streaming $k$-PCA is a challenging task overall, and most analyses in the literature adopt a two-phase analysis that involves finding a warm start by using large, constant learning rates, and then choosing learning rates of order $1/t$ \cite{huang2021streaming,allen2017first}.

The main message of Theorem \ref{k-PCA convergence} is not just about the error bound. In the proof of Theorem \ref{k-PCA convergence}, we introduce a new, conceptually simple technique that creates a one-to-one correspondence between $k$-PCA with $k>1$ and $1$-PCA. Roughly speaking, using the so-called Pl\"ucker embedding, we demonstrate that deriving a finite-sample bound for $k$-PCA with $k>1$ is equivalent to doing so for $1$-PCA in an enlarged space, where the magnitude of all parameters can be kept up to a factor that grows slower than any polynomial of the model parameters. 

Although our technique is conceptually simple, we expect it to work for a large class of subspace estimation problems using a stochastic approximation framework. Despite being simple, the analysis requires some probabilistic insight. In the context of Oja's algorithm, the initialization is no longer uniformly distributed on the enlarged space and requires a careful analysis of the anti-concentration properties of the initialization.  This is apparent in the  dependence on $k$ in the right-hand side of \eqref{k-PCA-F-bound}: we see that $r_k(\delta)$ is of order $k^{1+o(1)}$ as $k \to \infty$. The $o(1)$ term here is of order $\exp \lb \Theta \lb \sqrt{\log k} \rb \rb$, which is higher than $\mbox{polylog}(k)$, and is due to an anti-concentration bound needed for the analysis of initialization.  
As a consequence of Theorem \ref{k-PCA convergence}, we can recover the error bound that matches information lower bound \cite{huang2021streaming,allen2017first}. To this end, define 
\[
m
:=
\left\lceil
C_0
\frac{L\mu_1}{\delta\Delta_k^2} \cdot 
k^2\log^2\left(\frac{ep}{k}\right)
\right\rceil
\]
and
\[
\beta
:=
C_0
\max\left\{
1,\,
\frac{L}{\Delta_k},\,
\frac{L\mu_1}{\delta\Delta_k^2}
\right\},
\]
where $C_0>0$ is a sufficiently large universal constant. Choose
the learning rates
\begin{equation}
\label{eq:two-phase-kpca-step-size}
\eta_t
:=
\begin{cases}
\displaystyle
\sqrt{
\frac{c_0\delta}{L\mu_1m}
},
&1\le t\le m,
\\[1.2em]
\displaystyle
\frac{6}{
\Delta_k\lb \beta+t-m\rb
},
&m<t\le n,
\end{cases}
\end{equation}
where $c_0>0$ is a sufficiently small universal constant.

\begin{co} \label{cor:two-phase-kpca}
Suppose the assumptions of Theorem~\ref{k-PCA convergence} hold,
and fix $\delta\in(0,1/4)$. There exists a sufficiently large universal constant $C_1>0$ such
that, if
\begin{equation}
\label{eq:two-phase-kpca-sample-size}
n
\ge
C_1
\left[
\frac{L\mu_1}{\delta\Delta_k^2}
k^2\log^2\left(\frac{ep}{k}\right)
+
\left(
\max\left\{
1,\,
\frac{L}{\Delta_k},\,
\frac{L\mu_1}{\delta\Delta_k^2}
\right\}
\right)^{6/5}
\right],
\end{equation}
then the output of Algorithm~\ref{alg:oja-k-pca} satisfies
\begin{equation}
\label{eq:two-phase-kpca-rate}
\left\|
\sin\bm\Theta
\left(
\bm Q_n,\bm U_k^\star
\right)
\right\|_{\rm F}^2
\lesssim
\frac{r_k(\delta)}{\delta} \cdot 
\left[
\frac{\sigma_k^2}{\Delta_k^2n}
+
\frac1{n^2}
\right]
\end{equation}
with probability at least $1-\delta$.
\end{co}

In particular, when $L$, $\mu_1$, and $\Delta_k^{-1}$ are bounded,
a simple, sufficient condition for  \eqref{eq:two-phase-kpca-sample-size} is
\[
n
\gtrsim
\operatorname{poly}\left(\frac1\delta\right)
k^2\log^2\left(\frac{ep}{k}\right).
\]
Note that the $k^2$ dependence in the above condition is due to the fact that our loss metric is the Frobenius distance, which is expected to be larger than the operator norm by a factor $\Omega(\sqrt{k})$. Given this, our bound essentially matches \cite{huang2021streaming} up to polynomial factors in $1/\delta$ and polylog factors in terms of dimension. As noted above, the polynomial dependence on $1/\delta$ can be reduced to logarithmic dependence by employing the bag-of-medians trick in Algorithm \ref{alg:grassmann-median}.

\section{Averaging scheme for k-PCA}

The convergence result in Section \ref{sec:k-PCA-convergence} provides a general convergence result for Oja's algorithm for $k$-PCA using a new geometric perspective from multilinear algebra. However, a crucial problem with Theorem \ref{k-PCA convergence} and Algorithm \ref{alg:oja-k-pca} is that knowledge of the eigengap is required to choose the learning rates appropriately. As we have seen in Theorem \ref{k-PCA convergence}, such knowledge is crucial for having optimal convergence rates. In this section, we attempt to provide a partial solution to this issue by adapting a classical idea from the stochastic optimization literature: the Polyak--Ruppert averaging scheme \cite{polyak1992acceleration,ruppert1988efficient}. Roughly speaking, the averaging scheme allows us to run the algorithm with a generally sub-optimal choice of learning rates, and then average these sub-optimal estimators to get an optimal one, up to first order. For recent surveys and developments of the Polyak--Ruppert averaging scheme, the readers are referred to \cite{moulines2011non,tripuraneni2018averaging,khodadadian2026general,durmus2025finite} (see also the references therein).

In general, it is expected that a low-memory, adaptive algorithm (in the sense that it does not have any unknown tuning parameters) will be sub-optimal to some degree \cite{carmon2024price}. Note that proving such a result is a very difficult problem because we currently do not have any information-theoretic tools to prove these kinds of results. We will show that although the averaging scheme may not be able to achieve exactly the information lower bound, it achieves this bound up to {\it first order}. This means that the suboptimality appears only in terms that decay like $n^{-1-\Theta(1)}$. As in the case of linear stochastic approximation, we will show that the first-order error term in our error bound matches that of the maximum likelihood estimator (MLE), which is known to have the smallest possible variance by the Cram\'er--Rao lower bound \cite{anderson1963asymptotic}.

Our averaging framework consists of two phases. We run Oja's algorithm with a constant learning rate in phase I to find a good initialization, which is better than random guessing. Starting from the good initialization obtained in phase I, we run Oja's algorithm with a sequence of sub-optimal learning rates which result in a sequence of iterates with a suboptimal convergence rate. We then show that an appropriate averaging of these sub-optimal iterates will result in an estimator with an optimal convergence rate. Note that our procedure is fully adaptive and does not involve unknown parameters. To describe the final algorithm, we split it into a few algorithms outlined in the sections below.

\begin{algorithm}[H]
\caption{Polyak--Ruppert Oja for $k$-PCA}
\label{alg:adaptive-pr-kpca}
\begin{algorithmic}[1]
\Require Data $\bm X_1,\ldots,\bm X_n\in\mathbb R^p$;
target rank $k$; exponent $\alpha\in(1/2,1)$;
confidence level $\delta\in(0,1/8)$.

\State Set
\[
m\gets\left\lfloor\frac n2\right\rfloor,
\qquad
N\gets n-m,
\qquad
\delta_{\rm w}\gets\frac{\delta}{4},
\qquad
c_\alpha\gets100(1+\alpha).
\]

\State Set
\[
R\gets
\left\lceil
24\log\frac{1}{\delta_{\rm w}}
\right\rceil
=
\left\lceil
24\log\frac{4}{\delta}
\right\rceil,
\qquad
m_0\gets\left\lfloor\frac{m}{R}\right\rfloor,
\qquad
\eta_{\rm w}\gets c_\alpha n^{-\alpha}.
\]

\State Apply Algorithm~\ref{alg:kpca-warm-start} to
$\bm X_1,\ldots,\bm X_m$, using the parameters
$R,m_0,\eta_{\rm w},\delta_{\rm w}$ and the distance
\[
d(\bm W,\bm U)
:=
\|\sin\bm\Theta(\bm W,\bm U)\|_{\rm F}
=
\sqrt{k-\|\bm W^\top\bm U\|_{\rm F}^2}.
\]
Denote its output by $\widehat{\bm U}_{\rm warm}$.

\State Apply Algorithm~\ref{alg:kpca-phase-two} to
$\bm X_{m+1},\ldots,\bm X_n$, initialized at
\[
\bm Q_m\gets\widehat{\bm U}_{\rm warm},
\]
and using the learning rates
\[
\eta_t\gets t^{-\alpha},
\qquad
t=m+1,\ldots,n.
\]

\State In Algorithm~\ref{alg:kpca-phase-two}, get $\widehat{\bm U}_{\rm avg}$ by applying
Algorithm~\ref{alg:kpca-averaging} online to
\[
\bm Q_m,\bm Q_{m+1},\ldots,\bm Q_{n-1}.
\]

\State \Return 
$\widehat{\bm U}_{\rm avg}$.
\end{algorithmic}
\end{algorithm}

\subsection{Boosting and averaging}

 Let us present here two useful algorithms: one is a probabilistic technique known as bag-of-medians that can be used to improve the error bound for various iterative algorithms, and one is an averaging analog of the usual Euclidean averaging on the space of orthogonal subspaces. These two algorithms have been used in Algorithm \ref{alg:adaptive-pr-kpca} above as the pre-processing steps.  

We begin with the bag-of-medians algorithm. The form of this technique, tailored to our setting, can be found in Algorithm \ref{alg:grassmann-median} below. 

\begin{algorithm}[H]
\caption{Boosting with bag-of-medians}
\label{alg:grassmann-median}
\begin{algorithmic}[1]
\Require Orthonormal frames
$\bm W^{(1)},\ldots,\bm W^{(R)}\in\mathbb R^{p\times k}$.

\State Set $q\gets\lfloor R/2\rfloor+1$.

\For{$r=1,\ldots,R$}
    \State Let $\rho_r$ be the $q$-th smallest element of
    \[
    \left\{
    d\bigl(\bm W^{(r)},\bm W^{(s)}\bigr):
    1\le s\le R
    \right\}.
    \]
\EndFor

\State Choose
\[
\widehat r\in
\operatorname*{argmin}_{1\le r\le R}\rho_r.
\]

\State \Return $\bm W^{(\widehat r)}$.
\end{algorithmic}
\end{algorithm}

The distance $d$ in Algorithm \ref{alg:grassmann-median} can be any distance between subspaces (i.e., any proper metric on the Grassmannian manifold). We have the following performance guarantee for Algorithm \ref{alg:grassmann-median}, which, roughly speaking, states that we can select an estimator among a subset of independent estimators that is guaranteed to behave better than the typical performance of each estimator. Formally, we have
\begin{prop} \label{prop:boosting}
    Let $\bm W^{(1)},\ldots,\bm W^{(R)}\in\mathbb R^{p\times k}$ be independent. Assume that for some $a>0$,
    \[
    \min_{1 \leq r \leq R} \mb P \lb d^2 \lb \bm W^{(r)}, \bm U_\star \rb \leq a \rb \geq \frac 34.
    \]
    Let $\bm W^{(\widehat r)}$ be the output of Algorithm~\ref{alg:grassmann-median}. Then
    \[
    \mb P \lb d^2\bigl(\bm W^{(\widehat r)},\bm U_\star\bigr)
\le 9a \rb \geq 1-\delta
    \]
    as long as $R \geq 24 \log \lb 1/\delta \rb$.
\end{prop}

Proposition \ref{prop:boosting} gives a probability boost from $3/4$ (which can be replaced by any number greater than $1/2$) to $1-\delta$ and requires only $\Theta \lb \log \lb 1/\delta \rb \rb$ estimators to execute. To see how it is useful, let us take a look at the error bound in Theorem \ref{k-PCA convergence}. The right-hand side of \eqref{k-PCA-F-bound} involves polynomial dependence on $\delta$, and by running Algorithm \ref{alg:grassmann-median}, we can improve it to $\mbox{polylog}(1/\delta)$. The proof of Proposition \ref{prop:boosting} is based on a simple application of Hoeffding's inequality (or binomial tail bound) and is thus omitted. We will make use of this boosting algorithm in what follows. 

The second algorithm we will be using is the averaging analog of a vector-valued sequence in $\mb R^d$. In the case $k=1$, the averaging notion is straightforward: one can take the Euclidean averaging and project it back onto the hypersphere. Let us state a natural extension of this averaging concept to the case $k \geq 2$ in Algorithm \ref{alg:kpca-averaging} below. 
\begin{algorithm}[H]
\caption{Averaging for $k$-subspaces}
\label{alg:kpca-averaging}
\begin{algorithmic}[1]
\Require A stream of Oja iterates
$\bm Q_{t_0},\ldots,\bm Q_n\in\mathbb R^{p\times k}$ satisfying
$\bm Q_t^\top\bm Q_t=\bm I_k$.

\State Initialize $\overline{\bm Q}_{t_0}\gets\bm Q_{t_0}$.

\For{$t=t_0+1,\ldots,n$}
    \State Orthogonalize the current average:
    \[
    \widetilde{\bm Q}_{t-1}
    \gets
    \operatorname{polar}(\overline{\bm Q}_{t-1}).
    \]

    \State Compute an SVD
    \[
    \bm Q_t^\top\widetilde{\bm Q}_{t-1}
    =
    \bm L_t\bm D_t\bm U_t^\top,
    \qquad
    \bm R_t\gets\bm L_t\bm U_t^\top.
    \]

    \State Set $\gamma_t\gets (t-t_0+1)^{-1}$ and update
    \[
    \overline{\bm Q}_t
    \gets
    (1-\gamma_t)\overline{\bm Q}_{t-1}
    +
    \gamma_t\bm Q_t\bm R_t.
    \]
\EndFor

\State Set
\[
\widehat{\bm U}_{\rm avg}
\gets
\operatorname{polar}(\overline{\bm Q}_n).
\]

\State \Return $\widehat{\bm U}_{\rm avg}$.
\end{algorithmic}
\end{algorithm}

We note that averaging on Riemannian manifolds can be done in many other ways (and are in fact equivalent if the estimators are close to the true subspace), and we only choose the most natural concept of averaging that is also convenient for our analyses.

\subsection{Phase I: seeking a warm start}
Fix $\alpha\in(1/2,1)$ and
$\delta\in(0,1/2)$. In Algorithm~\ref{alg:kpca-warm-start}, set
\[
R
:=
\left\lceil 24\log\frac1\delta\right\rceil,
\qquad
m_0
:=
\left\lfloor\frac{n}{2R}\right\rfloor,
\qquad
\eta_{\rm warm}
:=
\frac{c_\alpha}{n^\alpha},
\]
where $c_\alpha>0$ depends only on $\alpha$. 

Additionally, suppose that
\begin{align}
\label{eq:kpca-warm-basic-conditions}
n\ge 4R,
\qquad
\eta_{\rm warm}\le\frac{1}{4L},
\qquad
L\mu_1 \cdot m_0 \cdot \eta_{\rm warm}^2\le c_0,
\end{align}
for some universal constant $c_0>0$ sufficiently small.

\begin{algorithm}[H]
\caption{Phase I: confidence-boosted warm start}
\label{alg:kpca-warm-start}
\begin{algorithmic}[1]
\Require Data $\bm X_1,\ldots,\bm X_m\in\mathbb R^p$;
target rank $k$; exponent $\alpha\in(1/2,1)$;
confidence level $\delta\in(0,1)$.

\State Set
\[
R\gets\left\lceil24\log(1/\delta)\right\rceil,
\qquad
b\gets\left\lfloor\frac{m}{R}\right\rfloor,
\qquad
\eta_{\rm w}\gets c_\alpha n^{-\alpha}.
\]

\For{$r=1,\ldots,R$}
    \State Define the disjoint block
    $
    \mathcal I_r
    :=
    \{(r-1)b+1,\ldots,rb\}.
    $

    \State Run Algorithm~\ref{alg:oja-k-pca} on
    $\{\bm X_t:t\in\mathcal I_r\}$ with 
    constant learning rate $\eta_{\rm w}$, and a uniformly distributed 
    initialization. Denote its output by $\bm W^{(r)}$.
\EndFor

\State Apply Algorithm~\ref{alg:grassmann-median} to
$\bm W^{(1)},\ldots,\bm W^{(R)}$ and denote its output by
$\widehat{\bm U}_{\rm warm}$.

\State \Return $\widehat{\bm U}_{\rm warm}$.
\end{algorithmic}
\end{algorithm}

Algorithm \ref{alg:kpca-warm-start} above looks for a warm start from a random initialization. Note that we have used the constant learning rate $\eta_w$ with a sub-optimal exponent $\alpha \in (1/2,1)$.  This is fine for our purpose, which is to do averaging later on. The constant $c_\alpha$ is flexible, and for theoretical analysis, we will fix it to be $c_\alpha=100(1+\alpha)$. We note that this constant is only chosen for theoretical analysis, and in practice, one might want to go for a lower constant to achieve better finite sample results. 
The distance $d$ in Step 6 of Algorithm \ref{alg:kpca-warm-start} is chosen to be the metric 
\[
d(\bm W,\bm U)
:=
\|\sin\Theta(\bm W,\bm U)\|_{\rm F}.
\]

\begin{thm}\label{thm:kpca-warm-start}
Let $\widehat{\bm U}_{\rm warm}$ denote the output of
Algorithm~\ref{alg:kpca-warm-start}, and suppose
\eqref{eq:kpca-warm-basic-conditions} holds. Then, with probability at
least $1-\delta$,
\begin{align}
\left\|
\sin\Theta
\left(
\widehat{\bm U}_{\rm warm},
\bm U_k^\star
\right)
\right\|_{\rm F}^2
\le
C_\alpha \cdot  r_k(1/4) \cdot 
\left[
\binom pk \cdot 
\exp\left\{
-\frac{
\widetilde{c}_\alpha \cdot \Delta_k \, n^{1-\alpha}
}{
\log(1/\delta)
}
\right\}
+
\frac{\sigma_k^2}{\Delta_k \,  n^\alpha}
\right],
\label{eq:kpca-warm-exact}
\end{align}
where $\widetilde{c}_\alpha,C_\alpha>0$ depend only on $\alpha$.
\end{thm}

As a direct consequence of Theorem \ref{thm:kpca-warm-start}, we can see that there exists a constant $C_\alpha>0$, depending only
on $\alpha$, such that if
\begin{align}
n^\alpha
\ge C_\alpha L,
\qquad 
n^{2\alpha-1}
\ge C_\alpha L\mu_1,
\qquad 
\frac{\Delta_k \, n^{1-\alpha}}{\log \lb 1/\delta \rb}
\ge
C_\alpha
\left\{
k\log \lb \frac{ep}{k} \rb+\log(en)
\right\},
\label{eq:kpca-warm-contraction-condition}
\end{align}
then
\begin{equation}
\left\|
\sin\Theta(
\widehat{\bm U}_{\rm warm},
\bm U_k^\star
)
\right\|_{\rm F}^2
\le
C_\alpha \cdot  r_k(1/4) \cdot 
\left(
1+\frac{\sigma_k^2}{\Delta_k}
\right) \cdot n^{-\alpha}
\end{equation}
with probability at least $1-\delta$.  

The first two conditions in \eqref{eq:kpca-warm-contraction-condition} guarantee that \eqref{step size cond} holds, while the third condition bounds the initialization term in \eqref{eq:kpca-warm-exact}. Additionally, if 
\begin{equation}
\label{eq:kpca-warm-neighborhood-condition}
n^\alpha
\ge
2C_\alpha r_k(1/4)
\left(
1+\frac{\sigma_k^2}{\Delta_k}
\right),
\end{equation}
then with probability at least $1-\delta$,
\[
\left\|
\sin\Theta(
\widehat{\bm U}_{\rm warm},
\bm U_k^\star
)
\right\|_{\rm F}^2
\le\frac12.
\]

\subsection{Phase II: averaging}

Let
\[
m:=\left\lfloor\frac n2\right\rfloor,
\qquad
N:=n-m.
\]
We initialize the second phase at $\bm Q_m:=\widehat{\bm U}_{\rm warm}$, where $\widehat{\bm U}_{\rm warm}$ is the output of
Algorithm~\ref{alg:kpca-warm-start}. For $t=m+1,\ldots,n$, choose the parameter-free learning rates
\[
\eta_t:=t^{-\alpha},
\qquad
\alpha\in  \lb \frac 12, 1  \rb,
\]
and generate the Oja iterates
\begin{equation}
\label{eq:phase-two-oja-update}
\bm Q_t
=
\operatorname{QR}\left[
\left(
\bm I_p+\eta_t\bm A_t
\right)\bm Q_{t-1}
\right],
\qquad
\bm A_t=\bm X_t\bm X_t^\top.
\end{equation}
Concurrently, we apply Algorithm~\ref{alg:kpca-averaging} to the stream $\bm Q_m,\bm Q_{m+1},\ldots,\bm Q_{n-1}$.
The output of the averaging algorithm is denoted by
$\widehat{\bm U}_{\rm avg}$.

\begin{algorithm}[H]
\caption{Phase II: averaging from a warm-start}
\label{alg:kpca-phase-two}
\begin{algorithmic}[1]
\Require Data $\bm X_{m+1},\ldots,\bm X_n$;
warm start $\widehat{\bm U}_{\rm warm}$;
exponent $\alpha\in(1/2,1)$.

\State Set
\[
\bm Q_m\gets\widehat{\bm U}_{\rm warm}.
\]

\For{$t=m+1,\ldots,n$}
    \State Set $\eta_t\gets t^{-\alpha}$ and update
    \[
    \bm Q_t
    \gets
    \operatorname{QR}\left[
    \left(
    \bm I_p+\eta_t\bm X_t\bm X_t^\top
    \right)\bm Q_{t-1}
    \right].
    \]
\EndFor

\State Apply Algorithm~\ref{alg:kpca-averaging} online to
\[
\bm Q_m,\ldots,\bm Q_{n-1}.
\]

\State \Return its output $\widehat{\bm U}_{\rm avg}$.
\end{algorithmic}
\end{algorithm}
Fix $\alpha\in(1/2,1)$ and $\delta\in(0,1/8)$. To state the finite-sample guarantee for the final algorithm, define
\begin{align} \label{ell}
\ell_\delta := 1+\log\frac{8}{\delta},
\qquad
\ell_{n,\delta} := 1+\log\frac{8n}{\delta}.
\end{align}
Put
\begin{equation}
\label{eq:phase-two-Mn}
M_{n,\delta}
:=
\max\left\{
1,\,
L,\,
r_k(1/4)
\left(
1+\frac{\sigma_k^2}{\Delta_k}
\right),\,
\frac{L^2 \cdot \ell_{n,\delta}}{\Delta_k}
\right\},
\end{equation}
and define
\begin{align}
E_1
&:=
\frac{
L^2(1+M_{n,\delta}) \cdot \ell_\delta^2
}{
\Delta_k^2
},
\qquad 
E_2:=
\frac{
L\lambda_1M_{n,\delta} \cdot \ell_\delta
}{
\Delta_k^2
},
\qquad 
E_3:=
\frac{
L^2M_{n,\delta}^2+L^4
}{
\Delta_k^2
},
\nonumber\\
E_4
&:=
\frac{M_{n,\delta}}{\Delta_k^2},
\qquad
E_5
:=
M_{n,\delta}^3,
\label{eq:kpca-phase-two-error-terms}
\end{align}
as well as
\begin{equation}
\label{eq:kpca-phase-two-remainder}
\mc E_n
:=
\frac{E_1}{n^2}
+
\frac{E_2}{n^{1+\alpha}}
+
\frac{E_3}{n^{2\alpha}}
+
\frac{E_4}{n^{2-\alpha}}
+
\frac{E_5}{n^{3\alpha}}.
\end{equation}

\begin{thm}\label{thm:kpca-phase-two}
There exists a constant $C_\alpha>0$, depending only on $\alpha$,
such that, if
\begin{align}
n \ge C_\alpha \cdot  \max\Bigg\{ \ell_\delta,\left[ \frac{\ell_\delta}{\Delta_k}  
\left\{ k\log\left(\frac{ep}{k}\right) +\log(en) \right\} \right]^{\frac{1}{1-\alpha}},
\left( L\mu_1 \right)^{\frac{1}{2\alpha-1}},
\,
M_{n,\delta}^{1/\alpha} \Bigg\},
\label{eq:kpca-phase-two-sample-size}
\end{align}
then the final estimator $\widehat{\bm U}_{\rm avg}$ from Algorithm \ref{alg:adaptive-pr-kpca} satisfies, with
probability at least $1-\delta$,
\begin{align}
\left\|  \sin\bm\Theta
\lb  \widehat{\bm U}_{\rm avg}, \bm U_k^\star \rb \right\|_{\rm F}^2
\le
C_\alpha \cdot  \left[  \frac{\ell_\delta}{n} \cdot   \sum_{i=1}^k\sum_{j=k+1}^p
\frac{ \mb E \left[  \bm u_j^\top  (\bm A_1-\bm\Sigma) \bm u_i\right]^2
}{(\lambda_i-\lambda_j)^2}  + \mc E_n \right].
\label{eq:kpca-phase-two-final}
\end{align}
\end{thm}

In particular, we have 
\begin{align}
\left\|  \sin\bm\Theta \left( \widehat{\bm U}_{\rm avg}, \bm U_k^\star\right) \right\|_{\rm F}^2
\le C_\alpha \cdot  \left[ \frac{ \ell_\delta \cdot \sigma_k^2}{n \, \Delta_k^2}
+ \mc E_n \right]
\end{align}
with probability at least $1-\delta$.

\section{Application: Elliptical component analysis}

Consider the \emph{elliptical} model:
\begin{align} \label{elliptical model}
    \bm{X} \stackrel{d}= \bm{\mu} + R \cdot \bm{\Sigma}^{1/2} \cdot \bm{U}.
\end{align}
This model was introduced in \cite{han2018eca} as a generalization of classical elliptical distributions that accommodates heavy-tailed data. 
In \eqref{elliptical model}, $\bm{\mu} \in \mb{R}^{p}$ is the unknown mean vector, $R$ is a positive random variable that may be heavy-tailed, and $\bm{U} \sim \mathrm{Unif}(\mathbb{S}^{p-1})$ is independent of $R$. 
We are interested in estimating the principal eigenvector of $\bm{\Sigma}$ from a sample of size $n$ drawn from \eqref{elliptical model}. 
Note that we do not impose any moment conditions on $R$, and the methods we develop work without any assumptions on $R$ beyond having a continuous distribution. 
In particular, it may happen that $\mathbb{E}\big[ |R|^{\alpha} \big] = \infty$ for all $\alpha > 0$.

Most existing approaches to PCA under the model \eqref{elliptical model} are based on spatial signs; see, for example, \cite{taskinen2012robustifying,durre2017spatial,han2014scale,han2018eca,zhao2024spatial} and the references therein. 
High-dimensional results have been obtained in \cite{han2018eca} and \cite{zhao2024spatial}. 
All of these results are based on constructing robust estimators of the covariance matrix. 
However, storing the full sample covariance matrix is memory intensive, especially for high-dimensional data. 
In particular, the procedure in \cite{han2018eca} is based on a U-statistic estimator and requires $O(n^2p^2)$ operations, which was later improved to $O(np^2)$ in \cite{zhao2024spatial}. 
More precisely, the estimators in \cite{han2018eca,zhao2024spatial} take the forms
\begin{align} \label{Han-Liu pca}
\hat{\bm{v}}_{\rm HL} 
:= \text{leading eigenvector of} \left[ \frac{2}{n(n-1)} \sum_{1\leq i<j \leq n} 
\frac{\big( \bm{X}_i - \bm{X}_j \big)\big( \bm{X}_i - \bm{X}_j \big)^\top}{\| \bm{X}_i - \bm{X}_j \|^2}  \right],
\end{align}
and 
\begin{align} \label{zwf pca}
\hat{\bm{v}}_{\rm ZWF}  
:= \text{leading eigenvector of} \left[  \sum_{i=1}^n  
\frac{\big( \bm{X}_i - \hat{\bm{\mu}}_i \big)\big( \bm{X}_i - \hat{\bm{\mu}}_i \big)^\top}
{ \| \bm{X}_i - \hat{\bm{\mu}}_i \|^2 } \right],
\end{align}
where $\hat{\bm{\mu}}_i$ is a robust mean estimator.

It is therefore desirable to have memory-efficient algorithms for performing PCA under the model \eqref{elliptical model}. 
To the best of our knowledge, this question has not been addressed before. 
The nonsparse estimators in \eqref{Han-Liu pca} and \eqref{zwf pca} require $O(n^2p^2)$ and $O(np^2)$ operations, respectively, and both rely on storing the full sample covariance matrix, which already costs $O(p^2)$ memory.

Let us describe how averaging theory applies in this example. We first apply a transformation that removes both the unknown mean $\bm\mu$ and the heavy-tailed radial component $R$, and then run Algorithm~\ref{alg:adaptive-pr-kpca} on the transformed stream. Note that the transformation can be computed using two data points at a time, so no sample covariance matrix is ever formed. Although the transformation changes the spectrum of the target and the eigengap, the orthonormal basis remains unchanged. 

To be more specific, suppose we observe $\bm X_1,\ldots,\bm X_{2n}$ i.i.d.\ from
\eqref{elliptical model} and define the \emph{spatial signs}
\begin{align} \label{eq:eca-preprocessing}
\bm{\mc X}_t
:=
\frac{\bm X_{2t}-\bm X_{2t-1}}{\left\| \bm X_{2t}-\bm X_{2t-1} \right\|},
\qquad t=1,\ldots,n .
\end{align}

\renewcommand{\thealgorithm}{A}
\begin{algorithm}[H]
\caption{Memory-efficient ECA}
\label{alg:eca}
\begin{algorithmic}[1]
\Require Data $\bm X_1,\ldots,\bm X_{2n}\in\mb R^p$ from
\eqref{elliptical model}; target rank $k$; exponent $\alpha\in(1/2,1)$;
confidence level $\delta\in(0,1/8)$.

\For{$t=1,\ldots,n$}
    \State Receive $\bm X_{2t-1},\bm X_{2t}$ and form $\bm{\mc X}_t$ as in
    \eqref{eq:eca-preprocessing}; discard $\bm X_{2t-1},\bm X_{2t}$.
\EndFor

\State Run Algorithm~\ref{alg:adaptive-pr-kpca} on the stream
$\bm{\mc X}_1,\ldots,\bm{\mc X}_n$ with parameters $k,\alpha,\delta$.

\State \Return its output $\widehat{\bm U}_{\rm avg}$.
\end{algorithmic}
\end{algorithm}
\addtocounter{algorithm}{-1}
\renewcommand{\thealgorithm}{\arabic{algorithm}}

 We label the algorithm above by a letter to emphasise that it is not a new algorithm, but rather
Algorithm~\ref{alg:adaptive-pr-kpca} applied to a preprocessed stream.
To state the guarantee, let 
$$\bm\Sigma = \sum_{i=1}^p\lambda_i\bm u_i\bm u_i^\top$$
be the spectral decomposition of the scatter matrix in \eqref{elliptical model},
with $\lambda_1\ge\cdots\ge\lambda_p\ge0$, and let
$\bm U_k^\star=[\bm u_1,\ldots,\bm u_k]$. 

Write $\bm U\sim\mathrm{Unif}(\mb
S^{p-1})$, $U_i:=\bm u_i^\top\bm U$,
and define
\begin{align} \label{eq:eca-transformed-spectrum}
D:=\sum_{\ell=1}^p\lambda_\ell U_\ell^2, \qquad 
\tilde\lambda_i
:=
\mb E\left[ \frac{\lambda_i U_i^2}{D} \right],
\qquad
\widetilde\Delta_k := \tilde\lambda_k-\tilde\lambda_{k+1},
\qquad
\tilde\sigma_{ij}^2
:=
\mb E\left[ \frac{\lambda_i\lambda_j U_i^2U_j^2}{D^2} \right] .
\end{align}
As we show in Section~\ref{sec:eca-proof}, $\widetilde\Delta_k>0$ whenever
$\Delta_k:=\lambda_k-\lambda_{k+1}>0$. Finally, put
\begin{align} \label{eq:eca-Mn}
\widetilde M_{n,\delta}
:=
r_k(1/4)\lb 1+\widetilde\Delta_k^{-1} \rb \ell_{n,\delta},
\qquad
\widetilde{\mc E}_n
:=
\widetilde M_{n,\delta}^3
\left[
\frac{\ell_\delta^2}{\widetilde\Delta_k^2 \, n^{2\alpha}}
+
\frac{1}{n^{2-\alpha}}
+
\frac{1}{n^{3\alpha}}
\right] .
\end{align}

\begin{thm} \label{thm:eca}
Suppose \eqref{elliptical model} holds with $R$ continuous, $\bm\Sigma\ne\bm0$,
and $\Delta_k>0$. Fix $\alpha\in(1/2,1)$ and $\delta\in(0,1/8)$. There is a
constant $C_\alpha>0$, depending only on $\alpha$, such that if
\begin{align} \label{eq:eca-sample-size}
n
\ge
C_\alpha \cdot \max\left\{
\left[
\frac{\ell_\delta}{\widetilde\Delta_k}
\left\{ k\log\lb \frac{ep}{k} \rb + \log(en) \right\}
\right]^{\frac1{1-\alpha}},
\;
\widetilde M_{n,\delta}^{1/\alpha}
\right\},
\end{align}
then the output $\widehat{\bm U}_{\rm avg}$ of Algorithm~\ref{alg:eca}
satisfies, with probability at least $1-\delta$,
\begin{align} \label{eq:eca-final}
\left\|
\sin\bm\Theta \lb \widehat{\bm U}_{\rm avg}, \bm U_k^\star \rb
\right\|_{\rm F}^2
\le
C_\alpha
\left[
\frac{\ell_\delta}{n} \cdot 
\sum_{i=1}^k\sum_{j=k+1}^p
\frac{\tilde\sigma_{ij}^2}{\lb \tilde\lambda_i-\tilde\lambda_j \rb^2}
+
\widetilde{\mc E}_n
\right] .
\end{align}
\end{thm}

Theorem \ref{thm:eca} gives a finite sample bound \eqref{eq:eca-final} for the averaging estimator from Algorithm \ref{alg:eca}. Note that the term $\widetilde{\mc E}_n$ in \eqref{eq:eca-final} is of order $o \lb 1/n \rb$. A few remarks are in order. 

\begin{remark} 
The bound \eqref{eq:eca-final} is a finite-sample bound that holds with no assumption on $R$ beyond continuity; in particular
it applies when $\mb E[|R|^\gamma]=\infty$, for every $\gamma>0$. This property is inherited from the sign transformation, which makes the transformed data bounded, so that the boundedness assumption of Section~\ref{sec:k-PCA-convergence} holds with universal constants. This is also the same principle as in \cite{han2014scale,han2018eca}. 
\end{remark}

\begin{remark}
The eigengap in \eqref{eq:eca-transformed-spectrum} is approximately the same as the eigengap of $\bm\Sigma$ itself. Indeed, in high dimensions, we expect $D$ to  concentrate around
$\operatorname{tr}(\bm\Sigma)/p$, so that
\[
\tilde\lambda_i \approx \frac{\lambda_i}{\operatorname{tr}(\bm\Sigma)},
\qquad
\tilde\sigma_{ij}^2 \approx \frac{\lambda_i\lambda_j}{\operatorname{tr}(\bm\Sigma)^2},
\]
using $\mb E[U_i^2U_j^2]=\{p(p+2)\}^{-1}$ for $i\ne j$. The factors of
$\operatorname{tr}(\bm\Sigma)$ cancel, and the leading term becomes
\[
\frac{\ell_\delta}{n}
\sum_{i=1}^k\sum_{j=k+1}^p
\frac{\lambda_i\lambda_j}{\lb \lambda_i-\lambda_j \rb^2},
\]
which is exactly the asymptotic variance of the maximum likelihood estimator in the Gaussian model \cite{anderson1963asymptotic}. Thus nothing is lost by discarding the radial component, which is the same phenomenon observed in \cite{han2018eca}. 

\end{remark}

\begin{remark}
Let us discuss the total cost and memory storage of Algorithm \ref{alg:eca}. Below are the costs of the steps in the algorithm:
\begin{itemize}
    \item Each iteration forms one sign vector in $O(p)$ operations.
    \item Each rank-one update of a $p\times k$ frame followed by a QR factorization requires $O(pk^2)$ operations.
    \item Performing  one averaging step whose dominant cost is the SVD of a $k\times k$ matrix. This results in $O(k^3)$ operations.
\end{itemize}
Thus we use in total $O(npk^2+nk^3)$ operations and $O\lb pk\log(1/\delta) \rb$ memory units.  For $k=1$ this is $O(np)$ time and
$O\lb p\log(1/\delta) \rb$ memory, against $O(n^2p^2)$ time and $O(p^2)$ memory for \eqref{Han-Liu pca}, and $O(np^2)$ time and $O(p^2)$ memory for \eqref{zwf pca}. 
\end{remark}

\section{Simulation studies} \label{sec:simulation}

We compare Algorithm~\ref{alg:eca} with the two existing ECA procedures: the multivariate Kendall's tau estimator
\eqref{Han-Liu pca} of \cite{han2018eca} and the spatial-sign
estimator \eqref{zwf pca} of \cite{zhao2024spatial}.

\subsection{Design}

Data are generated from \eqref{elliptical model} with $\bm\mu = 5 \cdot \bm 1_p$ and $R \sim F(p,1)$, so that
$\mb E \lft |R|^\gamma \rht = \infty$ for every $\gamma \ge 1/2$. In particular,
$R$ does not have a finite mean. Here $F(p,1)$ is the standard $F$ distribution with $p$ numerator degrees of freedom and one denominator degree of freedom.

We use four spectra, chosen to span a range of transformed
eigengaps $\widetilde\Delta_k$ and to include one $k>1$ subspace problem; the
constants $\tilde\lambda_i$ and
$V_k := \sum_{i\le k}\sum_{j>k} \tilde\sigma_{ij}^2 /
(\tilde\lambda_i-\tilde\lambda_j)^2$ are computed from
\eqref{eq:eca-transformed-spectrum} by Monte Carlo.

\begin{center}
\begin{tabular}{llccc}
\hline
& $p$, $\lambda$ & $k$ & $\widetilde\Delta_k$ & $V_k$ \\
\hline
S1 & $p=30$, $\lambda=(40,10,4,1,\dots,1)$   & 1 & 0.221 & 1.741 \\
S2 & $p=30$, $\lambda=(100,5,1,\dots,1)$     & 1 & 0.462 & 0.601 \\
S3 & $p=20$, $\lambda=(20,4,1,\dots,1)$      & 1 & 0.236 & 1.763 \\
S4 & $p=30$, $\lambda=(40,35,30,1,\dots,1)$  & 3 & 0.201 & 3.447 \\
\hline
\end{tabular}
\end{center}


The sample sizes are set to be 
$n \in \la 800,1600,3200,6400 \ra$. We report the mean of
$\| \sin\bm\Theta ( \widehat{\bm U}, \bm U_k^\star ) \|_{\rm F}^2$ over $30$
replications ($20$ at $n=6400$). Algorithm~\ref{alg:eca} is run with
$\alpha=2/3$.
We will consider two versions of Algorithm \ref{alg:eca}: the analyzed version, which uses disjoint pairs
\eqref{eq:eca-preprocessing} and $m = \lfloor n/2 \rfloor$, and a version with overlapping pairs
$\bm{\mc X}_t = ( \bm X_{t+1}-\bm X_t )/\| \bm X_{t+1}-\bm X_t \|$. Note that the second variant is essentially a $1$-dependent variant of the algorithm, and it seems to work better empirically. 

\subsection{Performance}

\begin{figure}[t]
\centering
\includegraphics[width=\textwidth]{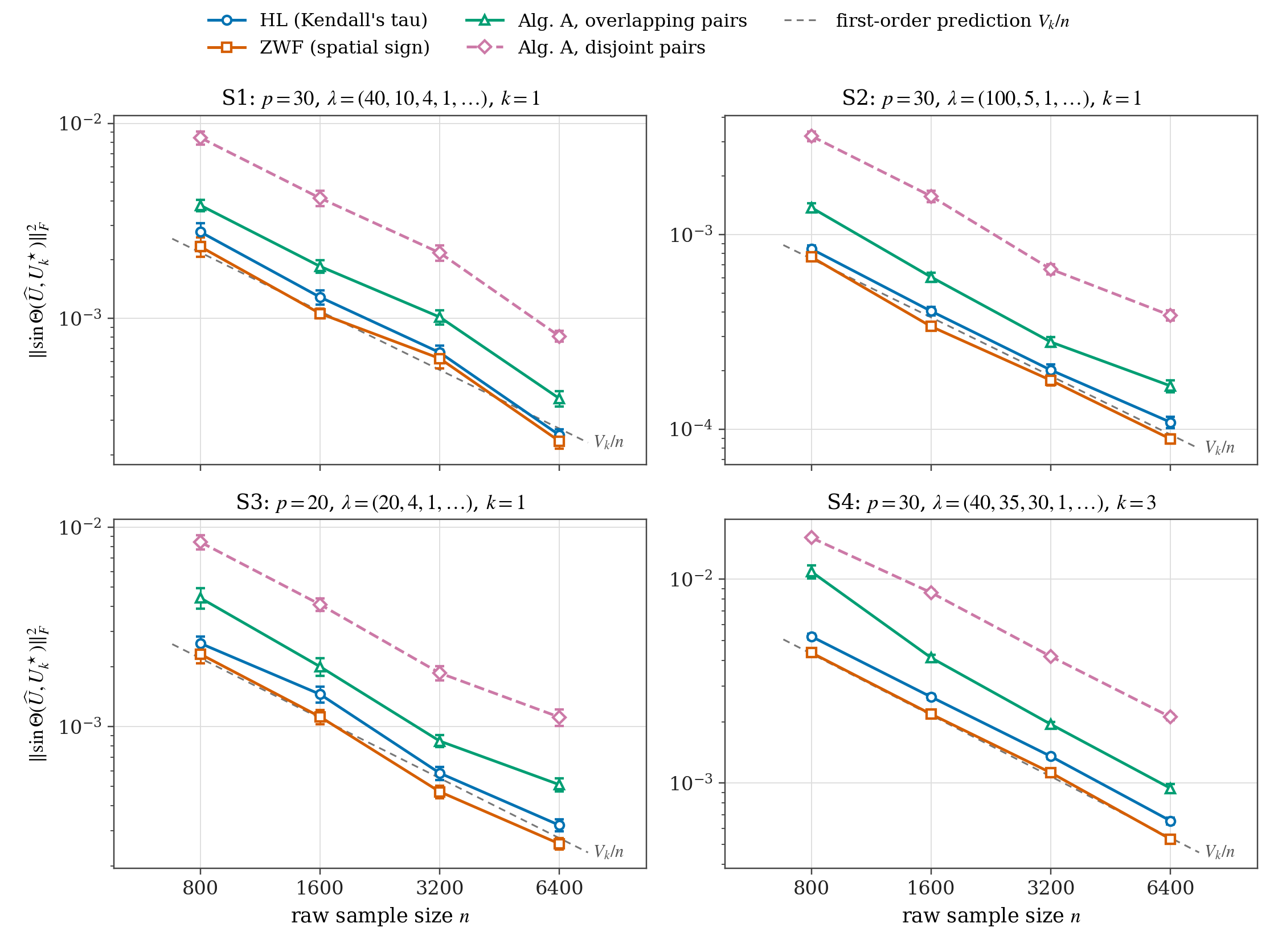}
\caption{Mean squared sine error against raw sample size, settings S1--S4.
Error bars are one standard error; the thin line is the first-order prediction
$V_k/n$ from \eqref{eq:eca-final}. Both axes logarithmic.}
\label{fig:eca-comparison}
\end{figure}

Figure~\ref{fig:eca-comparison} reports the mean sine-squared errors. For $n\ge1600$ every curve
is parallel to $V_k/n$: Algorithm~\ref{alg:eca} attains the $n^{-1}$ rate with
no knowledge of $\widetilde\Delta_k$, and the separation from HL and ZWF is within a universal constant.
The disjoint-pair version is consistently the weaker of the two. This is as expected and is because of the fact that using disjoint pairs discards a fraction of the data. Using overlapping pairs performs better, at the cost of dealing with $1$-dependent data.


\subsection{Higher dimensions}

\begin{figure}[t]
\centering
\includegraphics[width=\textwidth]{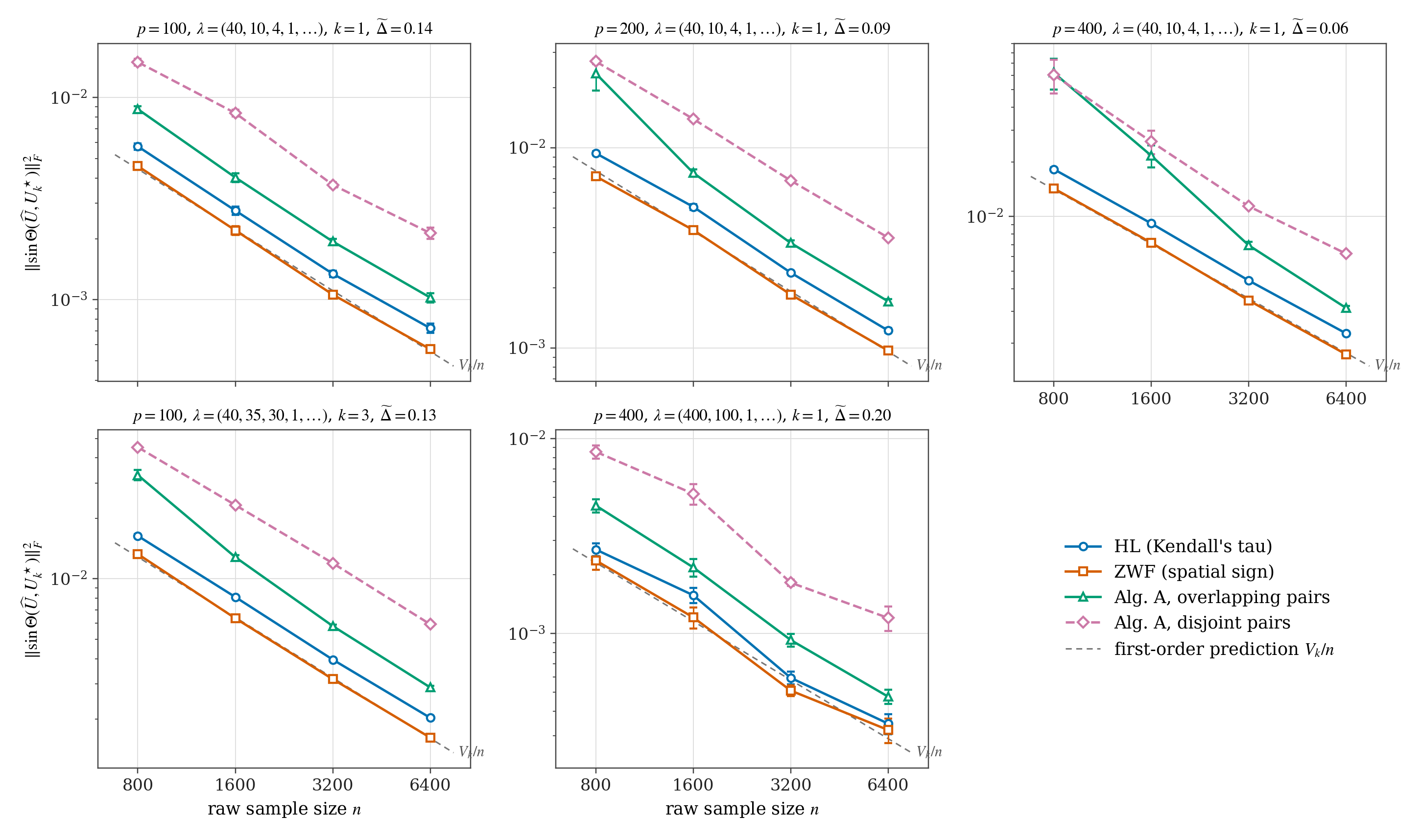}
\caption{Mean squared sine error for $p \in \la 100,200,400 \ra$ with the S1
spectrum, together with $k=3$ at $p=100$ and a strongly spiked spectrum at
$p=400$.}
\label{fig:eca-highp}
\end{figure}

Figure~\ref{fig:eca-highp} repeats the comparison at $p \in \la 100,200,400
\ra$, adding $k=3$ at $p=100$ and a strongly spiked spectrum
$\lambda = (400,100,1,\dots,1)$ at $p=400$. The results agree well with our theory that the error depends only on the effective variance term $V_k$ and not the ambient dimensions. Algorithm \ref{alg:eca} with overlapping pairs performs better than with disjoint pairs, up to a constant factor in both this case and the previous low-dimensional case $p=30$.

\subsection{Memory and computation}

\begin{figure}[t]
\centering
\includegraphics[width=0.62\textwidth]{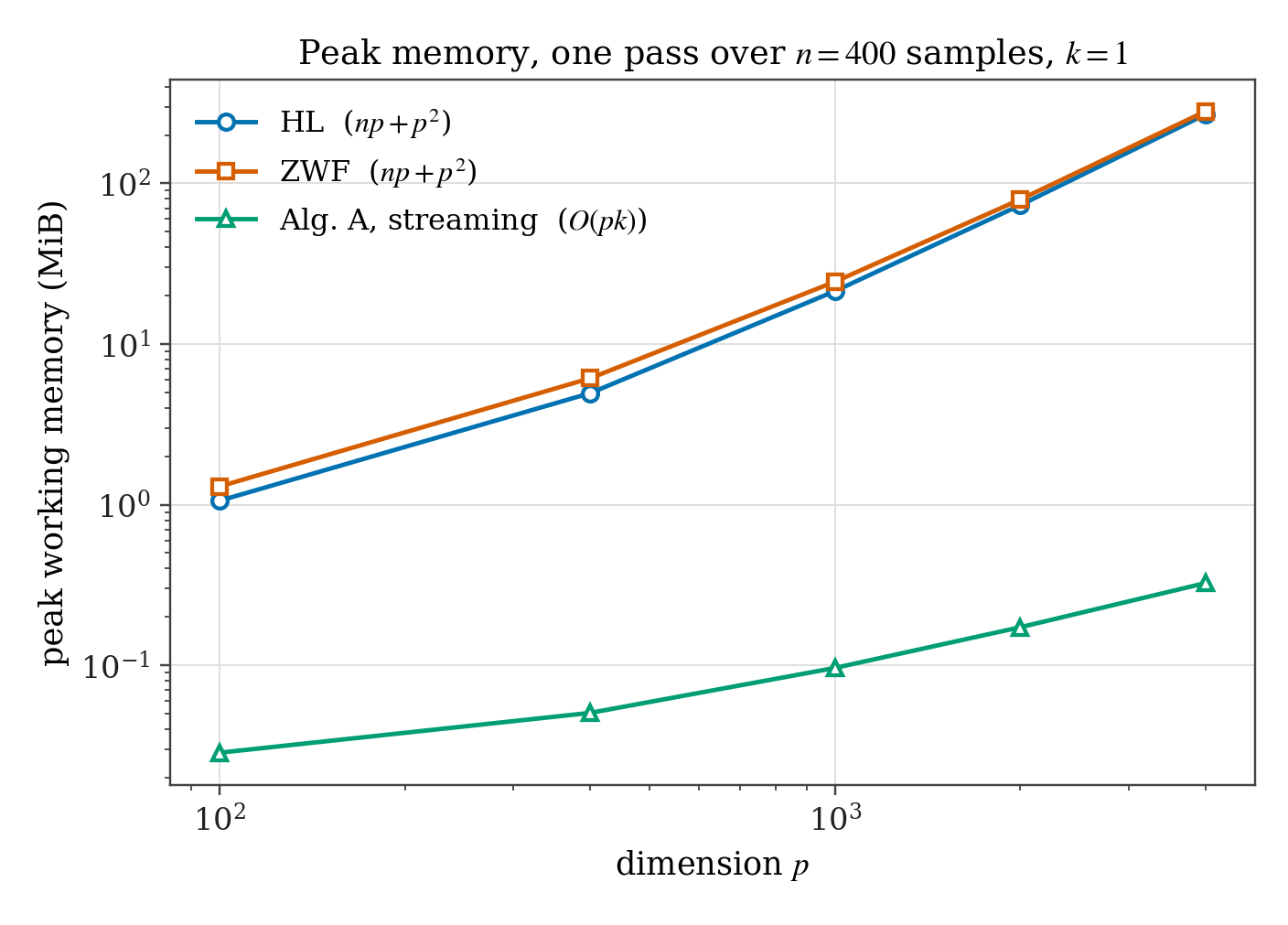}
\caption{Peak working memory measured by the Python module  \texttt{tracemalloc}, one pass over
$n=400$ observations with $k=1$, all three methods fed by the same
one-observation-at-a-time stream.}
\label{fig:eca-memory}
\end{figure}

Figure~\ref{fig:eca-memory} reports peak working memory when all methods
process the observations in an online setting. The memory used by
Algorithm~\ref{alg:eca} grows linearly in $p$, from $0.03$ MB at $p=100$ to
$0.34$ MB at $p=4000$, whereas the estimators \eqref{Han-Liu pca} and
\eqref{zwf pca} of \cite{han2018eca,zhao2024spatial} grow quadratically, from
$1.2$ and $1.4$ MB to $282$ and $295$ MB. At $p=4000$ this is a factor of
$827$ over \cite{han2018eca} and $864$ over \cite{zhao2024spatial}, and the ratio doubles with each doubling of
$p$. Moving from $k=1$ to $k=3$ at $p=1000$ raises the memory usage of
Algorithm~\ref{alg:eca} from $0.10$ MB to $0.20$ MB and leaves the competitors
unchanged, as expected from the $O(pk)$ versus $O(p^2)$ scaling.

The computational cost shows the same separation. In the same benchmark, one
pass at $p=4000$ takes $0.12$ s for Algorithm~\ref{alg:eca}, against $7.9$ s
for \cite{zhao2024spatial} and $40.1$ s for \cite{han2018eca}, the last reflecting the $O(n^2p^2)$ cost of the
$U$-statistic. The gap widens with $n$, since the cost of
Algorithm~\ref{alg:eca} is linear in $n$ while that of \eqref{Han-Liu pca} is
quadratic: under setting S1 with $n = 4\times10^5$, Algorithm~\ref{alg:eca}
completes in $8.3$ s, whereas \cite{han2018eca}
gives an estimate of roughly $2.1$ hours.

\section{Real data analysis} \label{sec:realdata}



We analyze the 10x Genomics E18 mouse brain dataset \cite{castiglione2026sgdgmf,hicks2021mbkmeans,townes2019featureselection,kidzinski2022gmf}, which contains
$1{,}306{,}127$ cells over $27{,}998$ genes. The
data are distributed publicly under CC BY 4.0 and are, to our knowledge, among
the largest single-cell RNA sequencing experiments in general use. Each cell in this dataset  is
a $27{,}998$-dimensional vector of counts.
Principal component analysis is the standard first step in the single-cell
analysis pipeline, where the leading components are used for clustering and
visualisation. 

\emph{\underline{Splitting and preprocessing}.} The cells are partitioned into
a training set of $1{,}240{,}686$ cells and a test set of $65{,}441$ cells.  Writing $\bm d$ for the difference of two independently drawn
preprocessed cells, all methods evaluated below estimate the $k$-leading eigenspace of
\[
\bm K
:=
\mb E\left[ \frac{\bm d \bm d^\top}{\|\bm d\|^2} \right],
\]
which is exactly the population multivariate Kendall's tau of
\cite{han2018eca}. We will also evaluate a variant of our Algorithm \ref{alg:eca}, which uses overlapping pairs for successive updates. This is equivalent to having $1$-dependent data. The
main algorithm fit uses $620{,}343$ disjoint training pairs, while this overlapping variant uses $1{,}240{,}685$  $1$-dependent updates.

\emph{\underline{Evaluation}.} There is no ground-truth eigenspace on real data, so we
therefore report the held-out \emph{pair-sign similarity}
\[
\mc P \lb \widehat{\bm U} \rb
:=
\mb E\left[
\frac{\|\widehat{\bm U}^\top \bm d\|^2}{\|\bm d\|^2}
\right],
\]
estimated over $32{,}720$ disjoint held-out pairs, which is maximised at the
population top-$k$ eigenspace and is therefore the natural population objective.
Larger values are better, and we fix $k=5$ throughout.

\emph{ \underline{Data processing}.} All methods process data in a sequential fashion.
On the full training set of $1{,}240{,}686$ cells, the exact estimator in \cite{han2018eca} is not computable:
$\binom{n}{2} \approx 7.7 \times 10^{11}$ pairs, and the estimated processing time is $2\times 10^6$ seconds, which is more than three weeks. For the estimator in \cite{zhao2024spatial}, we will use the power method to compute the principal eigenvector, rather than storing the sample covariance matrices. This reduces the memory storage required from $18$GB to less than $400$MB, at the cost of doing many passes through the whole dataset.

\begin{table}[h]
\centering
\begin{tabular}{lrr}
\hline
method & processing time (s) & $\mc P$ \\
\hline
Algorithm~\ref{alg:eca}                          & $191.9$ & $0.07856$ \\
Algorithm~\ref{alg:eca}, overlapping pairs       & $377.9$ & $0.07862$ \\
\cite{zhao2024spatial}                                & $1{,}013.6$ & $0.07879$ \\
\cite{han2018eca}                                   & \multicolumn{2}{c}{\centering{not computable}} \\
\hline
\end{tabular}
\caption{Full training set, $k=5$. Higher $\mc P$ is better}
\label{tab:neuron-full}
\end{table}

The pair-sign projection energy and processing time are reported in Table  \ref{tab:neuron-full}. Algorithm~\ref{alg:eca}  takes only  $192$ seconds to compute and uses  a single pass over the data. Its overlapping pair variant takes longer, which is roughly $377$ seconds, and the method based on spatial signs in \cite{zhao2024spatial} requires approximately $2.7$ times as long as the overlapping-pair variant. In terms of the pair-sign projection energy, we can see that all methods are quite close in performance.

\section{Equivalence between $k$-PCA and $1$-PCA}

 The goal of this section is to introduce the exterior product space (also known as Pl\"ucker embedding) and some related elementary facts. For a full textbook treatment of multilinear algebra, we refer to Chapter 14 of the monograph \cite{roman2005advanced}. Suppose $V$ is a vector space of dimension $p$ and  let $\la \bm{e}_1,\dots,\bm{e}_p \ra$ be its orthonormal basis. We define the $k$-th exterior power of $V$, denoted by $\bigwedge^k V$ as
\[
\mbox{span} \la  \bm{e}_{i_1} \wedge \dots \wedge \bm{e}_{i_k} : 1 \leq i_1<i_2<\dots<i_k \leq p \ra.
\]
In particular, the map
\[
\lb \bm{v}_1,\dots,\bm{v}_k \rb \mapsto \bm{v}_1 \wedge \bm{v}_2 \wedge\dots \wedge \bm{v}_k 
\]
is a linear mapping in each of its arguments (holding the others fixed) with the properties that 
\begin{align*}
\bm{v}_1 \wedge \dots  \wedge \bm{v}_i \dots \wedge \bm{v}_j \dots \wedge  \bm{v}_k &= - \bm{v}_1 \wedge \dots  \wedge \bm{v}_j \dots \wedge \bm{v}_i \dots \wedge  \bm{v}_k \\
\bm{v}_1 \wedge \dots  \wedge \bm{v}_i \dots \wedge \bm{v}_i \dots \wedge  \bm{v}_k &=0
\end{align*}
There are two types of elements in $\bigwedge^k V$: {\it decomposable} elements and {\it non-decomposable} elements. An element is called decomposable if it can be represented in the form $\bm x_1 \wedge \dots\wedge \bm x_k$ for some $\la \bm x_1,\dots,\bm x_k \ra \subset V$. Otherwise, we call it non-decomposable. In general,  $\bigwedge^k V$ contains non-decomposable elements. A simple example demonstrating this fact can be constructed by looking at $\bigwedge^2 \mb R^4$: one can check that the vector $\bm x =  \bm e_1 \wedge \bm e_2 + \bm e_3 \wedge \bm e_4$ cannot be represented in the form  $\bm u \wedge \bm v$ for some $\bm u, \bm v \in \mb{R}^4$.

\subsection{Metric and norm on the exterior space}
With the orthonormal basis $\la \bm e_1, \dots, \bm e_p \ra$ of $\mb R^p$, for the sake of presentation, we denote the basis of $\bigwedge^k \mb R^p$ as
\[
\la \bm e_I := \bm e_{i_1} \wedge \dots \wedge \bm e_{i_k}; \, I= \lb i_1, i_2,\dots, i_k \rb \ra_{1\leq i_1<\dots<i_k\leq p}.
\]
Let $\mc I_k$ be the collection of such $I$ above. Thus every element $\bm \xi \in \bigwedge^k \mb R^p$ can be rewritten as 
\[
\bm \xi = \sum_{I \in \mc I_k} \xi_I \cdot \bm e_I.
\]
We then define the exterior norm
\[
\| \bm \xi \|_{\wedge}:= \sqrt{\sum_{I \in \mc I_k} \xi_I^2 }.
\]
Since $\bigwedge^k \mb R^p$ is a finite-dimensional space, the norm above is complete and is also induced by the inner product 
\[
\langle \bm \xi, \bm \zeta \rangle_\wedge := \sum_{I \in \mc I_k} \xi_I  \zeta_I. 
\]
The definition above is not useful for computational purposes. We will make use of the following at various points in the later proofs.
\begin{prop} \label{prop:inner-product}
    For decomposable elements $\bm u, \bm v \in \bigwedge^k \mb R^p$, we have
    \[
    \langle \bm u, \bm v \rangle_{\wedge} = \mbox{det} \lb \bm U^\top \bm V \rb
    \]
    where $\bm U=[\bm u_1,\ldots,\bm u_k]$ and $\bm V=[\bm v_1,\ldots,\bm v_k]$, with $\bm u=\bm u_1\wedge\cdots\wedge\bm u_k$ and $\bm v=\bm v_1\wedge\cdots\wedge\bm v_k$.
\end{prop}

\noindent \textbf{Proof of Proposition \ref{prop:inner-product}.}
With $\bm u= \bm u_1 \wedge \bm u_2\wedge \dots \wedge \bm u_k$ and  $\bm v= \bm v_1 \wedge \bm v_2\wedge \dots \wedge \bm v_k$, write 
\[
\bm U:= \lb \bm u_1,\dots, \bm u_k \rb, \qquad \bm V:= \lb \bm v_1,\dots, \bm v_k \rb.
\]
Observe that
\begin{align*}
\bm u &= \lb \sum_{j_1=1}^p u_{1,j_1} \cdot  \bm e_{j_1} \rb \wedge \dots \wedge \lb  \sum_{j_k=1}^p u_{k,j_k} \cdot  \bm e_{j_k} \rb \\
&= \sum_{1\leq j_1,\dots, j_k \leq p} u_{1,j_1} \, u_{2,j_2}\dots u_{k,j_k} \cdot \bm e_{j_1} \wedge \dots \wedge \bm e_{j_k} = \sum_{I \in \mc I_k} \mbox{det} \lb \bm U_I^\top \rb \, \bm e_I. 
\end{align*}
Here $\bm U_I$ and $\bm V_I$ are the $k\times k$ submatrices of $\bm U$ and $\bm V$ consisting of the rows indexed by $I$.
Similarly, $\bm v =  \sum_{I \in \mc I_k} \mbox{det} \lb \bm V_I \rb \, \bm e_I$. Thus the Cauchy--Binet formula yields 
\[
\langle \bm u, \bm v \rangle_{\wedge} = \sum_{I \in \mc I_k}  \mbox{det} \lb \bm U_I^\top \rb \cdot  \mbox{det} \lb \bm V_I \rb = \det \lb \bm U^\top \bm V \rb.
\]
$\hfill$ $\square$

\subsection{Oja's algorithm on the exterior space}
From the exposition above, it is easy to see that the attractive feature of Pl\"ucker embedding is that it embeds a subspace of dimension $k$ in $V$ into $\bigwedge^k V$ so that the embedded image becomes a single vector. This transformation creates a one-to-one correspondence between the $k$-PCA in $\mb R^p$  with $k>1$ and $1$-PCA in the exterior product space $\bigwedge^k \mb R^p$, which we will make precise below. Of course, a clear advantage of this approach is that $1$-PCA is much easier to analyze than the case $k \geq 2$. 

To formalize this viewpoint, recall that
$\mc U_*:=\operatorname{span}(\bm u_1,\dots,\bm u_k)$
is the target top-$k$ principal subspace and define its exterior embedding by
\[
\bm \omega_*:=\bm u_1\wedge\cdots\wedge \bm u_k
\in \bigwedge^k \mb R^p.
\]
Let $\hat{\bm U}\in\mb R^{p\times k}$ be an estimator with orthonormal columns $\hat{\bm U}=[\hat{\bm u}_1,\dots,\hat{\bm u}_k]$ 
and put
\[
\hat{\mc U}:=\operatorname{span}(\hat{\bm u}_1,\dots,\hat{\bm u}_k).
\]
Define the associated exterior embedding
\[
\hat{\bm \omega}:=
\hat{\bm u}_1\wedge\cdots\wedge \hat{\bm u}_k
\in \bigwedge^k \mb R^p.
\]
By Proposition \ref{prop:inner-product}, one has
\begin{align}
\sin^2(\hat{\bm \omega},\bm \omega_*)
&:=
1-\langle \hat{\bm \omega},\bm \omega_*\rangle^2=
1-\det(\hat{\bm U}^{\top}\bm U_*)^2 =
1-\prod_{\ell=1}^k\cos^2\theta_\ell,
\label{eq:exterior-loss-principal-angles}
\end{align}
where $\bm U_*=[\bm u_1,\dots,\bm u_k]$ and $\theta_1,\dots,\theta_k$ are the principal angles between
$\hat{\mc U}$ and $\mc U_*$.

\begin{prop} \label{prop:equivalance-error}
    With the notation above, we have
    \begin{align}
        \|\sin \Theta(\hat{\mc U},\mc U_*)\|_{\rm op}^2
&\leq
\sin^2(\hat{\bm \omega},\bm \omega_*),
\label{eq:operator-loss-transfer}
    \end{align}
and 
\begin{align} \label{eq:frobenius-loss-transfer}
    \sin^2 \lb  \hat{\bm \omega},\bm \omega_* \rb
&\leq
\|\sin\Theta(\hat{\mc U},\mc U_*)\|_{\rm F}^2
\leq
\tan^2 \lb \hat{\bm \omega},\bm \omega_*  \rb
\end{align}    
where 
\[
\tan^2 \lb \hat{\bm \omega},\bm \omega_*  \rb := \frac{\sin^2 \lb \hat{\bm \omega},\bm \omega_*  \rb}{\cos^2 \lb \hat{\bm \omega},\bm \omega_*  \rb}.
\]
\end{prop}

\noindent \textbf{Proof of Proposition \ref{prop:equivalance-error}.}
To show \eqref{eq:operator-loss-transfer}, write
\begin{align*}
  \|\sin \Theta(\hat{\mc U},\mc U_*)\|_{\rm op}^2 = \max_{1\leq i \leq k} \,  \sin^2 \theta_i  \leq 1-\prod_{i=1}^k\cos^2\theta_i \leq \sin^2(\hat{\bm \omega},\bm \omega_*).
\end{align*}
To show \eqref{eq:frobenius-loss-transfer}, note that the first inequality is trivial, and for the second one, observe that
\begin{align*}
    \|\sin\Theta(\hat{\mc U},\mc U_*)\|_{\rm F}^2 = \sum_{i=1}^k \sin^2 \theta_i &\leq -\sum_{i=1}^k \log \lb 1 - \sin^2 \theta_i \rb \\
    & = - \sum_{i=1}^k \log \lb \cos^2 \theta_i \rb = \log \left[ 1 + \tan^2 \lb \hat{\bm \omega},\bm \omega_*  \rb  \right] \leq \tan^2 \lb \hat{\bm \omega},\bm \omega_*  \rb.
\end{align*}
The proof is completed.
$\hfill$ $\square$

The two-sided bound \eqref{eq:frobenius-loss-transfer} in Proposition \ref{prop:equivalance-error} states that the sine-squared error in the exterior space is of the same order as the Frobenius norm of the sine matrix in the ambient space. In particular, an algorithm achieves an optimal convergence rate for the k-PCA problem if and only if its embedded version in the exterior space achieves the same rate. This gives an equivalence, one-to-one correspondence between the $k$-PCA in $\mb R^p$ and the $1$-PCA in $\bigwedge^k \mb R^p$. 

Now, put
\[
\widetilde{\mb R}_{k,p}:=\bigwedge^k\mb R^p,
\]
and, consistently with Theorem~\ref{k-PCA convergence}, define
\begin{equation}
    \bm B_t
    :=
    \bm I_p+\eta_t\bm A_t
    =
    \bm I_p+\eta_t\bm X_t\bm X_t^\top.
    \label{eq:oja-update-matrix}
\end{equation}
Write
\[
    \bm Q_t
    =
    [\bm q_{1,t},\ldots,\bm q_{k,t}]
\]
for the Oja iterate from Algorithm~\ref{alg:oja-k-pca}. Equivalently,
\begin{equation}
    \bm Q_t
    =
    \operatorname{orth}(\bm B_t\bm Q_{t-1}),
    \label{eq:block-oja-Bt}
\end{equation}
where \(\operatorname{orth}\) denotes the \(Q\)-factor in the QR
factorization with positive diagonal.

For a fixed matrix \(\bm M\in\mb R^{p\times p}\), define the linear map
\(\mathcal L_{\bm M}\) on \(\widetilde{\mb R}_{k,p}\) by
\begin{align} \label{L}
\mathcal L_{\bm M}
\left(
\bm w_1\wedge\cdots\wedge \bm w_k
\right)
:=
\sum_{j=1}^k
\bm w_1\wedge\cdots\wedge \bm w_{j-1}
\wedge \bm M\bm w_j
\wedge \bm w_{j+1}\wedge\cdots\wedge \bm w_k.
\end{align}
Our next result describes the Oja iteration in the exterior space.

\begin{thm}\label{thm: Plucker embedding}
Define the Pl\"ucker image of \(\bm Q_t\) by
\begin{equation}
    \bm\omega_t
    :=
    \bm q_{1,t}\wedge\cdots\wedge\bm q_{k,t}
    \in\widetilde{\mb R}_{k,p}.
    \label{eq:plucker-iterate}
\end{equation}
Then
\begin{equation}
    \bm\omega_t
    =
    \frac{
        \left(\bigwedge^k\bm B_t\right)\bm\omega_{t-1}
    }{
        \left\|
            \left(\bigwedge^k\bm B_t\right)\bm\omega_{t-1}
        \right\|_{\wedge}
    }.
    \label{embbeded Oja}
\end{equation}
Moreover, since \(\bm A_t=\bm X_t\bm X_t^\top\) has rank at most one,
\begin{equation}
    \bigwedge^k\bm B_t
    =
    \bm I_{\widetilde{\mb R}_{k,p}}
    +
    \eta_t\mathcal L_{\bm A_t}.
    \label{embedded Oja update}
\end{equation}
\end{thm}

\noindent \textbf{Proof of Theorem \ref{thm: Plucker embedding}.}
We first recall the definition of the induced exterior map. For any linear
map \(\bm T:\mb R^p\to\mb R^p\), the induced linear map
\[
    \bigwedge^k\bm T:
    \widetilde{\mb R}_{k,p} \to \widetilde{\mb R}_{k,p}
\]
is defined on decomposable vectors by
\[
    \left(\bigwedge^k\bm T\right)
    (\bm w_1\wedge\cdots\wedge\bm w_k)
    =
    (\bm T\bm w_1)\wedge\cdots\wedge(\bm T\bm w_k),
\]
and then extended linearly to all of \(\bigwedge^k\mb R^p\).

\medskip
\noindent
\underline{\it Proof of \eqref{embbeded Oja}.}
Let
\[
    \bm Y_t:=\bm B_t\bm Q_{t-1}.
\]
By \eqref{eq:block-oja-Bt}, the QR factorization of \(\bm Y_t\) is
\[
    \bm Y_t=\bm Q_t\bm R_t,
\]
where \(\bm Q_t\) has orthonormal columns and \(\bm R_t\) is upper
triangular with positive diagonal. Taking the wedge product of the columns
of both sides gives
\begin{align*}
    (\bm B_t\bm q_{1,t-1})\wedge\cdots\wedge
      (\bm B_t\bm q_{k,t-1})
    &=
    \det(\bm R_t)\,
    \bm q_{1,t}\wedge\cdots\wedge\bm q_{k,t}.
\end{align*}
The left-hand side equals
\[
    \left(\bigwedge^k\bm B_t\right)\bm\omega_{t-1},
\]
while the right-hand side equals \(\det(\bm R_t)\bm\omega_t\). Hence
\begin{equation}
    \left(\bigwedge^k\bm B_t\right)\bm\omega_{t-1}
    =
    \det(\bm R_t)\bm\omega_t.
    \label{eq:wedge-qr-identity}
\end{equation}
Since the columns of \(\bm Q_t\) are orthonormal, $\|\bm\omega_t\|_{\wedge}^2=\det(\bm Q_t^\top\bm Q_t) =1$.
Moreover, the positive-diagonal convention gives \(\det(\bm R_t)>0\).
Taking norms in \eqref{eq:wedge-qr-identity} and dividing by
\(\det(\bm R_t)\) proves \eqref{embbeded Oja}.

\medskip
\noindent
\underline{ \it Proof of \eqref{embedded Oja update}.}
It suffices to verify the identity on a decomposable vector
\[
    \bm\omega=\bm w_1\wedge\cdots\wedge\bm w_k.
\]
By \eqref{eq:oja-update-matrix},
\[
    \bm B_t\bm w_j
    =
    \bm w_j+\eta_t\bm A_t\bm w_j
    =
    \bm w_j+\eta_t\langle\bm X_t,\bm w_j\rangle\bm X_t.
\]
Therefore,
\begin{align*}
    \left(\bigwedge^k\bm B_t\right)\bm\omega
    &=(\bm B_t\bm w_1)\wedge\cdots\wedge(\bm B_t\bm w_k)\\
    &=\bigwedge_{j=1}^k
      \left(
          \bm w_j+\eta_t\langle\bm X_t,\bm w_j\rangle\bm X_t
      \right).
\end{align*}
Expanding by multilinearity, every term containing \(\bm X_t\) in two or
more coordinates vanishes because \(\bm X_t\wedge\bm X_t=0\). Thus
\begin{align*}
    \left(\bigwedge^k\bm B_t\right)\bm\omega
    ={}&
    \bm w_1\wedge\cdots\wedge\bm w_k
    \\
    &+
    \eta_t\sum_{j=1}^k
    \bm w_1\wedge\cdots\wedge\bm w_{j-1}
    \wedge \bm A_t\bm w_j
    \wedge \bm w_{j+1}\wedge\cdots\wedge\bm w_k
    \\
    ={}&
    \left(
        \bm I_{\widetilde{\mb R}_{k,p}}
        +\eta_t\mathcal L_{\bm A_t}
    \right)\bm\omega.
\end{align*}
Decomposable vectors span \(\widetilde{\mb R}_{k,p}\), so the operator
identity \eqref{embedded Oja update} follows. \hfill\(\square\)

We next analyze the eigenvalues and eigenvectors of
\(\mathbb E[\mathcal L_{\bm A_t}]\). Recall that
\[
    \bm\Sigma=\mathbb E\bm A_t
\]
has the eigendecomposition
\[
    \bm\Sigma\bm u_i=\lambda_i\bm u_i,
    \qquad i=1,\ldots,p,
\]
where
\[
    \lambda_1\geq\lambda_2\geq\cdots\geq\lambda_p,
\]
and \(\bm u_1,\ldots,\bm u_p\) form an orthonormal basis of
\(\mb R^p\).

\begin{thm}
\label{thm:mean lift}
It holds that
$
    \mathcal L_{\bm\Sigma}
    =
    \mathbb E\left[\mathcal L_{\bm A_t}\right]
$.
Moreover, the two leading eigenvalues of \(\mathcal L_{\bm\Sigma}\) are
\begin{align*}
    \mu_1
    :=    \sum_{i=1}^k\lambda_i,
    \qquad 
    \mu_2
    :=    \sum_{i=1}^{k-1}\lambda_i+\lambda_{k+1}.
\end{align*}
In particular,
$\mu_1-\mu_2= \Delta_k$.
\end{thm}

\noindent \textbf{Proof of Theorem \ref{thm:mean lift}.}  Define
$\bar{\mathcal L}:= \mb E\left[\mathcal L_{\bm A_t}\right]$. We first show that
$\bar{\mathcal L}=\mathcal L_{\bm\Sigma}$.
Since decomposable $k$-vectors span $\widetilde{\mb R}_{k,p}$, it suffices to
verify the above on a decomposable element
\[
\bm\omega
=
\bm w_1\wedge\cdots\wedge \bm w_k.
\]
By definition,
\[
\mathcal L_{\bm A_t}\bm\omega
=
\sum_{j=1}^k
\bm w_1\wedge\cdots\wedge \bm w_{j-1}
\wedge
\bm A_t \bm w_j
\wedge
\bm w_{j+1}\wedge\cdots\wedge \bm w_k.
\]
Taking expectation and using the linearity of the wedge product in each
argument gives
\begin{align*}
\mb E\left[\mathcal L_{\bm A_t}\bm\omega\right]
&=
\sum_{j=1}^k
\bm w_1\wedge\cdots\wedge \bm w_{j-1}
\wedge
\mb E[\bm A_t]\bm w_j
\wedge
\bm w_{j+1}\wedge\cdots\wedge \bm w_k \\
&= \sum_{j=1}^k
\bm w_1\wedge\cdots\wedge \bm w_{j-1}
\wedge
\bm\Sigma\bm w_j
\wedge
\bm w_{j+1}\wedge\cdots\wedge \bm w_k
=
\mathcal L_{\bm\Sigma}\bm\omega.
\end{align*}
Therefore, $\bar{\mathcal L}=\mathcal L_{\bm\Sigma}$. Next we compute the eigenvectors and eigenvalues of $\mathcal L_{\bm\Sigma}$.
For an index set
\[
I=\{i_1<\cdots<i_k\},
\]
define
\[
\bm u_I
=
\bm u_{i_1}\wedge\cdots\wedge \bm u_{i_k}.
\]
Since $\bm u_1,\dots,\bm u_p$ are an orthonormal basis of $\mb R^p$, the
collection
\[
\{\bm u_I: |I|=k\}
\]
is an orthonormal basis of $\bigwedge^k\mb R^p$.

Now apply $\mathcal L_{\bm\Sigma}$ to $\bm u_I$. By definition,
\[
\mathcal L_{\bm\Sigma}\bm u_I
=
\sum_{\ell=1}^k
\bm u_{i_1}\wedge\cdots\wedge
\bm\Sigma\bm u_{i_\ell}
\wedge\cdots\wedge
\bm u_{i_k}.
\]
Since $\bm\Sigma\bm u_{i_\ell}= \lambda_{i_\ell}\bm u_{i_\ell}$, we get
\[
\mathcal L_{\bm\Sigma}\bm u_I
=
\sum_{\ell=1}^k
\lambda_{i_\ell}
\bm u_{i_1}\wedge\cdots\wedge
\bm u_{i_\ell}
\wedge\cdots\wedge
\bm u_{i_k}.
\]
Therefore,
\[
\mathcal L_{\bm\Sigma}\bm u_I
=
\left(
\sum_{\ell=1}^k\lambda_{i_\ell}
\right)
\bm u_I.
\]
Thus the eigenvalues of $\bar{\mathcal L}$ are all $k$-fold sums of the
eigenvalues of $\bm\Sigma$. It remains to identify the top two eigenvalues. Since $\lambda_1\geq\lambda_2\geq\cdots\geq\lambda_p$,
the largest and second-largest possible sums of $k$ distinct eigenvalues are
\[
\mu_1=\lambda_1+\cdots+\lambda_k,
\qquad
\text{and}
\qquad
\mu_2=\lambda_1+\dots+\lambda_{k-1}+\lambda_{k+1}.
\]
This completes the proof. $\hfill$ $\square$

\section{Proof of Theorem \ref{k-PCA convergence}}

The proof of Theorem \ref{k-PCA convergence} can be adapted from Theorem 3.1 in \cite{jain2016streaming}, except for the analysis of the initialization. For the standard $1$-PCA case in $\mb R^p$, the initialization is uniformly distributed on the hypersphere $\mb S^{p-1}$. However, this is no longer true in the exterior space $\widetilde{\mb R}_{k,p}$. The reasoning is simple: the support of the initialization has to be a subset of the set of all decomposable vectors, which cannot be the whole sphere in $\widetilde{\mb R}_{k,p}$.

\subsection{Anti-concentration}
Let $\bm Q_0=\lb \bm q_{1,0},\dots,\bm q_{k,0} \rb$ be a uniformly distributed $k$-frame and define
\[
\bm \omega_0:= \bm q_{1,0} \wedge \dots\wedge \bm q_{k,0}  \in \widetilde{\mb R}_{k,p}
\]
Recall $r_k(\delta)$ as in Theorem \ref{k-PCA convergence}. The main result of this subsection is the following anti-concentration result:

\begin{prop} \label{anti-concentration}
For a sufficiently large, universal constant $C>0$ in the definition of $r_k(\delta)$, we have 
\[
\mb P \la \binom{p}{k} \cdot \left| \langle \bm \omega_0, \bm \alpha  \rangle_\wedge \right|^2  \leq \frac{1}{r_k(\delta)} \ra \leq \frac{\delta}{3}
\]
where $\bm \alpha \in \widetilde{\mb R}_{k,p}$ is an arbitrary decomposable vector with unit length. Here $\delta \in (0,1)$.
\end{prop}

\noindent \textbf{Proof of Proposition \ref{anti-concentration}.}
We split the proof into several steps.

\noindent\underline{\it Step 1: Exact distributional representation of $\langle \bm \omega_0, \bm \alpha  \rangle_\wedge$.} We will show that with $r:=\min \la k, p-k \ra$, we have 
\begin{align} \label{diributional-wedge-product}
     \left| \langle \bm \omega_0, \bm \alpha  \rangle_\wedge \right|^2 \stackrel{d}{=} \prod_{j=1}^r B_j
\end{align}
where $B_j$'s are independent with distribution $B_j \sim \mbox{Beta} \lb j/2, (p-r)/2 \rb$.

Let $\bm A$ be the matrix representation of $\bm \alpha$. By Proposition \ref{prop:inner-product}, we have 
\[
 \left| \langle \bm\omega_0, \bm \alpha  \rangle_\wedge \right|^2 \stackrel{d}{=} \mbox{det} \lb \bm A^\top \bm Q_0 \rb^2. 
\]
By rotational invariance, we can assume, without loss of generality, that 
\[
\bm A:=
    \begin{pmatrix}
        \bm I_k\\
        \bm 0
    \end{pmatrix}.
\]
 Let $Z
    :=
    \left|
        \langle \bm\omega_0,\bm\alpha\rangle_\wedge
    \right|^2
$.
Write 
\[
    Z
    \stackrel{d}{=}
    \det\left(\bm E_k^\top\bm Q_0\right)^2,
    \qquad
    \bm E_k:=
    \begin{pmatrix}
        \bm I_k\\
        \bm 0
    \end{pmatrix}.
\]
Note that $\bm Q_0$ has the same distribution as the first $k$ columns of a Haar-distributed
orthogonal matrix $\bm H\in O(p)$. Therefore, $\bm E_k^\top \bm Q_0 \stackrel{d}{=} \bm H_{11}$, where $\bm H_{11}$ is the first $k\times k$ block of $\bm H$:
\[
\bm H =  \lb \begin{matrix}
    \bm H_{11} & \bm H_{12} \\
    \bm H_{21} & \bm H_{22}
\end{matrix}
\rb .
\]
Thus $Z\stackrel{d}{=}\det\left(\bm H_{11}\right)^2$.
By Proposition 2.1 in \cite{rouault2007asymptotic}, for every integer $1\leq m<p$,
\begin{equation}
    \det\left(\bm H_{1:m,1:m}\right)^2
    \stackrel{d}{=}
    \prod_{\ell=1}^{m}B_{\ell}^{(m)},
    \label{eq:haar-square-block-beta}
\end{equation}
where the random variables
$B_1^{(m)},\ldots,B_m^{(m)}$ are mutually independent and
\begin{equation}
    B_{\ell}^{(m)}
    \sim
    \operatorname{Beta}
    \left(
        \frac{\ell}{2},
        \frac{p-m}{2}
    \right),
    \qquad
    \ell=1,\ldots,m.
    \label{eq:haar-square-block-beta-variables}
\end{equation}

We now apply \eqref{eq:haar-square-block-beta} and \eqref{eq:haar-square-block-beta-variables} to $Z$. If
$k\leq p/2$, then $r=k$, and setting $m=k$ gives
\[
    Z
    \stackrel{d}{=}
    \prod_{\ell=1}^{k}B_\ell,
    \qquad
    B_\ell
    \stackrel{\mathrm{ind}}{\sim}
    \operatorname{Beta}
    \left(
        \frac{\ell}{2},
        \frac{p-k}{2}
    \right).
\]
This is exactly \eqref{diributional-wedge-product} in this case.

Suppose next that $k>p/2$, and put $r:=p-k<\frac p2$.
As computed above, we have
\[
    Z \stackrel{d}{=}\det(\bm H_{11})^2.
\]
On the other hand, orthogonality of the first $k$ columns gives
\[
    \bm H_{11}^{\top}\bm H_{11}
    +
    \bm H_{21}^{\top}\bm H_{21}
    =
    \bm I_k,
\]
whereas orthogonality of the last $r$ rows gives
\[
    \bm H_{21}\bm H_{21}^{\top}
    +
    \bm H_{22}\bm H_{22}^{\top}
    =
    \bm I_r.
\]
It follows from Sylvester's determinant identity that
\begin{align}
    \det(\bm H_{11})^2
    =
    \det\left(
        \bm I_k-\bm H_{21}^{\top}\bm H_{21}
    \right)
    =
    \det\left(
        \bm I_r-\bm H_{21}\bm H_{21}^{\top}
    \right)
    =
    \det\left(
        \bm H_{22}\bm H_{22}^{\top}
    \right)
    =
    \det(\bm H_{22})^2.
    \label{eq:complementary-haar-block}
\end{align}
By Haar invariance, $\bm H_{22}$ has the same distribution as
the top-left $r\times r$ block of a Haar orthogonal matrix.
Since $r=p-k<p/2$, the already proved case, applied with $k$
replaced by $r$, yields
\[
    Z
    \stackrel{d}{=}
    \prod_{\ell=1}^{r}B_\ell,
    \qquad
    B_\ell
    \stackrel{\mathrm{ind}}{\sim}
    \operatorname{Beta}
    \left(
        \frac{\ell}{2},
        \frac{p-r}{2}
    \right).
\]
Combining the two cases proves \eqref{diributional-wedge-product}.

\noindent \underline{\it Step 2: Sub-Gamma concentration.}
Put
\[
    q:=p-r,
    \qquad
    D_k:=\binom{p}{k}=\binom{p}{r},
    \qquad
    Z:=\prod_{j=1}^{r}B_j,
\]
where, by Step~1,
\[
    B_j
    \stackrel{\mathrm{ind}}{\sim}
    \operatorname{Beta}
    \left(
        \frac{j}{2},
        \frac{q}{2}
    \right),
    \qquad
    j=1,\ldots,r.
\]
Since $r=\min\{k,p-k\}$, we have $q\geq r$. Define
\[
    X:=-\log(D_kZ).
\]

We show that there exists a universal constant $C>0$ such that
\begin{equation}
    \mathbb E X\leq \log(Cr)
    \label{eq:mean-log-overlap}
\end{equation}
and that $X-\mathbb EX$ is sub-gamma on the right with variance
factor $C\log(er)$ and a universal scale factor.

Let $\psi$ and $\psi_1$ denote the digamma and trigamma functions,
respectively. Since
\[
    \mathbb E[-\log B_j]
    =
    \psi\left(\frac{q+j}{2}\right)
    -
    \psi\left(\frac{j}{2}\right),
\]
we have
\[
    \mathbb EX
    =
    -\log D_k
    +
    \sum_{j=1}^{r}
    \left[
        \psi\left(\frac{q+j}{2}\right)
        -
        \psi\left(\frac{j}{2}\right)
    \right].
\]
Using
\[
    \psi(x)\leq\log x,
    \qquad
    \psi(x)\geq\log\left(x-\frac12\right),
    \quad x>\frac12,
\]
and treating the term $j=1$ separately, we obtain
\begin{align*}
    \mathbb EX
    &\leq
    -\log D_k
    +
    \log(q+1)
    +
    C
    +
    \sum_{j=2}^{r}
    \log\left(\frac{q+j}{j-1}\right)
    \\
    &=
    -\log D_k
    +
    \log D_k
    +
    \log r
    +
    C
    \leq
    \log(Cr),
\end{align*}
which proves \eqref{eq:mean-log-overlap}. Note that we have used the identity
\[
  \sum_{j=2}^{r}
    \log\left(\frac{q+j}{j-1}\right) = \log \left[ \prod_{j=2}^r \frac{q+j}{j-1}  \right] = \log \binom{q+r}{r-1} = \log \binom{p}{r-1} = \log D_k + \log \lb \frac{r}{q+1} \rb
\]
to deduce the second line. Next, let
\[
    \Lambda(\theta)
    :=
    \log\mathbb E
    \exp\left\{
        \theta(X-\mathbb EX)
    \right\}.
\]
The beta negative-moment formula gives, for
$0\leq\theta<1/2$,
\[
    \mathbb E B_j^{-\theta}
    =
    \frac{
        \Gamma\left(\frac{j}{2}-\theta\right)
        \Gamma\left(\frac{q+j}{2}\right)
    }{
        \Gamma\left(\frac{j}{2}\right)
        \Gamma\left(\frac{q+j}{2}-\theta\right)
    }.
\]
Consequently,
\begin{align*}
    \Lambda''(\theta)  =  \frac{d^2}{d \theta^2} \left[ \sum_{j=1}^r \log \lb  \mb E B_j^{-\theta} \rb\right] &=   \frac{d^2}{d \theta^2} \left[ \sum_{j=1}^r \log \Gamma \lb \frac{j}{2} - \theta \rb - \log \Gamma \lb \frac{q+j}{2} - \theta \rb  \right] \\
  &= \sum_{j=1}^{r}
    \left[
        \psi_1\left(\frac{j}{2}-\theta\right)
        -
        \psi_1\left(\frac{q+j}{2}-\theta\right)
    \right].
\end{align*}
Since $\psi_1$ is positive and decreasing, for
$0\leq\theta\leq1/4$,
\begin{align*}
    \Lambda''(\theta)
    \leq
    \sum_{j=1}^{r}
    \psi_1\left(\frac{j}{2}-\theta\right)
    \leq
    C\sum_{j=1}^{r}
    \left(
        \frac1j+\frac1{j^2}
    \right)
    \leq
    C\log(er).
\end{align*}
Here we used the elementary bound
\[
    \psi_1(x)
    =
    \sum_{\ell=0}^{\infty}\frac1{(x+\ell)^2}
    \leq
    \frac1{x^2}+\frac1x.
\]
Next, note that $\Lambda(0)=\Lambda'(0)=0$, so integration twice yields
\[
    \Lambda(\theta)
    \leq
    C\theta^2\log(er),
    \qquad
    0\leq\theta\leq\frac14.
\]
In particular,
\[
    \Lambda(\theta)
    \leq
    \frac{
        C\theta^2\log(er)
    }{
        2(1-4\theta)
    },
    \qquad
    0<\theta<\frac14.
\]
Thus $X-\mathbb EX$ is sub-gamma on the right with variance
factor $C\log(er)$ and scale factor $4$.

By the sub-gamma concentration inequality, for every $u>0$,
\begin{equation}
    \mathbb P\left(
        X-\mathbb EX
        \geq
        C\sqrt{\log(er)\,u}
        +
        Cu
    \right)
    \leq
    e^{-u}.
    \label{eq:sub-gamma-log-overlap}
\end{equation}
Taking $u:=\log\left(\frac{3}{\delta}\right)$ and using \eqref{eq:mean-log-overlap}, we obtain
\begin{align*}
    \mathbb P\Bigg(
        X
        \geq
        \log(Cr)
        +
        C\sqrt{
            \log(er)
            \log\left(\frac{3}{\delta}\right)
        }
        +
        C\log\left(\frac{3}{\delta}\right)
    \Bigg)
    \leq
    \frac{\delta}{3}.
\end{align*}
Equivalently,
\begin{align}
    \mathbb P\Bigg(
        D_k Z
        \leq
        \frac{1}{
            Cr
            \exp\left\{
                C\sqrt{
                    \log(er)
                    \log(3/\delta)
                }
                +
                C\log(3/\delta)
            \right\}
        }
    \Bigg)
    \leq
    \frac{\delta}{3}.
    \label{eq:plucker-anticoncentration-subgamma}
\end{align}

Therefore, by taking the universal constant in
the definition of $r_k(\delta)$ to be sufficiently large and noting that $r\leq k$, we have
\[
    r_k(\delta)
    \geq
    Cr
    \exp\left\{
        C\sqrt{
            \log(er)\log(3/\delta)
        }
        +
        C\log(3/\delta)
    \right\}.
\]
It then follows from \eqref{eq:plucker-anticoncentration-subgamma} that
\[
    \mathbb P\left( \,
        \binom pk \cdot 
        \left|
            \langle\bm\omega_0,\bm\alpha\rangle_\wedge
        \right|^2
        \leq
        \frac1{r_k(\delta)}
    \right)
    \leq
    \frac{\delta}{3}.
\]

\subsection{Proof of \eqref{k-PCA-F-bound}.}

We adapt the proof of Theorem 3.1 in \cite{jain2016streaming}. The main difference is to have a better variance term, namely $\sigma_k^2$ in \eqref{k-PCA-F-bound}. Suppose $\bm \omega_\star$ is the image of  $\la \bm u_1,\dots, \bm u_k \ra$ under the Pl\"ucker embedding. Let us begin with a variance calculation. 

\begin{lemma} \label{lem:sigma_k}
    Suppose $\mathsf{\bm P}_\star:= \bm \omega_\star \bm \omega_\star^\top$ and  $ \mathsf{\bm P}_\star^\perp := \bm I_{\widetilde{\mb R}_{k,p}}
    -
    \mathsf{\bm P}_\star$. Then with $\mc L$ as in \eqref{L}, we have
    \[
   \mb E \lb  \| \mathsf{\bm P}_\star^\perp \, \mc L_{\bm A_t - \bm \Sigma} \, \bm \omega_{\star}  \|_\wedge^2   \rb=\sum_{i=1}^k\sum_{j=k+1}^p \mathbb E \left[  \left(  \bm u_j^\top (\bm A_t-\bm\Sigma) \bm u_i \right)^2 \right].
    \]
\end{lemma}

\noindent \textbf{Proof of Lemma \ref{lem:sigma_k}.}
Put 
\begin{align} \label{Et}
\bm E_t:= \bm A_t - \bm \Sigma.
\end{align}
Recall that $\bm \omega_\star =  \bm u_1 \wedge \dots \wedge \bm u_k$, so
\begin{align*}
    \mc L_{\bm A_t - \bm \Sigma} \, \bm \omega_{\star} =  \sum_{i=1}^k \bm u_1 \wedge \dots \wedge   \bm u_{i-1} \wedge \bm E_t \bm u_i  \wedge \bm u_{i+1} \wedge \dots \wedge  \bm u_k.
\end{align*}
Let $\la \bm u_{k+1},\dots, \bm u_p \ra$ be an orthonormal set such that the set $\la \bm u_1,\dots, \bm u_p \ra$ is an orthonormal basis of $\mb R^p$. Substituting $\bm E_t \bm u_i = \sum_{j=1}^p \lb \bm u_j^\top \, \bm E_t \, \bm u_i \rb \, \bm u_j$ into the preceding display, we obtain 
\begin{align*}
 \mc L_{\bm A_t - \bm \Sigma} \, \bm \omega_{\star} &= \left[ \sum_{i=1}^k  \bm u_i^\top \, \bm E_t \, \bm u_i \right] \cdot \bm \omega_\star \\
&+ \sum_{i=1}^k \sum_{j=k+1}^p  
 \lb \bm u_j^\top \, \bm E_t \, \bm u_i \rb \cdot \bm u_1 \wedge \dots \wedge   \bm u_{i-1} \wedge  \bm u_j  \wedge \bm u_{i+1} \wedge \dots \wedge  \bm u_k.
\end{align*}
Note that the operator $\mathsf{\bm P}_\star^\perp$ projects onto the orthogonal complement of $\bm \omega_\star$, so Parseval's identity yields 
\begin{align*}
 \| \mathsf{\bm P}_\star^\perp \, \mc L_{\bm A_t - \bm \Sigma} \, \bm \omega_{\star}  \|_\wedge^2  &= 
 \sum_{i=1}^k \sum_{j=k+1}^p  
 \lb \bm u_j^\top \, \bm E_t \, \bm u_i \rb^2 
 = \sum_{i=1}^k \sum_{j=k+1}^p  \left[    \bm u_j^\top \, (\bm A_t-\bm\Sigma) \, \bm u_i  \right]^2.
\end{align*}
The proof is completed by taking expectations on both sides. $\hfill$ $\square$

The next lemma collects some elementary estimates on the operator $\mc L$, the eigenvalues and variance profile of the embedded Oja iterates.

\begin{lemma} \label{embedded profile}
    For every integer $\ell \geq 1$, it holds that
    \begin{align}
          \|\bm A_t\|_{\rm op} &\leq L, \nonumber \\
           \mc L_{\bm A_t}^{\,\ell} &= \|\bm X_t\|^{2(\ell-1)} \mc L_{\bm A_t},  
           \qquad
    0\preceq \mc L_{\bm A_t} \preceq L \cdot \bm I_{\widetilde{\mb R}_{k,p}}, \label{eq:rank-one-powers-exterior}  \\
        \left\| \mb E\left[  \left( \mc L_{\bm A_t} -\mc L_{\bm\Sigma}\right)^2\right] \, \right\|_{\rm op}
   & \leq L\mu_1. \label{eq:exterior-global-variance-norm}
    \end{align}
\end{lemma}

\noindent \textbf{Proof of Lemma \ref{embedded profile}.} The first bound follows from the triangle inequality. To show the first identity in \eqref{eq:rank-one-powers-exterior}, it suffices, by linearity, to consider a decomposable vector
$\bm\omega=\bm w_1\wedge\cdots\wedge\bm w_k$. This is because decomposable vectors span $\widetilde{\mb R}_{k,p}$.

Note that $\bm A_t^2 = \| \bm X_t\|^2 \cdot \bm A_t$, so we can write
\begin{align*}
\mc L_{\bm A_t}^2 \bm \omega = \mc L_{\bm A_t} \lb \mc L_{\bm A_t}  \bm \omega   \rb &=  \mc L_{\bm A_t} \lb \sum_{j=1}^k
\bm w_1\wedge\cdots\wedge \bm w_{j-1}
\wedge \bm A_t \bm w_j
\wedge \bm w_{j+1}\wedge\cdots\wedge \bm w_k  \rb \\
&= \sum_{i=1}^k \sum_{j=1}^k \bm w_1 \wedge \dots \wedge \bm A_t \bm w_i \wedge \dots\wedge \bm A_t \bm w_j \wedge \dots \wedge \bm w_k \\
&= \sum_{j=1}^k \bm w_1\wedge\cdots\wedge \bm w_{j-1}
\wedge \bm A_t^2  \bm w_j
\wedge \bm w_{j+1}\wedge\cdots\wedge \bm w_k \\
&= \| \bm X_t\|^2 \cdot \sum_{j=1}^k \bm w_1\wedge\cdots\wedge \bm w_{j-1}
\wedge \bm A_t  \bm w_j
\wedge \bm w_{j+1}\wedge\cdots\wedge \bm w_k
\end{align*}
where the third equality follows from the fact that for $i \neq j$,
\[
\bm A_t \bm w_i \wedge \bm A_t \bm w_j = \langle \bm X_t, \bm w_i \rangle \cdot \langle  \bm X_t, \bm w_j \rangle \cdot \bm X_t \wedge \bm X_t =\bm 0. 
\]
Thus $\mc L_{\bm A_t}^2 = \| \bm X_t\|^2 \cdot \mc L_{\bm A_t}$. A simple induction argument then yields the first item in \eqref{eq:rank-one-powers-exterior}. The second item in \eqref{eq:rank-one-powers-exterior} follows from Theorem \ref{thm:mean lift}, which implies that the eigenvalues of $\mc L_{\bm A_t}$ are either $0$ or $\| \bm X_t \|^2$, and 
$\| \bm X_t \|^2 = \| \bm A_t \|_{\rm op} \leq L$.

To show \eqref{eq:exterior-global-variance-norm},  write 
\begin{align*}
\mb E\left[ \left( \mc L_{\bm A_t}-\mc L_{\bm\Sigma} \right)^2\right]
    = \mb E\left[ \mc L_{\bm A_t}^{\,2} \right] - \mc L_{\bm\Sigma}^{\,2} 
    \preceq \mb E\left[ \mc L_{\bm A_t}^{\,2} \right] \preceq L \cdot \mb E \left[ \mc L_{\bm A_t} \right]
\end{align*}
where the last bound follows from $\mc L_{\bm A_t}^2 = \| \bm X_t\|^2 \cdot \mc L_{\bm A_t}$ and $\| \bm X_t \|^2 \leq L$. Taking operator norms on both sides in the display above, we get \eqref{eq:exterior-global-variance-norm}. $\hfill$ $\square$

We split the proof of \eqref{k-PCA-F-bound} into several steps. Define
\begin{align*}
    \mfH_t:= \bm I_{\widetilde{\mb R}_{k,p}} + \eta_t \mc L _{\bm A_t} ; \qquad \mfB_t= \prod_{i=t}^1 \mfH_i; \qquad \mfB_0:=  \bm I_{\widetilde{\mb R}_{k,p}}.
\end{align*}
By Theorem \ref{thm: Plucker embedding}, the image of the Oja iterate $\bm Q_n$ under the embedding is 
\[
\bm \omega_n = \frac{\mfB_n \, \bm \omega_0}{\| \mfB_n \,  \bm \omega_0  \|_\wedge}.
\]

\noindent \underline{ \it Step 1: Removing the initialization.} Write
\begin{align*}
    \left\|
        \sin\Theta(\bm Q_n,\bm U_k^\star)
    \right\|_{\rm F}^2 \leq  \tan^2 \lb \bm \omega_n , \bm \omega_\star  \rb
    = \frac{1 - \left| \langle \bm \omega_\star, \bm \omega_n  \rangle_\wedge \right|^2 }{\left| \langle \bm \omega_\star, \bm \omega_n  \rangle_\wedge \right|^2} =  \frac{X_n}{Y_n \cdot \left| \langle \bm \omega_0, \bm \alpha_n  \rangle_\wedge \right|^2}
\end{align*}
where 
\[
X_n:= \bm \omega_0^\top \mfB_n^\top \mathsf{\bm P}_\star^\perp \mfB_n \bm \omega_0, \qquad Y_n:= \bm \omega_\star^\top \mfB_n \mfB_n^\top \bm \omega_\star, \qquad \bm \alpha_n:= \frac{ \mfB_n^\top \bm \omega_\star }{\| \mfB_n^\top \bm \omega_\star  \|_\wedge}.
\]
Put $D_k:= \binom{p}{k}$. By Proposition \ref{anti-concentration}, we have 
\[
\mb P \lb D_k \cdot \left| \langle \bm \omega_0, \bm \alpha_n  \rangle_\wedge \right|^2 \geq \frac{1}{r_k(\delta)} \, \Big|  \bm X_1,\dots, \bm X_n  \rb \geq 1 - \frac{\delta}{3}
\]
conditionally on the data $\bm X_i$'s.

Thus with probability at least $1-\delta/3$, 
\begin{align} \label{conditional bound-1}
 \tan^2 \lb \bm \omega_n , \bm \omega_\star  \rb  \leq D_k \cdot r_k(\delta) \cdot \frac{X_n}{Y_n}.
\end{align}

\noindent \underline{\it Step 2: High-probability bound on $X_n$.}
For $0\leq t\leq n$, define
\begin{align} \label{g}
    \mathsf{\bm G}_t
    :=
    \mb E\left[
        \mfB_t\mfB_t^\top
    \right],
    \qquad
    g_t
    :=
    \prod_{s=1}^t
    \left(
        1+\eta_s\mu_1
    \right)^2, 
    \qquad 
    g_0:=1.
\end{align}
Put
\[
    m_t
    :=
    \operatorname{tr}
    \left(
        \mathsf{\bm P}_\star^\perp
        \mathsf{\bm G}_t
    \right).
\]
Our goal is to bound $m_t$ via a recursion analysis. Independence gives
\[
    \mathsf{\bm G}_t
    =
    \mb E\left[
        \mfH_t
        \mathsf{\bm G}_{t-1}
        \mfH_t
    \right].
\]
Since
\[
    \mb E\left[
        \mfH_t^2
    \right]
    =
    \bm I_{\widetilde{\mb R}_{k,p}}
    +
    2\eta_t\mc L_{\bm\Sigma}
    +
    \eta_t^2
    \mb E\left[
        \mc L_{\bm A_t}^{\,2}
    \right],
\]
we obtain
\[
    \left\|
        \mb E\left[
            \mfH_t^2
        \right]
    \right\|_{\rm op}
    \leq
    1+2\eta_t\mu_1+\eta_t^2L\mu_1 \leq  \left(
        1+\eta_t\mu_1
    \right)^2
    \exp\left(
        L\mu_1\eta_t^2
    \right).
\]
Moreover, since $\mfG_{t-1} \preceq \| \mfG_{t-1} \|_{\rm op} \cdot \bm I_{\widetilde{\mb R}_{k,p}}$ and $\mfH_t$ is symmetric, it follows that
\begin{align*}
    \mfH_t \mfG_{t-1} \mfH_t = \mfH_t^\top \mfG_{t-1} \mfH_t \preceq \| \mfG_{t-1} \|_{\rm op} \cdot \mfH_t^2 
\end{align*}
Taking expectation, we deduce that
\[
\| \mfG_{t} \|_{\rm op} \leq \| \mfG_{t-1} \|_{\rm op} \cdot \left\| \mb E \left[ \mfH_t^2 \right] \right\|_{\rm op}
\leq  \| \mfG_{t-1} \|_{\rm op} \cdot \left(
        1+\eta_t\mu_1
    \right)^2
    \exp\left(
        L\mu_1\eta_t^2
    \right).
\]
It follows inductively that
\begin{equation}
    \left\|
        \mathsf{\bm G}_t
    \right\|_{\rm op}
    \leq
    g_t
    \exp\left(
        L\mu_1
        \sum_{s=1}^t\eta_s^2
    \right).
    \label{eq:total-second-moment}
\end{equation}
We now derive the recursion for $m_t$. For any
$\bm\xi\in\widetilde{\mb R}_{k,p}$, write
\[
    \bm\xi
    =
    \langle \bm \xi, \bm \omega_\star \rangle_\wedge \cdot\bm\omega_\star+\bm\xi_\perp,
    \qquad
    \bm\xi_\perp\perp\bm\omega_\star.
\]
 With $\bm E_t$ as in \eqref{Et}, Lemma~\ref{lem:sigma_k} and \eqref{eq:exterior-global-variance-norm} imply
\begin{align*}
    \bm\xi^\top
    \mb E\left[
        \mc L_{\bm E_t} \,
        \mathsf{\bm P}_\star^\perp \,
        \mc L_{\bm E_t}
    \right]
    \bm\xi
    &=
    \mb E\left\|
        \mathsf{\bm P}_\star^\perp \,
        \mc L_{\bm E_t}
        \left(
            \langle \bm \xi, \bm \omega_\star \rangle_\wedge \cdot \bm\omega_\star+\bm\xi_\perp
        \right)
    \right\|_\wedge^2
    \\
    &\leq
    2\langle \bm \xi, \bm \omega_\star \rangle_\wedge^2 \cdot 
    \mb E\left\|
        \mathsf{\bm P}_\star^\perp \,
        \mc L_{\bm E_t} \, 
        \bm\omega_\star
    \right\|_\wedge^2
    +
    2\mb E\left\|
        \mc L_{\bm E_t}
        \, \bm\xi_\perp
    \right\|_\wedge^2
    \\
    &\leq
    2\langle \bm \xi, \bm \omega_\star \rangle_\wedge^2 \cdot  \sigma_k^2
    +
    2L\mu_1 \cdot \|\bm\xi_\perp\|_\wedge^2
\end{align*}
where we have used the Cauchy--Schwarz inequality in the first estimate. Therefore,
\begin{equation}
        \mb E\left[
        \mc L_{\bm E_t} \,
        \mathsf{\bm P}_\star^\perp \,
        \mc L_{\bm E_t}
    \right]
    \preceq
    2\sigma_k^2\mathsf{\bm P}_\star
    +
    2L\mu_1\mathsf{\bm P}_\star^\perp.
    \label{eq:directional-variance-operator}
\end{equation}
Consequently, with $\bm E_t$ as in \eqref{Et} and by using $\mfH_t = \bm I_{\widetilde{\mb R}_{k,p}} + \eta_t \mc L_{\bm \Sigma} + \eta_t \mc L_{\bm E_t}$, we get
\begin{align*}
     \mb E\left[ \mfH_t\,
        \mathsf{\bm P}_\star^\perp\,
        \mfH_t
    \right]
    &= \mb E \left[ \lb  \bm I_{\widetilde{\mb R}_{k,p}} + \eta_t \mc L_{\bm \Sigma} + \eta_t \mc L_{\bm E_t} \rb \,  \mathsf{\bm P}_\star^\perp   
    \, \lb  \bm I_{\widetilde{\mb R}_{k,p}} + \eta_t \mc L_{\bm \Sigma} + \eta_t \mc L_{\bm E_t} \rb
    \right] \\
    &= \lb \bm I_{\widetilde{\mb R}_{k,p}} + \eta_t \mc L_{\bm \Sigma} \rb \,  \mathsf{\bm P}_\star^\perp\  \,
     \lb \bm I_{\widetilde{\mb R}_{k,p}} + \eta_t \mc L_{\bm \Sigma} \rb
     + \eta_t^2 \cdot \mb E \left[   \mc L_{\bm E_t} 
     \, \mathsf{\bm P}_\star^\perp \, \mc L_{\bm E_t} 
    \right] \\
    &+ \eta_t \cdot \underbrace{\mb E \left[ \lb  \bm I_{\widetilde{\mb R}_{k,p}} + \eta_t \mc L_{\bm \Sigma} \rb \,  \mathsf{\bm P}_\star^\perp \,  \mc L_{\bm E_t} \right]}_{= \bm 0}
    + \eta_t \cdot  \underbrace{\mb E \left[  \mc L_{\bm E_t} \, \mathsf{\bm P}_\star^\perp \, \lb  \bm I_{\widetilde{\mb R}_{k,p}} + \eta_t \mc L_{\bm \Sigma} \rb \right]}_{= \bm 0} \\
    &= \lb \bm I_{\widetilde{\mb R}_{k,p}} + \eta_t \mc L_{\bm \Sigma} \rb \,  \mathsf{\bm P}_\star^\perp\  \,
     \lb \bm I_{\widetilde{\mb R}_{k,p}} + \eta_t \mc L_{\bm \Sigma} \rb
     + \eta_t^2 \cdot \mb E \left[   \mc L_{\bm E_t} 
     \, \mathsf{\bm P}_\star^\perp \, \mc L_{\bm E_t} 
    \right] 
\end{align*}
Note that the two terms in the third line of the display above vanish because
\[
\mb E \left[ \mc L_{\bm E_t} \right] = \mc L_{\mb E \lft \bm E_t \rht} = \bm 0.
\]
Furthermore, since $\mathsf{\bm P}_\star^\perp$ commutes with $\mc L_{\bm \Sigma}$, we get from \eqref{eq:directional-variance-operator} that
\begin{align*}
\mb E\left[ \mfH_t\,
        \mathsf{\bm P}_\star^\perp\,
        \mfH_t
    \right]
   & = \mathsf{\bm P}_\star^\perp \, \lb \bm I_{\widetilde{\mb R}_{k,p}} + \eta_t \mc L_{\bm \Sigma}  \rb^2 \, \mathsf{\bm P}_\star^\perp + \eta_t^2 \cdot \mb E \left[   \mc L_{\bm E_t} 
     \, \mathsf{\bm P}_\star^\perp \, \mc L_{\bm E_t} 
    \right] \\
    & \preceq \lb 1 + \eta_t \mu_2 \rb^2 \cdot \mathsf{\bm P}_\star^\perp + \eta_t^2 \lb 2\sigma_k^2\mathsf{\bm P}_\star
    +
    2L\mu_1\mathsf{\bm P}_\star^\perp  \rb.
\end{align*}
Now, observe that
\begin{align*}
    m_t = \mb E \mbox{tr} \left[  \mathsf{\bm P}_\star^\perp \mfB_t  \mfB_t^\top\right]  = \mb E \mbox{tr} \left[ \mathsf{\bm P}_\star^\perp \mfH_t \mfB_{t-1} \mfB_{t-1}^\top \mfH_t  \right] 
    &= \mb E \mbox{tr} \left[ \mfB_{t-1} \mfB_{t-1}^\top \mfH_t  \mathsf{\bm P}_\star^\perp \mfH_t \right] \\
    &= \mbox{tr} \la \mfG_{t-1} \, \mb E\left[ \mfH_t\,
        \mathsf{\bm P}_\star^\perp\,
        \mfH_t
    \right] \ra
\end{align*}
Note that if $\bm A \preceq \bm B$ and $\bm C \succeq 0$, then $ \mbox{tr} \lb \bm A \bm C \rb \leq \mbox{tr} \lb \bm B \bm C \rb$, so 
\begin{equation}
\begin{aligned}
    m_t 
    \leq&
    \left[
        \left(
            1+\eta_t(\mu_1-\Delta_k)
        \right)^2
        +
        2L\mu_1\eta_t^2
    \right]
    m_{t-1}
    +
    2\sigma_k^2\eta_t^2
    \left\|
        \mathsf{\bm G}_{t-1}
    \right\|_{\rm op}.
\end{aligned}
\label{eq:orthogonal-mass-recursion-short}
\end{equation}
For every $t$,
\begin{align*}
    &\left(
        1+\eta_t(\mu_1-\Delta_k)
    \right)^2
    +
    2L\mu_1\eta_t^2
\leq
    \exp\left\{
        2\eta_t(\mu_1-\Delta_k)
        +
        3L\mu_1\eta_t^2
    \right\},
\end{align*}
where we used
$(\mu_1-\Delta_k)^2\leq\mu_1^2\leq L\mu_1$.
Also,
\[
    g_t
    \geq
    \exp\left\{
        2\mu_1\sum_{s=1}^t\eta_s
        -
        L\mu_1\sum_{s=1}^t\eta_s^2
    \right\},
\]
because $\log(1+x)\geq x-x^2/2$ for $x\geq0$.

Iterating \eqref{eq:orthogonal-mass-recursion-short}, using
\eqref{eq:total-second-moment}, and dividing by $g_n$, we obtain
\begin{align}
    \frac{m_n}{g_n}
    \leq
    C
    \exp\left(
        4L\mu_1
        \sum_{t=1}^n\eta_t^2
    \right)
    \Bigg[
        &
        D_k
        \exp\left(
            -2\Delta_k
            \sum_{t=1}^n\eta_t
        \right)
        +
        \sigma_k^2
        \sum_{i=1}^n
        \eta_i^2
        \exp\left(
            -2\Delta_k
            \sum_{j=i+1}^n\eta_j
        \right)
    \Bigg].
    \label{eq:expected-orthogonal-mass-final}
\end{align}
By the learning-rate condition
$L\mu_1\sum_{t=1}^n\eta_t^2\leq c\delta$, the first exponential
factor in \eqref{eq:expected-orthogonal-mass-final} is bounded by
a universal constant, provided $c>0$ is sufficiently small. Hence
\begin{equation}
\begin{aligned}
    \frac{m_n}{g_n}
    \lesssim
    &
    D_k
    \exp\left(
        -2\Delta_k
        \sum_{t=1}^n\eta_t
    \right)
    +
    \sigma_k^2
    \sum_{i=1}^n
    \eta_i^2
    \exp\left(
        -2\Delta_k
        \sum_{j=i+1}^n\eta_j
    \right).
\end{aligned}
\label{eq:expected-orthogonal-mass-simplified}
\end{equation}
Finally, Haar invariance gives
\[
    \mb E\left[
        \bm\omega_0\bm\omega_0^\top
    \right]
    =
    \frac1{D_k}
    \bm I_{\widetilde{\mb R}_{k,p}}.
\]
Thus
\[
    \mb E X_n
    =
    \frac{1}{D_k}
    \mb E\operatorname{tr}
    \left(
        \mathsf{\bm P}_\star^\perp
        \mfB_n\mfB_n^\top
    \right)
    =
    \frac{m_n}{D_k}.
\]
Therefore, by Markov's inequality,
\begin{equation}
    X_n
    \leq
    \frac{3m_n}{D_k\delta}
    \label{eq:Xn-high-probability}
\end{equation}
with probability at least $1-\delta/3$.

\noindent \underline{\it Step 3: High-probability bound on $Y_n$.}
Recall that
\[
    Y_n
    =
    \bm\omega_\star^\top
    \mfB_n\mfB_n^\top
    \bm\omega_\star.
\]
As in the moment computation in Step 2, we have 
\begin{align*}
    \mb E\left[
        \mfH_t
        \mathsf{\bm P}_\star
        \mfH_t
    \right]
    =
    \left(
        \bm I+\eta_t\mc L_{\bm\Sigma}
    \right)
    \mathsf{\bm P}_\star
    \left(
        \bm I+\eta_t\mc L_{\bm\Sigma}
    \right)
    +
    \eta_t^2
    \mb E\left[
        \mc L_{\bm E_t}
        \mathsf{\bm P}_\star
        \mc L_{\bm E_t}
    \right]
    \succeq
    \left(
        1+\eta_t\mu_1
    \right)^2
    \mathsf{\bm P}_\star.
\end{align*}
Therefore, with the same argument used in \eqref{eq:orthogonal-mass-recursion-short},
\begin{align*}
    \mb E Y_n
    =
    \operatorname{tr}
    \left(
        \mathsf{\bm P}_\star
        \mathsf{\bm G}_n
    \right) 
   &= \mb E \mbox{tr} \lb \mfB_{n-1} \mfB_{n-1}^\top \, \mfH_n \mathsf{\bm P}_{\star} \mfH_n  \rb  \\
   &= \mbox{tr} \lb \mfG_{n-1} \mb E \left[ \mfH_n \mathsf{\bm P}_{\star} \mfH_n  \right] \rb \\
   & \geq \lb 1 + \eta_n \mu_1 \rb^2 \mbox{tr} \lb \mathsf{\bm P}_\star
        \mathsf{\bm G}_{n-1} \rb 
    = \lb 1 + \eta_n \mu_1 \rb^2 \cdot \mb E Y_{n-1}
\end{align*}
Thus $\mb E Y_n \geq g_n$, where $g_n$ is as in \eqref{g}. We next bound the second moment of $Y_n$. Set
\[
    \bm v_n:=\bm\omega_\star,
    \qquad
    \bm v_{t-1}:=\mfH_t\bm v_t,
    \qquad
    t=n,n-1,\ldots,1.
\]
Then
\[
    \bm v_0
    =
    \mfB_n^\top\bm\omega_\star,
    \qquad
    Y_n=\|\bm v_0\|_\wedge^2.
\]
Moreover, $\bm v_t$ depends only on
$\mfH_{t+1},\ldots,\mfH_n$ and is therefore independent of
$\mfH_t$. By Cauchy--Schwarz,
\begin{align*}
    \|\mfH_t\bm v_t\|_\wedge^4
    &=
    \left(
        \bm v_t^\top
        \mfH_t^2
        \bm v_t
    \right)^2
    \leq
    \|\bm v_t\|_\wedge^2 \cdot 
    \bm v_t^\top \,
    \mfH_t^4 \,
    \bm v_t.
\end{align*}
Conditioning on $\bm v_t$ and iterating gives
\begin{equation}
    \mb E Y_n^2
    \leq
    \prod_{t=1}^n
    \left\|
        \mb E\left[
            \mfH_t^4
        \right]
    \right\|_{\rm op}.
    \label{eq:Yn-second-moment-product}
\end{equation}
Using \eqref{eq:rank-one-powers-exterior},
\begin{align*}
    \mb E\left[
        \mfH_t^4
    \right]
    &=
    \bm I
    +
    4\eta_t\mc L_{\bm\Sigma}
    +
    6\eta_t^2
    \mb E\left[
        \mc L_{\bm A_t}^{\,2}
    \right]
    +
    4\eta_t^3
    \mb E\left[
        \mc L_{\bm A_t}^{\,3}
    \right]
    +
    \eta_t^4
    \mb E\left[
        \mc L_{\bm A_t}^{\,4}
    \right].
\end{align*}
Since
\[
    \mb E\left[
        \mc L_{\bm A_t}^{\,\ell}
    \right]
    \preceq
    L^{\ell-1}
    \mc L_{\bm\Sigma},
    \qquad
    \ell=2,3,4,
\]
and $\eta_tL\leq1/4$, it follows that
\begin{align*}
    \left\|
         \mb E\left[
            \mfH_t^4
        \right]  
    \right\|_{\rm op}
    &\leq
    1
    +
    4\eta_t\mu_1
    +
    6\eta_t^2L\mu_1
    +
    4\eta_t^3L^2\mu_1
    +
    \eta_t^4L^3\mu_1
    \leq
    \exp\left(
        4\eta_t\mu_1
        +
        8L\mu_1\eta_t^2
    \right).
\end{align*}
Combining this with \eqref{eq:Yn-second-moment-product} yields
\[
    \mb E Y_n^2
    \leq
    \exp\left\{
        4\mu_1\sum_{t=1}^n\eta_t
        +
        8L\mu_1\sum_{t=1}^n\eta_t^2
    \right\}.
\]
On the other hand,
\[
    g_n^2
    =
    \prod_{t=1}^n
    \left(
        1+\eta_t\mu_1
    \right)^4
    \geq
    \exp\left\{
        4\mu_1\sum_{t=1}^n\eta_t
        -
        2L\mu_1\sum_{t=1}^n\eta_t^2
    \right\}.
\]
Since $\mb E Y_n\geq g_n$, we conclude that
\begin{equation}
    \frac{
        \operatorname{Var}(Y_n)
    }{
        \left(
            \mb E Y_n
        \right)^2
    }
    \leq
    \exp\left\{
        10L\mu_1
        \sum_{t=1}^n\eta_t^2
    \right\}
    -1.
    \label{eq:Yn-relative-variance}
\end{equation}
By choosing the universal constant $c$ in
\eqref{step size cond} sufficiently small,
\[
    \exp\left\{
        10L\mu_1
        \sum_{t=1}^n\eta_t^2
    \right\}
    -1
    \leq
    \frac{\delta}{12}.
\]
Chebyshev's inequality and $\mb E Y_n\geq g_n$ therefore imply
\begin{align}
    \mb P\left(
        Y_n<\frac{g_n}{2}
    \right)
    &\leq
    \mb P\left(
        |Y_n-\mb E Y_n|
        >
        \frac{\mb E Y_n}{2}
    \right)
    \leq
    4
    \frac{
        \operatorname{Var}(Y_n)
    }{
        (\mb E Y_n)^2
    }
    \leq
    \frac{\delta}{3}.
    \label{eq:Yn-high-probability}
\end{align}
Thus, with probability at least $1-\delta/3$, $Y_n\geq\frac{g_n}{2}$. Combining \eqref{conditional bound-1},
\eqref{eq:Xn-high-probability}, and \eqref{eq:Yn-high-probability}, and using a union bound, we obtain,
with probability at least $1-\delta$,
\begin{align*}
    \left\|
        \sin\bm\Theta
        \left(
            \bm Q_n,\bm U_k^\star
        \right)
    \right\|_{\rm F}^2
    \leq
    \tan^2
    \left(
        \bm\omega_n,\bm\omega_\star
    \right)
    \leq
    D_k \cdot r_k(\delta) \cdot 
    \frac{X_n}{Y_n}
    \leq
    \frac{6r_k(\delta)}{\delta} \cdot 
    \frac{m_n}{g_n}.
\end{align*}
The desired bound \eqref{k-PCA-F-bound} now follows from
\eqref{eq:expected-orthogonal-mass-simplified}.
\hfill$\square$

\section{Proof of Corollary \ref{cor:two-phase-kpca}}

\noindent\textbf{Proof of Corollary~\ref{cor:two-phase-kpca}.}
Put $N:=n-m$ and denote the constant learning rate in the first phase by
\[
\eta_{\rm w}
:=
\sqrt{
\frac{c_0\delta}{L\mu_1m}
}.
\]
We first verify the conditions of
Theorem~\ref{k-PCA convergence}, and then bound separately the
initialization term and the variance term in
\eqref{k-PCA-F-bound}.

\medskip
\noindent
\underline{\it Step 1: Verification of the learning-rate conditions.}
We begin with the first condition in \eqref{step size cond}. Write
\begin{align*}
    \max_{1\le t\le n}\eta_t = \max \la \eta_w, \max_{1\le s\le N} \eta_{m+s}  \ra \leq \max \la \sqrt{
\frac{c_0\delta}{L\mu_1m}
}, \, \max_{1\le s\le N} \frac{6}{\Delta_k(\beta+s)} \ra
\end{align*}
By definition,
\[
m \ge C_0
\frac{L\mu_1}{\delta\Delta_k^2} \cdot
k^2\log^2\left(\frac{ep}{k}\right),
\]
and thus 
\begin{align*}
\eta_{\rm w}
=
\sqrt{
\frac{c_0\delta}{L\mu_1m}
}
&\le
\sqrt{\frac{c_0}{C_0}}\cdot 
\frac{
\delta\Delta_k
}{
L\mu_1k\log(ep/k)
}
\leq \frac{1}{4L}
\end{align*}
if $C_0$ is sufficiently large relative to $c_0$ since $\Delta_k\le\lambda_k\le\mu_1$, $\delta \leq 1/4$
and $k\log(ep/k)\ge1$.

For the second term, writing $t=m+s$ with $1\le s\le N$, we have
\[
\eta_{m+s}
=
\frac{6}{\Delta_k(\beta+s)}
\le
\frac{6}{\Delta_k\beta}.
\]
Since $\beta \ge C_0\frac{L}{\Delta_k}$, it follows that
\[
\max_{1\le s\le N}\eta_{m+s}
\le
\frac{6}{C_0L}
\le
\frac{1}{4L}
\]
if $C_0$ is sufficiently large. Hence
\begin{equation}
\label{eq:cor-max-step}
\max_{1\le t\le n}\eta_t
\le
\frac{1}{4L}.
\end{equation}
This verifies the first condition in \eqref{step size cond}. We next control the sum of the squared learning rates. Write
\begin{align*}
L\mu_1\sum_{t=1}^n\eta_t^2 &= L\mu_1 \cdot \sum_{t=1}^m \eta_t^2 + L\mu_1 \cdot \sum_{s=1}^N \eta^2_{m+s} \\
&\leq L\mu_1 \cdot \frac{c_0 \delta}{L \mu_1} + L\mu_1 \cdot \frac{36}{\Delta_k^2} \sum_{s=1}^N\frac{1}{(\beta+s)^2} \\
&\leq c_0 \delta + L\mu_1 \cdot \frac{36}{\Delta_k^2}
\int_{\beta}^{\infty}\frac{dx}{x^2} = c_0 \delta + \frac{36L \mu_1}{\Delta_k^2 \, \beta}.
 \end{align*}
Moreover, $\beta \ge C_0 \cdot (L\mu_1)/({\delta\Delta_k^2})$,
and therefore
\[
L\mu_1\sum_{t=1}^n\eta_t^2 \le \left( c_0+\frac{36}{C_0}\right)\delta.
\]
The last expression involving $c_0$ and $C_0$ can be made arbitrarily small by choosing $c_0$ small enough and $C_0$ large enough. Together with \eqref{eq:cor-max-step}, this verifies all the learning-rate
conditions in \eqref{step size cond}.

\medskip
\noindent
\underline{\it Step 2: Error bound via Theorem \ref{k-PCA convergence}.}
By taking $C_1$
sufficiently large, condition
\eqref{eq:two-phase-kpca-sample-size} implies
\begin{equation}
\label{eq:cor-N-properties}
m\le\frac n2,
\qquad
N=n-m\ge\frac n2,
\qquad
N\ge C \cdot \beta^{6/5}
\end{equation}
for a sufficiently large universal constant $C>0$.

By Theorem \ref{k-PCA convergence}, we have, with probability at least
$1-\delta$,
\begin{align*}
\left\|
\sin\bm\Theta
\left(
\bm Q_n,\bm U_k^\star
\right)
\right\|_{\rm F}^2
&\lesssim
\frac{r_k(\delta)}{\delta} \cdot 
\left[ \binom pk \cdot \exp\lb -2\Delta_k\sum_{t=1}^n\eta_t \rb 
+
\sigma_k^2 \cdot \sum_{i=1}^n \eta_i^2 \cdot \exp\lb -2\Delta_k\sum_{j=i+1}^n\eta_j\rb
\right].
\end{align*}

To bound the first term inside the bracket, write 
\begin{align} \label{first term}
     \binom pk \cdot \exp\lb -2\Delta_k\sum_{t=1}^n\eta_t \rb &= \exp \la \log \binom pk   -2\Delta_k\sum_{t=1}^n\eta_t \ra \nonumber \\
     &\leq \exp \la  k \cdot \log \lb ep/k \rb - 2\Delta_k\cdot m \eta_w - 2\Delta_k \sum_{s=1}^N\eta_{m+s}  \ra \nonumber \\
     &\leq \exp \la k\cdot \log \lb \frac{ep}{k} \rb - 2\Delta_k \cdot \sqrt{\frac{c_0 \delta m}{L \mu_1}} - 12 \sum_{s=1}^N \frac{1}{\beta+s} \ra \nonumber \\
     &\leq \exp \la  k\cdot \log \lb \frac{ep}{k} \rb - 2\sqrt{c_0C_0}\cdot 
k\log\left(\frac{ep}{k}\right) - 12\log \lb \frac{\beta+N+1}{\beta+1} \rb \ra \nonumber \\
&\leq \left( \frac{\beta+1}{\beta+N+1} \right)^{12},
\end{align}
where the last bound follows if $c_0$ and $C_0$ are chosen such that $c_0C_0 \geq 1$.
By \eqref{eq:cor-N-properties},
\[
\beta^{12}
=
\left(\beta^{6/5}\right)^{10}
\lesssim
N^{10}.
\]
Hence
\begin{equation}
\label{eq:cor-initialization-bound}
\binom pk
\exp\left\{
-2\Delta_k\sum_{t=1}^n\eta_t
\right\}
\lesssim
\frac{1}{n^2}.
\end{equation}

We next bound the convolution term. Split the convolution term as
\[
\sum_{i=1}^n
\eta_i^2
\exp\left\{
-2\Delta_k\sum_{j=i+1}^n\eta_j
\right\}
= \underbrace{\sum_{i=1}^m \eta_i^2 \exp\left\{ -2\Delta_k\sum_{j=i+1}^n\eta_j\right\} }_{V_{\rm I}}
+ \underbrace{\sum_{s=1}^N \eta_{m+s}^2 \cdot  \exp\left\{ -2\Delta_k\sum_{j=s+1}^N \eta_{m+j} \right\}}_{ V_{\rm II}}
\]
Let us bound $V_{\rm I}$. The exponential term is largest when $i=m$, hence
\[
\begin{aligned}
V_{\rm I}
&\leq
\exp\left\{-2\Delta_k\sum_{j=m+1}^n\eta_j\right\}
\cdot \sum_{i=1}^m \eta_i^2 \\
&= O\lb n^{-2}\rb \cdot \frac{c_0\delta}{L\mu_1}
\lesssim \frac{\delta}{L\mu_1 n^2}
\leq \frac{\delta}{\Delta_k^2 n^2},
\end{aligned}
\]
where the $O(n^{-2})$ bound follows from the same argument as in
\eqref{first term}. The last two inequalities use
$L\geq\lambda_1\geq\Delta_k$,
$\mu_1\geq\lambda_k\geq\Delta_k$, and $\delta<1$.

To bound $V_{\rm II}$, writing $i=m+s$, we have
\begin{align*}
V_{\rm II}
&=
\frac{36}{\Delta_k^2}
\sum_{s=1}^N
\frac{1}{(\beta+s)^2}
\exp\left\{
-12
\sum_{\ell=s+1}^N
\frac{1}{\beta+\ell}
\right\}.
\end{align*}
For every $1\le s\le N$,
\begin{align*}
\sum_{\ell=s+1}^N\frac{1}{\beta+\ell}
&\ge
\int_{s+1}^{N+1}\frac{dx}{\beta+x}
=
\log\left(
\frac{\beta+N+1}{\beta+s+1}
\right).
\end{align*}
Hence
\begin{align*}
V_{\rm II}
&\le
\frac{36}{
\Delta_k^2(\beta+N+1)^{12}
}
\sum_{s=1}^N
\frac{
(\beta+s+1)^{12}
}{
(\beta+s)^2
}.
\end{align*}
Because $\beta+s\ge1$,
\[
\frac{
(\beta+s+1)^{12}
}{
(\beta+s)^2
}
\le
2^{12}(\beta+s)^{10}.
\]
It follows that
\begin{align*}
V_{\rm II}
&\le
\frac{C}{
\Delta_k^2(\beta+N+1)^{12}
}
\sum_{s=1}^N
(\beta+s)^{10}
\le
\frac{
CN(\beta+N)^{10}
}{
\Delta_k^2(\beta+N+1)^{12}
}
\le
\frac{C}{
\Delta_k^2(\beta+N+1)
}.
\end{align*}
Using $N\ge n/2$ from \eqref{eq:cor-N-properties}, we conclude that
\begin{equation*}
V_{\rm II}
\lesssim
\frac{1}{\Delta_k^2N}
\lesssim
\frac{1}{\Delta_k^2 \, n}.
\end{equation*}
Thus, the bounds on $V_{\rm I}$ and $V_{\rm II}$ give
\begin{equation}
\label{eq:cor-total-variance}
\sum_{i=1}^n
\eta_i^2
\exp\left\{
-2\Delta_k\sum_{j=i+1}^n\eta_j
\right\}
\lesssim
\frac{1}{\Delta_k^2 \, n}.
\end{equation}
Consequently, \eqref{eq:cor-initialization-bound} and
\eqref{eq:cor-total-variance} give, with probability at least
$1-\delta$,
\begin{align*}
\left\|
\sin\bm\Theta
\left(
\bm Q_n,\bm U_k^\star
\right)
\right\|_{\rm F}^2
&\lesssim
\frac{r_k(\delta)}{\delta}
\left[
\frac{1}{n^2}
+
\sigma_k^2
\frac{1}{\Delta_k^2n}
\right].
\end{align*}
This is precisely \eqref{eq:two-phase-kpca-rate} and completes the
proof.
\hfill$\square$

\section{Proof of Theorem \ref{thm:kpca-warm-start}}

For each $r \in [R]$, let $\bm W_0^{(r)},\bm W_1^{(r)},\ldots,\bm W_{m_0}^{(r)}$ denote the Oja iterates generated from the $r$-th data block, and write
\[
\bm W_t^{(r)}
=
[\bm w_{1,t}^{(r)},\ldots,\bm w_{k,t}^{(r)}].
\]
Apply Theorem~\ref{k-PCA convergence} directly to the original $k$-PCA iterates $\bm W_t^{(r)}$, with run length
$m_0$, constant learning rate $\eta_{\rm w}$, and confidence parameter $1/4$, to deduce that, for every $r=1,\ldots,R$,
\begin{align}
\mb P\Bigg\{ \left\| \sin\Theta \left(
\bm W_{m_0}^{(r)},
\bm U_k^\star
\right)
\right\|_{\rm F}^2
\nonumber 
\lesssim
r_k \lb  \frac 14 \rb \cdot 
\left[
\binom pk \cdot 
\exp\la
-2\Delta_km_0\eta_{\rm w}
\ra
+
\sigma_k^2\eta_{\rm w}^2
\cdot 
\sum_{\ell=0}^{m_0-1}
\exp\left\{
-2\ell \cdot \Delta_k\eta_{\rm w}
\right\}
\right]
\Bigg\}
\ge\frac34,
\label{eq:warm-single-run}
\end{align}
where $r_k(.)$ is as in Theorem \ref{k-PCA convergence}.

Since
$
\Delta_k\eta_{\rm w}
\le
L\eta_{\rm w}
\le 1/4
$,
we have
$
1-\exp\{-2\Delta_k\eta_{\rm w}\}
\ge
\Delta_k\eta_{\rm w}.
$
Therefore,
\begin{align*}
\eta_{\rm w}^2
\sum_{\ell=0}^{m_0-1}
\exp\left\{
-2\Delta_k\eta_{\rm w}\ell
\right\}
&\le
\frac{\eta_{\rm w}^2}
{1-\exp\{-2\Delta_k\eta_{\rm w}\}}
\le
\frac{\eta_{\rm w}}{\Delta_k}.
\end{align*}
It follows that
\begin{equation*}
\mb P\left\{
d^2\left(
\bm W_{m_0}^{(r)},\bm U_k^\star
\right)
\le a_n
\right\}
\ge\frac34,
\end{equation*}
where
\[
a_n
:=
C \cdot r_k(1/4) \cdot 
\left[
\binom pk
\exp\left\{
-2\Delta_km_0\eta_{\rm w}
\right\}
+
\frac{\sigma_k^2\eta_{\rm w}}{\Delta_k}
\right]
; \qquad 
d(\bm W,\bm U)
=
\|\sin\Theta(\bm W,\bm U)\|_{\rm F}.
\]
Proposition~\ref{prop:boosting} therefore gives
\[
\mb P\left\{
d^2\left(
\widehat{\bm U}_{\rm warm},
\bm U_k^\star
\right)
\le9a_n
\right\}
\ge1-\delta,
\]
since
$R =\left\lceil24\log\frac1\delta\right\rceil$.

Consequently,
\begin{align}
\left\|
\sin\Theta
\left(
\widehat{\bm U}_{\rm warm},
\bm U_k^\star
\right)
\right\|_{\rm F}^2
\le
C r_k(1/4)
\left[
\binom pk
\exp\left\{
-2\Delta_km_0\eta_{\rm w}
\right\}
+
\frac{\sigma_k^2\eta_{\rm w}}{\Delta_k}
\right]
\label{eq:warm-boosted-exact}
\end{align}
with probability at least $1-\delta$.

Finally, because $n\ge4R$ and $\eta_{\rm w}=c_\alpha n^{-\alpha}$, we have
\[
m_0\ge\frac{n}{4R}, \qquad R\le C\log\frac1\delta, \qquad 2\Delta_km_0\eta_{\rm w}
\ge \frac{ \bar c_\alpha\Delta_kn^{1-\alpha}}{\log(1/\delta)},
\qquad 
\frac{\sigma_k^2\eta_{\rm w}}{\Delta_k}
=
\frac{
c_\alpha\sigma_k^2
}{
\Delta_kn^\alpha
}.
\]
Substituting these two bounds into \eqref{eq:warm-boosted-exact} proves  \eqref{eq:kpca-warm-exact}. $\hfill$ $\square$

\section{Proof of Theorem \ref{thm:kpca-phase-two}}

Throughout the proof, $C>0$ denotes a universal constant and
$C_\alpha>0$ denotes a constant depending only on $\alpha$. Their
values may change from line to line.

Set
\[
\bm U_\star:=\bm U_k^\star,
\qquad
\bm U_\perp:=
[\bm u_{k+1},\ldots,\bm u_p],
\]
and
\begin{align} \label{eq:Lambda-diag}
\bm\Lambda_\star
:=
\operatorname{diag}
(\lambda_1,\ldots,\lambda_k),
\qquad
\bm\Lambda_\perp
:=
\operatorname{diag}
(\lambda_{k+1},\ldots,\lambda_p).
\end{align}
Recall that
\[
m=\left\lfloor\frac n2\right\rfloor,
\qquad
N=n-m,
\]
and hence
\begin{equation}
\label{eq:thm3-mN}
m\asymp n,
\qquad
N\asymp n.
\end{equation}

Let $\mc F_m$ contain the observations and all the randomizations
used in the warm-start phase, and, for $t\ge m+1$, define
\[
\mc F_t
:=
\sigma\left(
\mc F_m,\bm X_{m+1},\ldots,\bm X_t
\right).
\]
The observations used during Phase II are independent of $\mc F_m$.

\medskip
\noindent
\underline{\it Step 1: Verifying the warm start.}
Denote by $(\rm i)$--$(\rm iv)$ the four terms inside the maximum in
\eqref{eq:kpca-phase-two-sample-size}, in the order listed. Phase I is run at
confidence level $\delta/4$, so that
\[
R = \left\lceil 24\log\frac4\delta \right\rceil \lesssim \ell_\delta,
\qquad
m_0 = \left\lfloor\frac{m}{R}\right\rfloor \le n,
\qquad
\eta_{\rm w} = c_\alpha n^{-\alpha}.
\]
Note that $(\rm iv)$ reads $n^\alpha \ge C_\alpha^\alpha M_{n,\delta}$, and
$C_\alpha^\alpha$ can be made arbitrarily large by increasing $C_\alpha$; we use
this repeatedly below without further comment.

We first verify the hypotheses \eqref{eq:kpca-warm-basic-conditions} of
Theorem~\ref{thm:kpca-warm-start}. Term $(\rm i)$ gives $n\ge 4R$. Term
$(\rm iv)$, together with $M_{n,\delta}\ge L$, gives
$\eta_{\rm w} = c_\alpha n^{-\alpha}\le 1/(4L)$. Term $(\rm iii)$ gives
\[
L\mu_1 m_0 \cdot \eta_{\rm w}^2
\le
c_\alpha^2 \, L\mu_1 \cdot  n^{1-2\alpha}
\le
c_0 .
\]
Theorem~\ref{thm:kpca-warm-start}, applied with confidence parameter $\delta/4$,
therefore yields an event $\mc E_{\rm w}$ with
$\mb P(\mc E_{\rm w})\ge 1-\delta/4$ on which \eqref{eq:kpca-warm-exact} holds.

It remains to simplify the bound in \eqref{eq:kpca-warm-exact}. Term $(\rm ii)$ states
\[
\frac{\Delta_k n^{1-\alpha}}{\log(4/\delta)}
\ge
C_\alpha \left\{ k\log\left(\frac{ep}{k}\right) + \log(en) \right\},
\]
so, since $\log\binom pk \le k\log(ep/k)$, the initialization term in
\eqref{eq:kpca-warm-exact} obeys
\[
\binom pk
\exp\left\{ -\frac{c_\alpha \Delta_k n^{1-\alpha}}{\log(4/\delta)} \right\}
\le
n^{-\alpha}
\]
after increasing $C_\alpha$. Hence, on $\mc E_{\rm w}$,
\begin{align}
\left\| \sin\bm\Theta \lb \bm Q_m,\bm U_\star \rb \right\|_{\rm F}^2
\le
C_\alpha \, r_k(1/4) \lb 1+\frac{\sigma_k^2}{\Delta_k} \rb n^{-\alpha}
\le
C_\alpha \, M_{n,\delta} \, n^{-\alpha}
=: r_n ,
\label{eq:thm3-warm-bound}
\end{align}
the last inequality by \eqref{eq:phase-two-Mn}. Finally, applying $(\rm iv)$
once more, with its constant now taken large relative to the one appearing in
$r_n$, gives
\begin{equation}
\label{eq:thm3-local-smallness}
r_n\le c_0,
\qquad
L m^{-\alpha}\le\frac14 .
\end{equation}

\medskip
\noindent
\underline{\it Step 2: Maximal inequality via the Pl\"ucker embedding.}
Let $\bm\omega_t := \bm q_{1,t}\wedge\cdots\wedge\bm q_{k,t}$ denote the
Pl\"ucker image of $\bm Q_t = [\bm q_{1,t},\ldots,\bm q_{k,t}]$, put
$\bm\omega_\star := \bm u_1\wedge\cdots\wedge\bm u_k$, and let
\[
S_t^\wedge
:=
1-\left| \left\langle \bm\omega_t,\bm\omega_\star \right\rangle_\wedge \right|^2 .
\]
By Theorem~\ref{thm: Plucker embedding} the iterates $\bm\omega_t$ obey the
embedded rank-one recursion \eqref{embbeded Oja}, and $\eta_t L \le 1/4$ for
$t \ge m+1$ by \eqref{eq:thm3-local-smallness}. Lemma~\ref{lem:one-step-exterior}
therefore applies and yields the one-step inequality
\begin{align}
S_t^\wedge
\le{}&
\left(
1-2\Delta_k\eta_t
\right) \cdot S_{t-1}^\wedge
+
2\Delta_k\eta_t \cdot 
\left(S_{t-1}^\wedge\right)^2
+
Q_t
+
CL^2 \cdot \eta_t^2.
\label{eq:thm3-exterior-one-step}
\end{align}
where $\la Q_t \ra$ satisfies  \eqref{eq:one-step-Q}.

The goal of this step is to convert \eqref{eq:thm3-exterior-one-step} into a
bound on $S_t^\wedge$ that holds \emph{simultaneously} for all $m \le t \le n$.
Throughout we assume the three conditions
\begin{align}
\Delta_k m^{1-\alpha} \ge 8\alpha,
\qquad
\Delta_k m^{-\alpha} \le \frac12,
\qquad
H_n m^{-\alpha} \le \frac12,
\label{eq:thm3-max-conditions}
\end{align}
where
\[
H_n := C_\alpha M_{n,\delta},
\qquad
b_t := H_n \eta_t = H_n t^{-\alpha},
\qquad m \le t \le n,
\]
and $C_\alpha$ is chosen large enough that \eqref{eq:thm3-warm-bound} gives
\begin{align} \label{eq:thm3-initial-quarter}
S_m^\wedge \le \frac14 b_m .
\end{align}
The point of \eqref{eq:thm3-max-conditions} is to make the maximal inequality for the stopped process below valid. We assume these conditions hold, and then verify \eqref{eq:thm3-max-conditions} at the end of this step. Define the
stopping time
\[
\tau := \inf \la m \le t \le n : S_t^\wedge > b_t \ra,
\]
with the convention $\tau = \infty$ if the set is empty, and the stopped
processes
\[
\widetilde S_t := S_t^\wedge \cdot \bm 1_{\la t \le \tau \ra},
\qquad
\widetilde Q_t := Q_t \cdot \bm 1_{\la t \le \tau \ra}.
\]

\medskip
\noindent
\underline{\emph{Step 2a: the stopped recursion.}}
The third condition in \eqref{eq:thm3-max-conditions} gives $b_t \le 1/2$ for
all $t \in [m,n]$. Since $\la t \le \tau \ra \subset \la t-1 \le \tau \ra$ and
$\la t \le \tau \ra \subset \la S_{t-1}^\wedge \le b_{t-1} \ra$, multiplying
\eqref{eq:thm3-exterior-one-step} by $ \bm 1_{\la t \le \tau \ra}$ yields
 the one-step recursion 
\begin{align}
\widetilde S_t
&\le
\lb 1-2\Delta_k\eta_t \rb S_{t-1}^\wedge \bm 1_{\la t-1 \le \tau \ra}
+ 2\Delta_k\eta_t b_{t-1} S_{t-1}^\wedge \bm 1_{\la t-1 \le \tau \ra}
+ \widetilde Q_t + CL^2\eta_t^2
\nonumber \\
&\le
\lb 1-\Delta_k\eta_t \rb \widetilde S_{t-1} + \widetilde Q_t + CL^2\eta_t^2 .
\label{eq:thm3-stopped-recursion}
\end{align}
Note that \eqref{eq:thm3-stopped-recursion} holds for all $t \ge m+1$. Put
\[
a_t := 1-\Delta_k\eta_t,
\qquad
G_{i,t} := \prod_{j=i+1}^t a_j,
\qquad
G_{t,t} := 1 .
\]
The second condition in \eqref{eq:thm3-max-conditions} gives
$\Delta_k \eta_t \le 1/2$, hence $a_t \in [1/2,1]$ for all $t \in [m,n]$.
Iterating \eqref{eq:thm3-stopped-recursion} gives
\begin{align}
\widetilde S_t
\le
G_{m,t} \cdot S_m^\wedge
+
\sum_{i=m+1}^t G_{i,t} \cdot \widetilde Q_i
+
CL^2 \sum_{i=m+1}^t G_{i,t} \cdot \eta_i^2 .
\label{eq:thm3-exterior-unrolled}
\end{align}
There are three terms in the right-hand side of \eqref{eq:thm3-exterior-unrolled}: the first and third terms are deterministic conditionally on $S_m^\wedge$, and the second term is the terminal value of a martingale in the summation index for each fixed terminal time $t$. We will bound them separately in what follows.

\medskip
\noindent
\underline{\emph{Step 2b: bounding the deterministic terms in \eqref{eq:thm3-exterior-unrolled}.}} Only the first term and the third term are deterministic.

Consider the first term. Since $\alpha\in(1/2,1)$ and $\Delta_k m^{1-\alpha}\geq8\alpha$ by
\eqref{eq:thm3-max-conditions}, for every $j\geq m+1$,
\[
j^\alpha-(j-1)^\alpha
\leq \alpha m^{\alpha-1}
\leq \Delta_k.
\]
Consequently,
\[
0\leq a_j=1-\Delta_k \cdot j^{-\alpha}
\leq \left(\frac{j-1}{j}\right)^\alpha,
\]
and multiplying over $j=m+1,\ldots,t$ gives
\[
G_{m,t}
=\prod_{j=m+1}^t a_j
\leq \left(\frac{m}{t}\right)^\alpha.
\]
Using \eqref{eq:thm3-initial-quarter}, we therefore obtain
\begin{align}
G_{m,t}S_m^\wedge
&\leq
\left(\frac{m}{t}\right)^\alpha
\cdot\frac14 H_n m^{-\alpha}
=
\frac14 H_n t^{-\alpha}
=
\frac14 b_t.
\label{eq:thm3-kernel-init}
\end{align}
For the third term in \eqref{eq:thm3-exterior-unrolled}, we will show that
\begin{align}
\sum_{i=m+1}^t G_{i,t}\eta_i^2 \le \frac{2\eta_t}{\Delta_k} .
\label{eq:thm3-kernel-one}
\end{align}
Indeed, by using
$$G_{i,t}-G_{i-1,t} = \lb 1-a_i \rb G_{i,t} = \Delta_k \eta_i G_{i,t},$$
we get
$$G_{i,t}\eta_i^2 = \Delta_k^{-1}\eta_i \lb G_{i,t}-G_{i-1,t}\rb.$$
Therefore,
\[
\sum_{i=m+1}^t \eta_i \lb G_{i,t}-G_{i-1,t} \rb
=
\eta_t - \eta_{m+1}G_{m,t} + \sum_{i=m+1}^{t-1} \lb \eta_i-\eta_{i+1} \rb G_{i,t}
\le
\eta_t + \sum_{i=m+1}^{t-1} \lb \eta_i-\eta_{i+1} \rb G_{i,t} .
\]
Since $\eta_i - \eta_{i+1} \le \alpha i^{-\alpha-1}$ and, by the first
condition in \eqref{eq:thm3-max-conditions}, $\Delta_k i^{1-\alpha} \ge 8\alpha$
for $i \ge m$, we have $\eta_i-\eta_{i+1} \le \frac{\Delta_k}{2}\eta_i^2$.
Substituting and dividing by $\Delta_k$,
\[
\sum_{i=m+1}^t G_{i,t}\eta_i^2
\le
\frac{\eta_t}{\Delta_k} + \frac12 \sum_{i=m+1}^t G_{i,t}\eta_i^2 ,
\]
which is \eqref{eq:thm3-kernel-one}. Consequently,
\begin{align}
CL^2 \sum_{i=m+1}^t G_{i,t}\eta_i^2
\le
\frac{2CL^2 \eta_t}{\Delta_k}
\le
\frac14 b_t ,
\label{eq:thm3-kernel-var}
\end{align}
because $M_{n,\delta} \ge L^2\ell_{n,\delta}/\Delta_k \ge L^2/\Delta_k$ by
\eqref{eq:phase-two-Mn}, so $H_n \ge 8CL^2/\Delta_k$ after increasing
$C_\alpha$.

\medskip
\noindent
\underline{\emph{Step 2c: bounding the martingale term in \eqref{eq:thm3-exterior-unrolled}.}}
Since $\la t \le \tau \ra \in \mc F_{t-1}$, we have
$\mb E[\widetilde Q_t \mid \mc F_{t-1}] = 0$. Moreover, on
$\la i \le \tau \ra$ we have $S_{i-1}^\wedge \le b_{i-1} = H_n \eta_{i-1}$, so
\eqref{eq:one-step-Q} gives
\begin{align} \label{eq:thm3-stopped-Q-size}
\left| \widetilde Q_i \right|
\le
4L\eta_i \sqrt{H_n \eta_{i-1}} .
\end{align}
Fix $t \in [m+1,n]$ and set
\[
M_{s,t} := \sum_{i=m+1}^{s \wedge t} G_{i,t}\widetilde Q_i,
\qquad m \le s \le t,
\]
which is a martingale in $s$ with increments bounded, by
\eqref{eq:thm3-stopped-Q-size}, by
$4L\sqrt{H_n}\, G_{i,t}\eta_i\sqrt{\eta_{i-1}}$. We claim
\begin{align}
\sum_{i=m+1}^t G_{i,t}^2 \eta_i^2 \eta_{i-1}
\le
\frac{4\eta_t^2}{\Delta_k} .
\label{eq:thm3-kernel-two}
\end{align}
Indeed, since $\Delta_k\eta_i \le 1/2$,
\[
G_{i,t}^2 - G_{i-1,t}^2
=
\Delta_k\eta_i \lb 2-\Delta_k\eta_i \rb G_{i,t}^2
\ge
\Delta_k \eta_i G_{i,t}^2 .
\]
Putting $w_i := \eta_i\eta_{i-1}$, this gives
$G_{i,t}^2\eta_i^2\eta_{i-1} \le \Delta_k^{-1} w_i \lb G_{i,t}^2-G_{i-1,t}^2 \rb$,
and thus
\[
\sum_{i=m+1}^t G_{i,t}^2\eta_i^2\eta_{i-1}
\le
\frac{w_t}{\Delta_k}
+
\frac1{\Delta_k}\sum_{i=m+1}^{t-1} \lb w_i-w_{i+1} \rb G_{i,t}^2 .
\]
Since $w_i - w_{i+1} = \eta_i \lb \eta_{i-1}-\eta_{i+1} \rb
\le 2\alpha \eta_i (i-1)^{-\alpha-1}$ and $i-1 \ge i/2$ for $i \ge 2$, the
first condition in \eqref{eq:thm3-max-conditions} gives
$w_i - w_{i+1} \le \frac{\Delta_k}{2}\eta_i w_i$. Absorbing the resulting sum
into the left-hand side and using $w_t = \eta_t\eta_{t-1} \le 2\eta_t^2$ proves
\eqref{eq:thm3-kernel-two}. Therefore,
\[
\sum_{i=m+1}^t \left| M_{i,t}-M_{i-1,t} \right|^2
\le
16L^2 H_n \sum_{i=m+1}^t G_{i,t}^2\eta_i^2\eta_{i-1}
\le
\frac{64 L^2 H_n}{\Delta_k}\eta_t^2.
\]
Now, Azuma--Hoeffding's inequality gives, for every $x>0$,
\begin{align}
\mb P \lb
\left| M_{t,t} \right|
>
8\sqrt 2 \, L \eta_t \sqrt{\frac{H_n x}{\Delta_k}}
\;\Big|\; \mc F_m
\rb
\le
2e^{-x} .
\label{eq:thm3-azuma-fixed-t}
\end{align}

\medskip
\noindent
\underline{\emph{Step 2d: the maximal inequality.}}
Take $x := \log \lb 8n/\delta \rb \le \ell_{n,\delta}$ in
\eqref{eq:thm3-azuma-fixed-t} and apply a union bound over
$t = m+1,\ldots,n$. This produces an event $\mc E_{\rm max}$ with
\[
\mb P \lb \mc E_{\rm max}^c \mid \mc F_m \rb
\le
2n e^{-x}
=
\frac{\delta}{4},
\]
on which
\begin{align}
\left| \sum_{i=m+1}^t G_{i,t}\widetilde Q_i \right|
\le
8\sqrt 2\, L\eta_t \sqrt{\frac{H_n \ell_{n,\delta}}{\Delta_k}}
\le
\frac14 b_t
\label{eq:thm3-max-martingale}
\end{align}
simultaneously for all $m+1 \le t \le n$. The last inequality holds because
$M_{n,\delta} \ge L^2\ell_{n,\delta}/\Delta_k$, so that
$H_n \ge 2048\, L^2 \ell_{n,\delta}/\Delta_k$ after increasing $C_\alpha$.

Combining \eqref{eq:thm3-exterior-unrolled}, \eqref{eq:thm3-kernel-init},
\eqref{eq:thm3-kernel-var} and \eqref{eq:thm3-max-martingale}, we obtain, on
$\mc E_{\rm w} \cap \mc E_{\rm max}$,
\begin{align} \label{eq:thm3-stopped-three-quarters}
\widetilde S_t \le \frac34 b_t
\qquad \text{for all } m+1 \le t \le n,
\end{align}
while $\widetilde S_m = S_m^\wedge \le \frac14 b_m$ by
\eqref{eq:thm3-initial-quarter}. We claim $\tau > n$ on this event. Suppose
instead $\tau \le n$. By the definition of $\tau$,
$\widetilde S_\tau = S_\tau^\wedge > b_\tau$, whereas
\eqref{eq:thm3-stopped-three-quarters} gives
$\widetilde S_\tau \le \frac34 b_\tau < b_\tau$, a contradiction. Hence
$\tau > n$ and therefore $S_t^\wedge \le b_t$ for all $m \le t \le n$. Since
$m = \lfloor n/2 \rfloor \ge n/3$ for $n \ge 3$, we have
$t^{-\alpha} \le 3^\alpha n^{-\alpha}$ for $t \ge m$, and we conclude that
\begin{align} \label{eq:thm3-exterior-maximal}
\max_{m \le t \le n} S_t^\wedge
\le
3^\alpha H_n n^{-\alpha}
=
C_\alpha \cdot M_{n,\delta} \cdot n^{-\alpha}
=
r_n
\end{align}
on $\mc E_{\rm w} \cap \mc E_{\rm max}$.

\medskip
\noindent
\underline{\emph{Verification of \eqref{eq:thm3-max-conditions}.}}
Recall $m \asymp n$ by \eqref{eq:thm3-mN}. The second term in
\eqref{eq:kpca-phase-two-sample-size} gives
$\Delta_k n^{1-\alpha} \ge C_\alpha$, whence
$\Delta_k m^{1-\alpha} \ge 8\alpha$ after increasing $C_\alpha$. The last term
in \eqref{eq:kpca-phase-two-sample-size} gives
$n^\alpha \ge C_\alpha M_{n,\delta}$; since
$\Delta_k \le \lambda_k \le L \le M_{n,\delta}$, this yields
$\Delta_k m^{-\alpha} \le 1/2$, and since $H_n = C_\alpha M_{n,\delta}$ it also
yields $H_n m^{-\alpha} \le 1/2$, which is
\eqref{eq:thm3-max-conditions}.

\medskip
\noindent
\underline{\it Step 3: Constructing the recursion.}
Define
\begin{align} \label{eq:thm3-T-definition}
\bm T_t
:=
\bm U_\perp^\top\bm Q_t
\left(
\bm U_\star^\top\bm Q_t
\right)^{-1}
\in\mb R^{(p-k)\times k}.
\end{align}
Note that the singular values of \(\bm T_t\) are the tangents of the principal
angles. Thus on the event in \eqref{eq:thm3-exterior-maximal}, Proposition~\ref{prop:equivalance-error} gives
\begin{align}
\|\bm T_t\|_{\rm F}^2
=
\sum_{\ell=1}^k\tan^2\theta_{\ell,t}
&\le
\frac{
S_t^\wedge
}{
1-S_t^\wedge
}
\le
2r_n.
\label{eq:thm3-T-maximal}
\end{align}
Let us get a recursion for $\bm T_t$. From the QR decomposition, $ \bm B_t \bm Q_{t-1} = \bm Q_t \bm R_t$, with $\bm R_t$ being the upper triangular matrix with positive diagonal entries, so 
\[
\bm T_t = \bm U_\perp^\top\bm Q_t \left( \bm U_\star^\top\bm Q_t \right)^{-1}
=  \bm U_\perp^\top \bm Q_t \bm R_t  \lb \bm U_\star^\top \bm Q_t \bm R_t \rb^{-1}
= \bm U_\perp^\top \bm B_t \bm Q_{t-1} \cdot \lb \bm U_\star^\top \bm B_t \bm Q_{t-1} \rb^{-1}.
\]
For \(t\ge m+1\), put
\[
\bm a_t := \bm U_\star^\top\bm X_t,
\qquad
\bm b_t := \bm U_\perp^\top\bm X_t.
\]
Observe that $\bm X_t = \bm U_\star \bm a_t + \bm U_\perp \bm b_t$ and $\bm U_\perp^\top \bm Q_{t-1} = \bm T_{t-1} \lb \bm U_\star^\top \bm Q_{t-1} \rb $, so 
\[
\bm X_t^\top \bm Q_{t-1} = \bm a_t^\top \lb \bm U_\star^\top \bm Q_{t-1} \rb + \bm b_t^\top \bm T_{t-1} \lb \bm U_\star^\top \bm Q_{t-1} \rb
= \lb \bm a_t^\top + \bm b_t^\top \bm T_{t-1} \rb \cdot \lb \bm U_\star^\top \bm Q_{t-1} \rb.
\]
Thus
\[
\bm U_\star^\top \bm B_t \bm Q_{t-1} = \bm U_\star^\top \bm Q_{t-1} + \eta_t \cdot \bm U_\star^\top \bm X_t \bm X_t^\top \bm Q_{t-1}
= \left[ \bm I_k + \eta_t \cdot \bm a_t \lb \bm a_t^\top + \bm b_t^\top \bm T_{t-1} \rb  \right] \cdot \lb \bm U_\star^\top \bm Q_{t-1} \rb,
\]
and
\[
\bm U_\perp^\top \bm B_t \bm Q_{t-1} =  \bm U_\perp^\top \bm Q_{t-1} + \eta_t \cdot \bm U_\perp^\top \bm X_t \bm X_t^\top \bm Q_{t-1}
= \left[ \bm T_{t-1} + \eta_t \cdot  \bm b_t \lb \bm a_t^\top + \bm b_t^\top \bm T_{t-1} \rb  \right] \cdot \lb \bm U_\star^\top \bm Q_{t-1} \rb.
\]
Consequently,
\[
\bm T_t = \left[ \bm T_{t-1} + \eta_t \cdot  \bm b_t \lb \bm a_t^\top + \bm b_t^\top \bm T_{t-1} \rb  \right] \cdot  
\left[ \bm I_k + \eta_t \cdot \bm a_t \lb \bm a_t^\top + \bm b_t^\top \bm T_{t-1} \rb  \right]^{-1}.
\]
Now, recall the Sherman-Morrison formula (which is elementary and can be verified easily)
\[
\lb \bm I_k + \bm u \bm v^\top \rb^{-1} = \bm I_k - \frac{\bm u \bm v^\top}{1+ \bm v^\top \bm u}, 
\qquad \text{for all vectors $\bm u, \bm v \in \mb R^k$}.
\]
Applying the formula above for $\bm u \equiv \eta_ta_t$ and $\bm v \equiv \bm a_t+  \bm T_{t-1}^\top \bm b_t$, we obtain
\begin{align}
\bm T_t
=
\bm T_{t-1}
+
\frac{\eta_t}{d_t}
\left(
\bm b_t-\bm T_{t-1}\bm a_t
\right)
\left(
\bm a_t+\bm T_{t-1}^\top\bm b_t
\right)^\top,
\label{eq:thm3-exact-T}
\end{align}
where
\[
d_t
:=
1+
\eta_t
\bm a_t^\top
\left(
\bm a_t+\bm T_{t-1}^\top\bm b_t
\right).
\]
Recall $\bm \Lambda_\star$ and $\bm \Lambda_\perp$ in \eqref{eq:Lambda-diag}. For $\bm T \in \mb R^{(p-k) \times k}$, define 
\[
\mc G(\bm T)
:=
\bm T\bm\Lambda_\star
-
\bm\Lambda_\perp\bm T,
\qquad 
\bm\xi_t := \bm b_t\bm a_t^\top,
\qquad
\bm\zeta_t :=
\left( \bm b_t\bm b_t^\top-\bm\Lambda_\perp \right)\bm T_{t-1}
- \bm T_{t-1} \left( \bm a_t\bm a_t^\top-\bm\Lambda_\star \right).
\]
Here, we view $\mc G(.)$ as a linear operator $\mc G: \mb R^{(p-k)\times k} \to \mb R^{(p-k)\times k}$. Regarding $\mc G$ as a $(pk-k^2)\times(pk-k^2)$ square matrix, we have the entrywise equality
\[
[
\mc G(\bm T)
]_{j-k,i}
=
(\lambda_i-\lambda_j)T_{j-k,i},
\qquad
i\le k<j.
\]
Thus
\begin{equation}
\label{eq:thm3-G-inverse}
\|
\mc G^{-1}
\|_{{\rm F}\to{\rm F}}
\le
\Delta_k^{-1}.
\end{equation}
Expanding \eqref{eq:thm3-exact-T} gives
\begin{align*}
    \frac{\eta_t}{d_t} \cdot \left( \bm b_t-\bm T_{t-1}\bm a_t \right) \cdot \left( \bm a_t+\bm T_{t-1}^\top\bm b_t \right)^\top 
    = &  \frac{\eta_t}{d_t} \cdot \left[  \bm b_t \bm a_t^\top + \bm b_t \bm b_t^\top \cdot \bm T_{t-1} - \bm T_{t-1} \cdot \bm a_t \bm a_t^\top - \bm T_{t-1} \cdot \bm a_t \bm b_t^\top \cdot \bm T_{t-1}   \right] \\
    = & \frac{\eta_t}{d_t} \cdot    \bm b_t \bm a_t^\top + \frac{\eta_t}{d_t} \cdot \left[ \bm b_t \bm b_t^\top \cdot \bm T_{t-1} - \bm T_{t-1} \cdot \bm a_t \bm a_t^\top \right] -  \frac{\eta_t}{d_t} \cdot \left[ 
     \bm T_{t-1} \cdot \bm a_t \bm b_t^\top \cdot \bm T_{t-1}\right] \\
    =&  \frac{\eta_t}{d_t} \cdot \left[ \bm \zeta_t + \bm \xi_t - \mc G \lb \bm T_{t-1} \rb \right]  -  \frac{\eta_t}{d_t} \cdot \left[ 
     \bm T_{t-1} \cdot \bm a_t \bm b_t^\top \cdot \bm T_{t-1}\right].
\end{align*}
By the decomposition $1/d_t=1+\lb 1/d_t -1 \rb$, we can rewrite the above as
\begin{align}
\bm T_t
={}&
\bm T_{t-1}
-
\eta_t \cdot \mc G(\bm T_{t-1})
+
\eta_t\bm\xi_t
+
\eta_t\bm\zeta_t
+
\eta_t\bm{\mathcal R}_t,
\label{eq:thm3-linear-T}
\end{align}
where
\begin{align}
\bm{\mathcal R}_t
:={}&
\left(
\frac1{d_t}-1
\right)
\left[
\bm\xi_t
-
\mc G(\bm T_{t-1})
+
\bm\zeta_t
\right]
-
\frac1{d_t}
\bm T_{t-1}\bm a_t
\bm b_t^\top\bm T_{t-1}.
\label{eq:thm3-remainder}
\end{align}
The goal of the decomposition is to center the terms $\bm b_t \bm b_t^\top$ and $\bm a_t \bm a_t^\top$: since \(\bm U_\star\) and \(\bm U_\perp\) are eigenvector matrices, we have
\[
\mb E[
\bm b_t\bm a_t^\top]
=
\bm0,
\quad
\mb E[
\bm b_t\bm b_t^\top]
=
\bm\Lambda_\perp,
\quad
\mb E[
\bm a_t\bm a_t^\top]
=
\bm\Lambda_\star.
\]
Let us collect some properties of these terms in what follows. Observe that
\begin{equation}
\label{eq:thm3-xi-zeta-centered}
\mb E[
\bm\xi_t\mid\mc F_{t-1}]
=
\mb E[
\bm\zeta_t\mid\mc F_{t-1}]
=
\bm0,
\end{equation}
due to the centering effect. So $\la \bm \xi_t; m+1\leq t\leq n \ra$ and $\la \bm \zeta_t; m+1\leq t\leq n \ra$ are (matrix-valued) martingale difference sequences. 

To do the averaging, let us collect some useful bounds on these quantities. 
\begin{align}
    \max \la \|\bm a_t\|, \|\bm b_t\|\  \ra &\leq \sqrt{L}, \nonumber \\
    \|\bm\zeta_t\|_{\rm F} &\le 4L \cdot \|\bm T_{t-1}\|_{\rm F}, \label{eq:thm3-zeta-size} \\
    \mb E\left[ \|\bm\zeta_t\|_{\rm F}^2 \mid\mc F_{t-1} \right] &\le 4L\lambda_1 \|\bm T_{t-1}\|_{\rm F}^2, 
    \label{eq:thm3-zeta-variance} \\
\|\bm{\mathcal R}_t\|_{\rm F} &\le C\left[ L \cdot \|\bm T_{t-1}\|_{\rm F}^2 + L^2\cdot \eta_t \right] \label{eq:thm3-remainder-bound}.
\end{align}
The first bound follows from \(\|\bm X_t\|^2\le L\), and immediately yields \eqref{eq:thm3-zeta-size}. To get \eqref{eq:thm3-zeta-variance}, note that
\begin{align*}
\mb E\left[ \left( \bm b_t\bm b_t^\top-\bm\Lambda_\perp \right)^2 \right]
\preceq \mb E \left[  \|\bm b_t\|^2\bm b_t\bm b_t^\top \right] \preceq L\bm\Lambda_\perp \preceq L\lambda_1 \cdot \bm I,
\end{align*}
and the analogous inequality holds for \(\bm a_t\bm a_t^\top-\bm\Lambda_\star\).



On the event \eqref{eq:thm3-T-maximal}, we have 
\[
|d_t-1|
\le
\eta_tL
\left(
1+\|\bm T_{t-1}\|_{\rm F}
\right)
\le\frac12.
\]
Thus
\[
\left| \frac{1}{d_t} -1 \right| = \frac{|d_t-1|}{|d_t|} \leq 2|d_t-1| \leq 2\eta_tL \cdot \lb 1 + \| \bm T_{t-1} \|_{\rm F} \rb.
\]
Consequently, using this in \eqref{eq:thm3-remainder}, we obtain
\begin{align*}
    \|\bm{\mathcal R}_t\|_{\rm F} &\leq \left| \frac{1}{d_t} -1 \right| \cdot \left[ \| \bm \xi_t \|_{\rm F} +   \| \mc G(\bm T_{t-1}) \|_{\rm F} + \| \bm \zeta_t \|_{\rm F}   \right]
    + \frac{1}{d_t} \cdot \| \bm T_{t-1} \bm a_t
\bm b_t^\top\bm T_{t-1} \|_{\rm F} \\
&\lesssim \eta_tL \cdot L + \| \bm T_{t-1} \|_{\rm F} \cdot \| \bm a_t \| \cdot \| \bm b_t \| \cdot \| \bm T_{t-1} \|_{\rm F} \\
&\lesssim L^2 \cdot \eta_t + L \cdot  \| \bm T_{t-1} \|_{\rm F}^2,
\end{align*}
which is  \eqref{eq:thm3-remainder-bound}.

\medskip
\noindent
\underline{\it Step 4: Averaging decomposition.}
Rearranging \eqref{eq:thm3-linear-T}, applying
\(\mc G^{-1}\), and averaging over \(t=m+1,\ldots,n\) gives
\begin{align}
\frac1N
\sum_{t=m}^{n-1}\bm T_t
={}&
\underbrace{
\frac1N
\sum_{t=m+1}^n
\mc G^{-1}(\bm\xi_t)
}_{I}
+
\underbrace{
\frac1N
\sum_{t=m+1}^n
\mc G^{-1}(\bm\zeta_t)
}_{II}
\nonumber\\
&+
\underbrace{
\frac1N
\sum_{t=m+1}^n
\mc G^{-1}(\bm{\mathcal R}_t)
}_{III}
+
\underbrace{
\frac1N
\sum_{t=m+1}^n
\mc G^{-1}
\left(
\frac{
\bm T_{t-1}-\bm T_t
}{
\eta_t
}
\right)
}_{IV}.
\label{eq:thm3-PR-decomposition}
\end{align}

We bound the four terms above separately.

\medskip
\noindent
\underline{\emph{Term \(I\).}}
The matrices
\(\mc G^{-1}(\bm\xi_t)\) are independent and centered so term I is the average of i.i.d. matrices. Note that by \eqref{eq:thm3-G-inverse}, 
\[
\|
\mc G^{-1}(\bm\xi_t)
\|_{\rm F}
\le
\frac L{\Delta_k}.
\]
Furthermore,
\begin{align*}
\mb E\left[ \| \mc G^{-1}(\bm\xi_t) \|_{\rm F}^2 \right] 
&= \sum_{i=1}^k \sum_{j=k+1}^p 
\frac{ \mb E\left[ \left\{ \bm u_j^\top (\bm A_t-\bm\Sigma) \bm u_i \right\}^2 \right] }{ (\lambda_i-\lambda_j)^2},
\end{align*}
because 
\[
[\bm\xi_t]_{j-k,i}
=
\bm u_j^\top\bm A_t\bm u_i
=
\bm u_j^\top
(\bm A_t-\bm\Sigma)
\bm u_i
\qquad 
\text{for \(i\le k<j\). }
\]
Recall $\ell_\delta$ in \eqref{ell}. Applying Lemma \ref{Bernstein's inequality} with \(x=\ell_\delta\), we deduce that 
with probability at least \(1-\delta/8\),
\begin{align}
\|I\|_{\rm F}^2
\lesssim \left[
\frac{\ell_\delta}{N}
\sum_{i=1}^k
\sum_{j=k+1}^p
\frac{
\mb E\left[
\left\{
\bm u_j^\top
(\bm A_1-\bm\Sigma)
\bm u_i
\right\}^2
\right]
}{
(\lambda_i-\lambda_j)^2
}
+
\frac{
L^2\ell_\delta^2
}{
\Delta_k^2N^2
}
\right].
\label{eq:thm3-I-bound}
\end{align}

\medskip
\noindent
\underline{\emph{Term \(II\).}}
To avoid conditioning on a future event, define
\[
\tau_T := \inf\left\{ t\ge m: \|\bm T_t\|_{\rm F}^2>2r_n \right\},
\qquad 
\widetilde{\bm\zeta}_t := \bm\zeta_t \bm 1_{\{\tau_T\ge t\}}.
\]
Then
\[
\mb E[
\widetilde{\bm\zeta}_t
\mid\mc F_{t-1}]
=
\bm0.
\]
By \eqref{eq:thm3-zeta-size},
\eqref{eq:thm3-zeta-variance}, and
\eqref{eq:thm3-G-inverse},
\[
\|
\mc G^{-1}
(\widetilde{\bm\zeta}_t)
\|_{\rm F}
\le
\frac{
CL\sqrt{r_n}
}{
\Delta_k
}
\]
and
\[
\mb E\left[
\|
\mc G^{-1}
(\widetilde{\bm\zeta}_t)
\|_{\rm F}^2
\mid\mc F_{t-1}
\right]
\le
\frac{
CL\lambda_1r_n
}{
\Delta_k^2
}.
\]
Applying Lemma \ref{Bernstein's inequality} with the same $x=\ell_\delta$, where $\ell_\delta$ is as in \eqref{ell}, yields 
\begin{align*}
\left\| \frac1N \sum_{t=m+1}^n \mc G^{-1} (\widetilde{\bm\zeta}_t) \right\|_{\rm F}^2
\lesssim \left[ \frac{ L\lambda_1r_n\ell_\delta }{ \Delta_k^2N}
+ \frac{L^2r_n\ell_\delta^2}{\Delta_k^2N^2} \right]
\end{align*}
with probability at least \(1-\delta/8\).

On \(\mc E_{\rm max } \cap \mc E_w\) (recall that we define $\mc E_w$ in step 1 and $\mc E_{\rm max}$ in step 2d), we have  \(\tau_T>n\), so
\(\widetilde{\bm\zeta}_t=\bm\zeta_t\) for all \(t\le n\), and thus
\begin{align}
\left\| \frac1N \sum_{t=m+1}^n \mc G^{-1} (\widetilde{\bm\zeta}_t) \right\|_{\rm F}^2
\lesssim \left[ \frac{ L\lambda_1r_n\ell_\delta }{ \Delta_k^2N}
+ \frac{L^2r_n\ell_\delta^2}{\Delta_k^2N^2} \right]
\label{eq:thm3-II-bound}
\end{align}
with probability at least \(1-\delta/8\).

\medskip
\noindent
\underline{\emph{Term \(III\).}}
Using \eqref{eq:thm3-T-maximal}, \eqref{eq:thm3-G-inverse}, and \eqref{eq:thm3-remainder-bound}, we obtain 
\begin{align*}
\|III\|_{\rm F}
\le
\frac{C}{\Delta_k \cdot N}
\sum_{t=m+1}^n
\left[
Lr_n+L^2\eta_t
\right]
\le
\frac{C_\alpha}{\Delta_k}
\left[
Lr_n+L^2n^{-\alpha}
\right].
\end{align*}
Therefore,
\begin{equation}
\label{eq:thm3-III-bound}
\|III\|_{\rm F}^2
\le
\frac{C_\alpha}{\Delta_k^2}
\left[
L^2r_n^2+L^4n^{-2\alpha}
\right].
\end{equation}

\medskip
\noindent
\underline{\emph{Term \(IV\).}}
Note that since \(\eta_t=t^{-\alpha}\), we can write 
\begin{align*}
\sum_{t=m+1}^n \frac{ \bm T_{t-1}-\bm T_t}{\eta_t}
=(m+1)^\alpha\bm T_m - n^\alpha\bm T_n
+
\sum_{t=m+1}^{n-1} \left[ (t+1)^\alpha-t^\alpha \right]\bm T_t.
\end{align*}
By \eqref{eq:thm3-T-maximal}, the Frobenius norm of the right-hand side is bounded by $C_\alpha n^\alpha\sqrt{r_n}$.
Hence
\begin{equation} \label{eq:thm3-IV-bound}
\|IV\|_{\rm F}^2 \le \frac{ C_\alpha \cdot n^{2\alpha-2} \cdot r_n}{\Delta_k^2}.
\end{equation}
Combining
\eqref{eq:thm3-PR-decomposition}--\eqref{eq:thm3-IV-bound}, we obtain
\begin{align}
\left\|
\frac1N
\sum_{t=m}^{n-1}\bm T_t
\right\|_{\rm F}^2
\nonumber
&\le
C_\alpha \cdot \left[
\frac{\ell_\delta}{N} \cdot 
\sum_{i=1}^k
\sum_{j=k+1}^p
\frac{
\mb E\left[
\left\{
\bm u_j^\top
(\bm A_1-\bm\Sigma)
\bm u_i
\right\}^2
\right]
}{
(\lambda_i-\lambda_j)^2
}
\right] \\
&+ C_\alpha \cdot \left[
\frac{
L^2\ell_\delta^2
}{
\Delta_k^2N^2
}
+ 
\frac{
L\lambda_1r_n\ell_\delta
}{
\Delta_k^2N
}
+
\frac{
L^2r_n\ell_\delta^2
}{
\Delta_k^2N^2
}
+
\frac{
L^2r_n^2+L^4n^{-2\alpha}
}{
\Delta_k^2
}
+
\frac{
n^{2\alpha-2}r_n
}{
\Delta_k^2 
}
\right].
\label{eq:thm3-average-T}
\end{align}
Keep in mind that $r_n \asymp n^{-\alpha}$ and $\alpha \in (1/2,1)$, so the second term in the right-hand side of \eqref{eq:thm3-average-T} is of order $o(1/N)$.

\medskip
\noindent
\underline{\it Step 5: Bounding the sine-squared loss.}
Define
\[
\widehat{\bm T} := \bm U_\perp^\top \widehat{\bm U}_{\rm avg} \left( \bm U_\star^\top \widehat{\bm U}_{\rm avg} \right)^{-1}.
\]
Note that
\begin{align}
\left\| \sin\bm\Theta
\left( \widehat{\bm U}_{\rm avg}, \bm U_\star \right) \right\|_{\rm F}^2
&\le
\|\widehat{\bm T}\|_{\rm F}^2.
\label{eq:thm3-sine-tangent}
\end{align}
Therefore, to finish the proof, it suffices to show that
\begin{equation} \label{eq:thm3-procrustes-final}
\left\| \widehat{\bm T} - \frac1N \sum_{t=m}^{n-1}\bm T_t \right\|_{\rm F}^2
\lesssim r_n^3.
\end{equation}
Given \eqref{eq:thm3-sine-tangent} and \eqref{eq:thm3-procrustes-final}, the proof is completed by employing \eqref{eq:thm3-average-T}.

It remains to prove \eqref{eq:thm3-procrustes-final}, which is a deterministic
statement on the event \eqref{eq:thm3-T-maximal}. Put $\rho := \sqrt{2r_n}$, so that
\begin{align} \label{eq:thm3-rho}
\max_{m \le t \le n-1} \| \bm T_t \|_{\rm F} \le \rho,
\end{align}
and $\rho^2 \le 2c_0$ is smaller than any prescribed universal constant by
\eqref{eq:thm3-local-smallness}.

Let $\bm R_t$, $m+1 \le t \le n-1$, be the alignment matrices produced by
Algorithm~\ref{alg:kpca-averaging} with $t_0 = m$, and set $\bm R_m := \bm I_k$.
Define the aligned frames, their $\bm U_\star$-blocks and the running averages
\[
\bm Z_t := \bm Q_t \bm R_t,
\qquad
\bm \Gamma_t := \bm U_\star^\top \bm Z_t,
\qquad
\bm G_t := \bm U_\star^\top \overline{\bm Q}_t,
\qquad m \le t \le n-1,
\]
together with the convention $\bm G_{m-1} := \bm \Gamma_m$. Since
$\gamma_t = (t-m+1)^{-1}$, unrolling Algorithm~\ref{alg:kpca-averaging} shows that
the running average is a uniform average,
\begin{align} \label{eq:thm3-uniform-average}
\overline{\bm Q}_t = \frac{1}{t-m+1} \sum_{i=m}^t \bm Z_i,
\qquad \text{and} \qquad
\bm G_t = \frac{1}{t-m+1} \sum_{i=m}^t \bm \Gamma_i .
\end{align}
Also, $\bm U_\star \bm U_\star^\top + \bm U_\perp \bm U_\perp^\top = \bm I_p$ and
$\bm U_\perp^\top \bm Q_t = \bm T_t \lb \bm U_\star^\top \bm Q_t \rb$ give the
factorization
\begin{align} \label{eq:thm3-frame-factorization}
\bm Z_t = \lb \bm U_\star + \bm U_\perp \bm T_t \rb \bm \Gamma_t,
\qquad
\bm U_\perp^\top \bm Z_t = \bm T_t \bm \Gamma_t .
\end{align}
Note that with this notation, our final output has the form
\[
\bm{\hat{U}}_{\rm avg} = \mbox{polar} \lb \overline{\bm Q}_{n-1} \rb, 
\]
and that $\overline{\bm Q}_{n-1}$ is the unnormalized version of $\bm{\hat{U}}_{\rm avg}$. Define 
\[
\overline{\bm T} := \frac1N \sum_{t=m}^{n-1} \bm T_t.
\]
By \eqref{eq:thm3-uniform-average} and \eqref{eq:thm3-frame-factorization},
\[
\bm U_\star^\top \overline{\bm Q}_{n-1} = \bm G_{n-1},
\qquad
\bm U_\perp^\top \overline{\bm Q}_{n-1}
= \frac1N \sum_{t=m}^{n-1} \bm T_t \bm \Gamma_t .
\]
Thus
\begin{align} 
\widehat{\bm T} - \overline{\bm T}
=   \bm U_\perp^\top \,  \overline{\bm Q}_{n-1} \cdot  \left( \bm U_\star^\top \,  \overline{\bm Q}_{n-1} \right)^{-1}  - \overline{\bm T} \nonumber 
&= \left[  \frac1N \sum_{t=m}^{n-1} \bm T_t \bm \Gamma_t \right] \cdot  \bm G_{n-1}^{-1}  - \overline{\bm T} \nonumber \\
&=\left[
\frac1N \sum_{t=m}^{n-1} \bm T_t \lb \bm \Gamma_t - \bm G_{n-1} \rb
\right]
\bm G_{n-1}^{-1} . \nonumber
\end{align}
Consequently,
\begin{align} \label{eq:thm3-covariance-identity}
\widehat{\bm T} - \overline{\bm T}
=
\left[
\frac1N \sum_{t=m}^{n-1} \bm T_t \lb \bm \Gamma_t - \bm G_{n-1} \rb
\right]
\bm G_{n-1}^{-1} .
\end{align}
Write
\begin{align}
    \| \widehat{\bm T} - \overline{\bm T}  \|_{\rm F} &\leq \| \bm G_{n-1}^{-1} \|_{\rm op} 
    \cdot \frac 1N \sum_{t=m}^{n-1} \| \bm T_t \|_{\rm F} \cdot \|  \bm \Gamma_t - \bm G_{n-1}  \|_{\rm op} \nonumber \\
    &\leq \rho \cdot \| \bm G_{n-1}^{-1} \|_{\rm op}  \cdot  \frac 1N \sum_{t=m}^{n-1} \|  \bm \Gamma_t - \bm G_{n-1}  \|_{\rm op}
    \label{eq:averaging-bound}
\end{align}
where we have used the elementary bound 
\[
\left\| \bm A \bm B \bm C \right\|_{\rm F} \leq \| \bm A \|_{\rm F} \cdot \| \bm B \|_{\rm op} \cdot \| \bm C \|_{\rm op}.
\]
To finish the proof, we will bound two terms $ \| \bm G_{n-1}^{-1} \|_{\rm F}$ and $ \sum_{t=m}^{n-1} \|  \bm \Gamma_t - \bm G_{n-1}  \|_{\rm F}$ separately. Before doing so, let us collect some useful bounds. Write
\begin{align*}
    \bm I_k = \bm Z_t^\top \bm Z_t & = \bm \Gamma_t^\top \cdot \lb \bm U_\star + \bm U_\perp \bm T_t \rb^\top \cdot \lb \bm U_\star + \bm U_\perp \bm T_t \rb \cdot \bm \Gamma_t \\
    &= \bm \Gamma_t^\top \lb \bm I_k + \bm T_t^\top \bm T_t \rb \bm \Gamma_t.
\end{align*}
Consequently $\bm \Gamma_t^\top \bm \Gamma_t \preceq \bm I_k$, so
$\| \bm \Gamma_t \|_{\rm op} \le 1$, and
\[
\left\| \bm \Gamma_t^\top \bm \Gamma_t - \bm I_k \right\|_{\rm op}
= \left\| \bm \Gamma_t^\top \bm T_t^\top \bm T_t \bm \Gamma_t \right\|_{\rm op} \le \rho^2
\]
for all $t \in [m,n-1]$. 

Write $\bm O_t := \operatorname{polar}(\bm \Gamma_t)$, so that
$\bm \Gamma_t = \bm O_t \lb \bm \Gamma_t^\top \bm \Gamma_t \rb^{1/2}$.
Lemma~\ref{lem:psd-sqrt} applied with $\bm B = \bm I_k$ yields
\begin{align} \label{eq:thm3-Gamma-near-orth}
\left\| \bm \Gamma_t - \bm O_t \right\|_{\rm op} \le \rho^2 .
\end{align}
Also, by \eqref{eq:thm3-uniform-average},
\eqref{eq:thm3-frame-factorization} and \eqref{eq:thm3-rho}, we have
$\left\| \bm G_t \right\|_{\rm op} \le 1$ and
\begin{align} \label{eq:thm3-perp-block}
\left\| \bm U_\perp^\top \overline{\bm Q}_t \right\|_{\rm op}
\le \max_{m \le i \le t} \left( \| \bm T_i \|_{\rm op} \cdot  \| \bm \Gamma_i \|_{\rm op} \right)
\le \rho .
\end{align}

\noindent
\underline{\emph{Step 5a: Bounding  $ \| \bm G_{n-1}^{-1} \|_{\rm op}$.}}
Set 
\[\bm P_t := \lb \bm G_t^\top \bm G_t \rb^{1/2},\] 
so that $\bm P_t \preceq \bm I_k$, because $\left\| \bm G_t \right\|_{\rm op} \le 1$. We prove by induction that, for a universal
constant $c_1 \ge 1$ and all $m \le t \le n-1$,
\begin{align} \label{eq:thm3-induction}
\bm P_t \succeq \lb 1 - c_1 \rho^2 \rb \bm I_k,
\qquad \text{and} \qquad
\| \bm \Gamma_t - \bm G_{t-1} \|_{\rm op} \le 2c_1 \rho^2 ,
\end{align}
where $\rho$ is small enough that $c_1 \rho^2 \le 1/2$.

Note that once \eqref{eq:thm3-induction} is proved, we immediately get
\begin{align} \label{eq:thm3-Gn^-1-inverse}
    \left\| \bm G_{n-1}^{-1} \right\|_{\rm op} \leq 2. 
\end{align}
For $t=m$ both claims are immediate: $\bm G_m = \bm G_{m-1} = \bm \Gamma_m$, and
$$\bm P_m^2 = \bm \Gamma_m^\top \bm \Gamma_m \succeq \lb 1-\rho^2 \rb \bm I_k.$$
Let $t \ge m+1$, put $j := t-m+1$, and assume the claims in
\eqref{eq:thm3-induction} are true  at $t-1$; in particular
$\lambda_{\min}(\bm P_{t-1}) \ge 1/2$, so $\bm G_{t-1}$ is invertible and admits the
polar decomposition 

$$\bm G_{t-1} = \bm \Omega_{t-1} \bm P_{t-1} \qquad \text{with
$\bm \Omega_{t-1} \in O(k)$.}$$ 
Define $\bm P' := \lb \overline{\bm Q}_{t-1}^\top \overline{\bm Q}_{t-1} \rb^{1/2}$.
Then $\bm P'^2 = \bm P_{t-1}^2
+ \lb \bm U_\perp^\top \overline{\bm Q}_{t-1} \rb^\top
\lb \bm U_\perp^\top \overline{\bm Q}_{t-1} \rb$, so
$\bm P'^2 \succeq \bm P_{t-1}^2$ and, by \eqref{eq:thm3-perp-block},
$ \| \bm P'^2 - \bm P_{t-1}^2 \|_{\rm op} \le \rho^2$. Also, 
\[\lambda_{\rm min} \lb \bm P'  \rb \geq \lambda_{\rm min} \lb \bm P_{t-1}  \rb \geq 1/2.\] 
Note that the bound above also yields $\| \bm P'^{-1} \|_{\rm op} \le 2$. Lemma~\ref{lem:psd-sqrt} then gives
\[
\| \bm P' - \bm P_{t-1} \|_{\rm op} \le \frac{ \| \bm P'^2 - \bm P_{t-1}^2 \|_{\rm op}}{ \lambda_{\rm min} \lb \bm P'  \rb +\lambda_{\rm min} \lb \bm P_{t-1}  \rb  } \leq \frac{\rho^2}{1/2+1/2} = \rho^2. 
\]
Since $\widetilde{\bm Q}_{t-1} = \overline{\bm Q}_{t-1} \bm P'^{-1}$, we obtain
\begin{align} \label{eq:thm3-tilde-blocks}
\left\| \bm U_\star^\top \widetilde{\bm Q}_{t-1} - \bm \Omega_{t-1} \right\|_{\rm op}
= \left\| \bm G_{t-1} \,  \bm P'^{-1} -  \bm \Omega_{t-1}   \right\|_{\rm op}
&= \left\| \bm \Omega_{t-1} \bm P_{t-1} \,  \bm P'^{-1} -  \bm \Omega_{t-1}   \right\|_{\rm op} \nonumber \\
&=\left\| \bm P_{t-1} \bm P'^{-1} - \bm I_k \right\|_{\rm op} \le 2\rho^2,  \\
\left\| \bm U_\perp^\top \widetilde{\bm Q}_{t-1} \right\|_{\rm op} &\le 2\rho , \nonumber
\end{align}
where the last bound follows from \eqref{eq:thm3-perp-block}. 

By the definition of $\bm R_t$ in Algorithm~\ref{alg:kpca-averaging}, the matrix
$\bm Z_t^\top \widetilde{\bm Q}_{t-1} = \bm U_t \bm D_t \bm U_t^\top$ is symmetric
positive semi-definite. We claim that there exists $\bm F \in \mb R^{k \times k} $ such that
\[
\bm Z_t^\top \widetilde{\bm Q}_{t-1}
= \bm \Gamma_t^\top \bm U_\star^\top \widetilde{\bm Q}_{t-1}
+ \bm \Gamma_t^\top \bm T_t^\top \bm U_\perp^\top \widetilde{\bm Q}_{t-1}
= \bm O_t^\top \bm \Omega_{t-1} + \bm F,
\qquad \text{and} \qquad 
\| \bm F \|_{\rm op} \le 6\rho^2 .
\]
To see this, substitute $\bm \Gamma_t= \bm \Gamma_t-\bm O_t+\bm O_t$ and $\bm U_\star^\top \widetilde{\bm Q}_{t-1} = \bm \Omega_{t-1} + \bm E$ with $\| \bm E \|_{\rm op} \leq 2\rho^2$ into the above, then write
\begin{align*}
    \bm \Gamma_t^\top \bm U_\star^\top \widetilde{\bm Q}_{t-1} + \bm \Gamma_t^\top \bm T_t^\top \bm U_\perp^\top \widetilde{\bm Q}_{t-1}
    = \bm O_t^\top \bm \Omega_{t-1} + \bm F
\end{align*}
where
\[
\bm F:=  \bm \Gamma_t^\top \bm T_t^\top \bm U_\perp^\top \widetilde{\bm Q}_{t-1} +  \bm O_t^\top \bm E+ \lb \bm \Gamma_t - \bm O_t \rb^\top \cdot \lb \bm \Omega_{t-1} + \bm E \rb.
\]
The claim that $\| \bm F \|_{\rm op} \le 6\rho^2$ follows from 
\begin{align*}
    \left\|    \bm \Gamma_t^\top \bm T_t^\top \bm U_\perp^\top \widetilde{\bm Q}_{t-1} \right\|_{\rm op} &\leq 2\rho^2,
    \qquad 
    \left\|  \bm O_t^\top \bm E \right\|_{\rm op} \leq 2\rho^2, \\
        \left\|  \lb \bm \Gamma_t - \bm O_t \rb^\top \cdot \lb \bm \Omega_{t-1} + \bm E \rb \right\|_{\rm op}
        &\leq \rho^2(1+2\rho^2) \leq 2\rho^2. 
\end{align*}
Here we use \eqref{eq:thm3-Gamma-near-orth} and the triangle inequality to deduce the third estimate. 

Since $\bm O_t^\top \bm \Omega_{t-1}$ is orthogonal, Lemma~\ref{lem:orth-psd} gives
$\| \bm O_t^\top \bm \Omega_{t-1} - \bm I_k \|_{\rm op} \le 6\sqrt2 \rho^2$, and
therefore, using \eqref{eq:thm3-Gamma-near-orth} once more,
\begin{align} \label{eq:thm3-Gamma-Omega}
\| \bm \Gamma_t - \bm \Omega_{t-1} \|_{\rm op} \le 10 \rho^2.
\end{align}
Now, the second claim in \eqref{eq:thm3-induction} follows from
\[
\| \bm \Gamma_t - \bm G_{t-1} \|_{\rm op}  \leq \| \bm \Gamma_t - \bm \Omega_{t-1} \|_{\rm op} + \| \bm \Omega_{t-1} - \bm G_{t-1}\|_{\rm op} \leq 10\rho^2 +c_1 \rho^2 \leq 2c_1\rho^2,
\]
which holds as long as $c_1 \geq 10$.

For the first claim, \eqref{eq:thm3-uniform-average} gives
\[
\bm \Omega_{t-1}^\top \bm G_t
=
\frac{j-1}{j} \cdot \bm P_{t-1}
+
\frac1j \cdot  \bm \Omega_{t-1}^\top \bm \Gamma_t .
\]
Its symmetric part is bounded below, by the induction hypothesis and
\eqref{eq:thm3-Gamma-Omega}, by
\[\left[ \frac{j-1}{j} \lb 1-c_1\rho^2 \rb + \frac1j \lb 1-c_1\rho^2 \rb \right] \bm I_k
= \lb 1 - c_1 \rho^2 \rb \bm I_k.\] 
Consequently,
\[
\lambda_{\min}(\bm P_t) = \sigma_{\min}(\bm G_t)
= \sigma_{\min} \lb \bm \Omega_{t-1}^\top \bm G_t \rb
\ge 1 - c_1 \rho^2 .
\]
This finishes the induction.

\noindent
\underline{\emph{Step 5b: Bounding  $ \sum_{t=m}^{n-1} \|  \bm \Gamma_t - \bm G_{n-1}  \|_{\rm op}$.}}
Write
\[
\frac1N \sum_{t=m}^{n-1} \| \bm \Gamma_t - \bm G_{n-1} \|_{\rm op}
\le
\frac1N \sum_{t=m}^{n-1}
\left[ \| \bm \Gamma_t - \bm G_{t-1} \|_{\rm op}
+ \| \bm G_{t-1} - \bm G_{n-1} \|_{\rm op} \right] \leq 2c_1\rho^2 + \frac 1N 
 \sum_{t=m}^{n-1} \| \bm G_t - \bm G_{t-1} \|_{\rm op}.
\]
where the second inequality follows from \eqref{eq:thm3-induction}. 

For $i \ge m+1$ we have $\bm G_i - \bm G_{i-1} = \frac{1}{i-m+1} \lb \bm \Gamma_i - \bm G_{i-1} \rb$, and \eqref{eq:thm3-induction} again gives
\[
(i-m+1) \cdot \| \bm G_i - \bm G_{i-1} \|_{\rm op} \leq 2c_1\rho^2. 
\]
Consequently,
\[
\sum_{t=m}^{n-1} \| \bm G_{t-1} - \bm G_{n-1} \|_{\rm op}
\le
\sum_{t=m}^{n-1} \sum_{i=t}^{n-1} \| \bm G_i - \bm G_{i-1} \|_{\rm op}
=
\sum_{i=m}^{n-1} (i-m+1) \| \bm G_i - \bm G_{i-1} \|_{\rm op}
\le
2c_1 N \rho^2 .
\]
Thus
\[
\frac1N \sum_{t=m}^{n-1} \| \bm \Gamma_t - \bm G_{n-1} \|_{\rm op}
\le
\frac1N \sum_{t=m}^{n-1}
\left[ \| \bm \Gamma_t - \bm G_{t-1} \|_{\rm op}
+ \| \bm G_{t-1} - \bm G_{n-1} \|_{\rm op} \right]
\le
4c_1 \rho^2 .
\]
Now, substituting the above into \eqref{eq:averaging-bound} and using  \eqref{eq:thm3-Gn^-1-inverse}, we get
\[
\left\| \widehat{\bm T} - \overline{\bm T} \right\|_{\rm F} 
\le
2 \cdot \frac1N \sum_{t=m}^{n-1}
\| \bm T_t \|_{\rm F} \cdot  \| \bm \Gamma_t - \bm G_{n-1} \|_{\rm op}
\le
8 c_1 \rho^3 .
\]
Squaring and substituting $\rho^2 = 2r_n$ proves
\eqref{eq:thm3-procrustes-final}. 

Finally, from \eqref{eq:thm3-average-T}, our final bound is of the form 
\begin{align*}
\left\|
\sin\bm\Theta
\left(
\widehat{\bm U}_{\rm avg},
\bm U_k^\star
\right)
\right\|_{\rm F}^2
\le
C_\alpha
\left[
\frac{\ell_\delta}{n}
\sum_{i=1}^k
\sum_{j=k+1}^p
\frac{
\mb E\left[
\left\{
\bm u_j^\top
(\bm A_1-\bm\Sigma)
\bm u_i
\right\}^2
\right]
}{
(\lambda_i-\lambda_j)^2
}
+
\mc E_n
\right]
\end{align*}
with $\mc E_n =o(1/n)$. The warm-start event fails with probability at most \(\delta/4\),
the maximal event fails conditionally with probability at most
\(\delta/4\), and the two Bernstein events fail with total
probability at most \(\delta/4\). A union bound therefore gives the
asserted probability at least \(1-\delta\). This proves
\eqref{eq:kpca-phase-two-final}.

\hfill $\square$

\section{Technical results}

\subsection{Proof of Theorem \ref{thm:eca}} \label{sec:eca-proof}

The proof consists of verifying that the transformed data satisfy the
assumptions of Section~\ref{sec:k-PCA-convergence} with universal constants, and
then invoking Theorem~\ref{thm:kpca-phase-two}. Write
\[
\widetilde{\bm\Sigma}
:=
\mb E\left[ \bm{\mc X}_1\bm{\mc X}_1^\top \right],
\qquad
\bm{\mc A}_t := \bm{\mc X}_t\bm{\mc X}_t^\top .
\]
By Lemma~\ref{lem:eca-sign-data}, the stream $\bm{\mc X}_1,\ldots,\bm{\mc X}_n$
satisfies the assumptions of Section~\ref{sec:k-PCA-convergence} with
\[
M=1, \qquad L\le2, \qquad \mu_1\le1, \qquad \lambda_1(\widetilde{\bm\Sigma})\le1, \qquad \sigma_k^2\le1,
\]
target $\bm U_k^\star$ and eigengap $\widetilde\Delta_k>0$. 

Substituting these
into \eqref{eq:phase-two-Mn} gives $M_{n,\delta}\lesssim\widetilde
M_{n,\delta}$, and also substituting these bounds into
\eqref{eq:kpca-phase-two-error-terms}--\eqref{eq:kpca-phase-two-remainder} gives
$\mc E_n\lesssim\widetilde{\mc E}_n$.
Finally, note that
$$(L\mu_1)^{1/(2\alpha-1)}\le2^{1/(2\alpha-1)}\lesssim_\alpha1 \qquad \text{and} \qquad \ell_\delta\le\widetilde M_{n,\delta}^{1/\alpha},$$ 
so
\eqref{eq:eca-sample-size} implies
\eqref{eq:kpca-phase-two-sample-size}. 

Theorem~\ref{thm:kpca-phase-two} now
gives \eqref{eq:eca-final}, the variance term being the one computed in
Lemma~\ref{lem:eca-sign-data}\emph{(iv)}.  The proof is completed. $\hfill\square$

\subsection{Technical lemmas}

\begin{lemma}[Bernstein's inequality] \label{Bernstein's inequality}
 Suppose $\bm Y_1,\ldots,\bm Y_q$ are martingale differences in a finite-dimensional Hilbert space $(\mb H, \|.\|)$   with respect to some filtration $\la  \mc F_t; t \geq 1\ra$. Assume that
\[
\|\bm Y_s\|\le b, \qquad \text{and} \qquad 
\sum_{s=1}^q
\mb E\left[
\|\bm Y_s\|^2
\mid\mc F_{s-1}
\right]
\le v.
\]
Then, for every $x>0$, with probability at least $1-2e^{-x}$,
\begin{equation*}
\left\| \sum_{s=1}^q\bm Y_s \right\| \le C\left( \sqrt{vx}+bx\right)
\end{equation*}
for some universal constant $C>0$.
\end{lemma}

\noindent \textbf{Proof of Lemma \ref{Bernstein's inequality}.}
Let $\bm S_j:= \sum_{s=1}^j \bm Y_s$. Theorem 3.4 in \cite{pinelis1994optimum} yields
\[
\mb P \lb \max_{1\leq j \leq q} \| \bm S_j \| \geq t \rb \leq 2 \exp \la -\frac{v}{b^2} \cdot h \lb \frac{bt}{v} \rb \ra
\]
where $h(u):= (1+u) \log (1+u) - u$.

By using the elementary inequality
\[
h(u) \geq \frac{u^2}{2(1+u/3)}, \qquad u\geq 0,
\]
this implies
\[
\mb P \lb \max_{1\leq j \leq q} \| \bm S_j \| \geq t \rb \leq 2 \cdot \exp \la - \frac{t^2}{2(v+bt/3)} \ra.
\]
Consequently, with probability at least $1-2e^{-x}$,
\[
 \| \bm S_q \|  \leq \max_{1\leq j \leq q} \| \bm S_j \| \leq \sqrt{2vx} + \frac{2bx}{3}. 
\]
Note that this is stronger than what we need. The proof is completed. $\hfill$ $\square$

\begin{lemma} \label{lem:psd-sqrt}
Let $\bm A,\bm B\in\mb R^{k\times k}$ be symmetric positive semi-definite matrices such that
$$\lambda_{\min}(\bm A)+\lambda_{\min}(\bm B)>0.$$ 
Then
\[
\| \bm A - \bm B \|_{\rm op}
\le
\frac{\| \bm A^2 - \bm B^2 \|_{\rm op}}
{\lambda_{\min}(\bm A)+\lambda_{\min}(\bm B)} .
\]
\end{lemma}

\noindent \textbf{Proof of Lemma \ref{lem:psd-sqrt}.}
 Let $\bm v$ be a unit eigenvector of $\bm A-\bm B$ with
eigenvalue $\nu$ satisfying $|\nu| = \| \bm A - \bm B \|_{\rm op}$. From the expression
$$\bm A^2 - \bm B^2 = \bm A \lb \bm A - \bm B \rb + \lb \bm A - \bm B \rb \bm B$$
and the symmetry of $\bm A - \bm B$, we have 
\[
\bm v^\top \lb \bm A^2 - \bm B^2 \rb \bm v
= \nu \, \bm v^\top \bm A \bm v + \nu \, \bm v^\top \bm B \bm v
= \nu \, \bm v^\top \lb \bm A + \bm B \rb \bm v .
\]
Hence 
$$|\nu| \cdot \left[ \lambda_{\min}(\bm A)+\lambda_{\min}(\bm B) \right]
\le \| \bm A^2 - \bm B^2 \|_{\rm op}.$$ $\hfill \square$

\begin{lemma} \label{lem:orth-psd}
Let $\bm O\in\mb R^{k\times k}$ be an orthogonal matrix and let $\bm F\in\mb R^{k\times k}$
satisfy $\|\bm F\|_{\rm op}\le\ve\le 1/2$. Suppose $\bm O+\bm F$ is positive
semi-definite. Then, we have 
$$\| \bm O - \bm I_k \|_{\rm op} \le \sqrt2 \, \ve.$$
\end{lemma}

\noindent \textbf{Proof of Lemma \ref{lem:orth-psd}.}
Decompose 
$$\bm O = \bm M + \bm K$$ 
with
\[
\bm M := \frac12 \lb \bm O + \bm O^\top \rb \qquad \text{and} \qquad \bm K := \frac12 \lb \bm O - \bm O^\top \rb.
\]
Note that $\bm M$ is symmetric and $\bm K$ is antisymmetric. Using 
$\bm O^\top\bm O = \bm I_k$ and $\bm O\bm O^\top = \bm I_k$  gives $\bm M^2 - \bm K^2 = \bm I_k$. Thus
\[
\bm I_k - \bm M^2 = \bm K^\top \bm K \succeq \bm 0 .
\]
Since $\bm O + \bm F$ is symmetric, $\bm O - \bm O^\top = \bm F^\top - \bm F$, so
$\| \bm K \|_{\rm op} \le \ve$. Therefore, every eigenvalue $\mu$ of $\bm M$
satisfies 
$$\mu^2 \ge 1-\ve^2.$$ 
Being symmetric, $\bm O + \bm F$ coincides with its
own symmetric part $\bm M + \frac12 \lb \bm F + \bm F^\top \rb$, so positive
semi-definiteness gives $\bm M \succeq -\ve \bm I_k$. 

As $\ve\le1/2$ we have
$\sqrt{1-\ve^2}>\ve$, and hence $\mu \ge \sqrt{1-\ve^2}$ for every such $\mu$.
Finally,
\[
\lb \bm O - \bm I_k \rb^\top \lb \bm O - \bm I_k \rb = 2 \lb \bm I_k - \bm M \rb,
\]
so $\| \bm O - \bm I_k \|_{\rm op}^2 \le 2 \lb 1 - \sqrt{1-\ve^2} \rb \le 2\ve^2$.
$\hfill \square$

\begin{lemma} \label{lem:one-step-exterior}
Let $\la \bm\omega_t; t \geq 1 \ra$ be generated by the embedded Oja recursion
\eqref{embbeded Oja} and  let $\la \mc F_t \ra$ be the filtration generated by the
data. Suppose $\eta_t L \le 1/4$. Put
\[
S_t^\wedge := 1 - \left| \langle \bm\omega_t, \bm\omega_\star \rangle_\wedge \right|^2 .
\]
Then there is a sequence $\la Q_t \ra$, adapted to $\la \mc F_t \ra$ (the natural filtration generated by the data), such that
\begin{align} \label{eq:one-step-Q}
\mb E \left[ Q_t \mid \mc F_{t-1} \right] = 0,
\qquad
\left| Q_t \right| \le 4L\eta_t \sqrt{S_{t-1}^\wedge},
\end{align}
and such that, for a universal constant $C>0$,
\begin{align} \label{eq:exterior-one-step}
S_t^\wedge
\le
\lb 1 - 2\Delta_k\eta_t \rb \cdot S_{t-1}^\wedge
+
2\Delta_k\eta_t \cdot \lb S_{t-1}^\wedge \rb^2
+
Q_t
+
CL^2 \cdot \eta_t^2 .
\end{align}
\end{lemma}

\noindent \textbf{Proof of Lemma \ref{lem:one-step-exterior}.}
Conditionally on $\mc F_{t-1}$, put
\[
c := \left\langle \bm\omega_{t-1},\bm\omega_\star \right\rangle_\wedge,
\qquad
u_t := \left\langle \bm\omega_\star, \mc L_{\bm A_t}\bm\omega_{t-1} \right\rangle_\wedge,
\]
\[
v_t := \left\langle \bm\omega_{t-1}, \mc L_{\bm A_t}\bm\omega_{t-1} \right\rangle_\wedge,
\qquad
w_t := \left\| \mc L_{\bm A_t}\bm\omega_{t-1} \right\|_\wedge^2 ,
\]
and write $\overline{\mc L} := \mc L_{\bm\Sigma}$. 

Throughout we use
\eqref{eq:rank-one-powers-exterior}, which gives
$0 \preceq \mc L_{\bm A_t} \preceq L \cdot \bm I_{\widetilde{\mb R}_{k,p}}$ and,
after taking expectations, $\mu_1 \le L$. Since
$\| \bm\omega_{t-1} \|_\wedge = \| \bm\omega_\star \|_\wedge = 1$, this yields
\begin{align} \label{eq:one-step-crude}
0 \le v_t \le L,
\qquad
|u_t| \le L,
\qquad
0 \le w_t \le L^2 .
\end{align}
By \eqref{embbeded Oja},
\[
\left\langle \bm\omega_t,\bm\omega_\star \right\rangle_\wedge
=
\frac{c+\eta_t u_t}
{\left\| \lb \bm I_{\widetilde{\mb R}_{k,p}} + \eta_t \mc L_{\bm A_t} \rb \bm\omega_{t-1} \right\|_\wedge},
\qquad
\left\| \lb \bm I_{\widetilde{\mb R}_{k,p}} + \eta_t \mc L_{\bm A_t} \rb \bm\omega_{t-1} \right\|_\wedge^2
=
D_t ,
\]
where $D_t := 1 + 2\eta_t v_t + \eta_t^2 w_t$.

Hence, using $1-c^2 = S_{t-1}^\wedge$,
\[
S_t^\wedge
=
\frac{D_t - \lb c+\eta_t u_t \rb^2}{D_t}
=
\frac{S_{t-1}^\wedge + 2\eta_t \lb v_t - cu_t \rb + \eta_t^2 \lb w_t - u_t^2 \rb}{D_t} .
\]
Subtracting $S_{t-1}^\wedge$ and using $1-S_{t-1}^\wedge = c^2$ to simplify the
numerator gives the exact identity
\begin{align} \label{eq:thm3-exterior-exact}
S_t^\wedge - S_{t-1}^\wedge
=
\frac{2\eta_t \lb c^2 v_t - cu_t \rb + \eta_t^2 \lb c^2 w_t - u_t^2 \rb}{D_t} .
\end{align}
Put 
\[
A_t := 2\eta_t \lb c^2v_t - cu_t \rb \qquad \text{and} \qquad B_t := \eta_t^2 \lb c^2w_t - u_t^2 \rb.
\]
It is immediate that
$$ S_t^\wedge - S_{t-1}^\wedge = \lb A_t+B_t \rb/D_t.$$ 
Moreover, by \eqref{eq:one-step-crude}, $D_t \ge 1$ and
$|D_t - 1| \le 2\eta_t L + \eta_t^2L^2 \le 3\eta_t L$, while
$|B_t| \le \eta_t^2 L^2$ and $|A_t| \le 4L\eta_t$ (using $|c| \le 1$). Therefore,
\[
\left| \frac{A_t+B_t}{D_t} - A_t \right|
\le
|A_t| \cdot \frac{|D_t-1|}{D_t} + \frac{|B_t|}{D_t}
\le
12L^2\eta_t^2 + L^2\eta_t^2
=
CL^2\eta_t^2 ,
\]
and consequently
\begin{align} \label{eq:one-step-after-denominator}
S_t^\wedge
\le
S_{t-1}^\wedge + A_t + CL^2 \cdot \eta_t^2 .
\end{align}
Write
\[
\bm\omega_{t-1} = c \cdot \bm\omega_\star + \sqrt{S_{t-1}^\wedge} \cdot \bm e,
\qquad
\bm e \perp \bm\omega_\star,
\qquad
\|\bm e\|_\wedge = 1 .
\]
Since $\overline{\mc L}\bm\omega_\star = \mu_1 \bm\omega_\star$ and
$\bm e \perp \bm\omega_\star$, we have
\[
\mb E[v_t \mid \mc F_{t-1}] = c^2\mu_1 + S_{t-1}^\wedge \langle \bm e, \overline{\mc L}\bm e\rangle_\wedge
\qquad \text{and}
\qquad 
\mb E[u_t \mid \mc F_{t-1}] = c\mu_1.
\]
 Hence
\[
\mb E \left[ c^2v_t - cu_t \mid \mc F_{t-1} \right]
=
c^2 S_{t-1}^\wedge
\left[ \left\langle \bm e, \overline{\mc L}\bm e \right\rangle_\wedge - \mu_1 \right] .
\]
By Theorem~\ref{thm:mean lift}, the restriction of $\overline{\mc L}$ to
$\bm\omega_\star^\perp$ has largest eigenvalue $\mu_2 = \mu_1 - \Delta_k$, so
$$\langle \bm e, \overline{\mc L}\bm e\rangle_\wedge \le \mu_1 - \Delta_k.$$
Together with $c^2 = 1-S_{t-1}^\wedge$ this gives
\begin{align} \label{eq:thm3-exterior-drift}
\mb E \left[ A_t \mid \mc F_{t-1} \right]
\le
-2\Delta_k\eta_t \cdot S_{t-1}^\wedge \lb 1 - S_{t-1}^\wedge \rb .
\end{align}
Define 
$$Q_t := A_t - \mb E[A_t \mid \mc F_{t-1}],$$ 
which is $\mc F_t$-measurable and satisfies the first claim in \eqref{eq:one-step-Q}. 

Since
$c^2v_t - cu_t
= \left\langle c^2\bm\omega_{t-1} - c\bm\omega_\star, \mc L_{\bm A_t}\bm\omega_{t-1} \right\rangle_\wedge$
and the same identity holds with $\mc L_{\bm A_t}$ replaced by $\overline{\mc L}$, we deduce that
\[
Q_t
=
2\eta_t
\left\langle
c^2\bm\omega_{t-1} - c\bm\omega_\star,
\lb \mc L_{\bm A_t} - \overline{\mc L} \rb \bm\omega_{t-1}
\right\rangle_\wedge .
\]
Substituting $\bm\omega_{t-1} = c\bm\omega_\star + \sqrt{S_{t-1}^\wedge}\bm e$ and
using $c^2-1 = -S_{t-1}^\wedge$, we get the following expression for the first term inside the inner product above:
\[
c^2\bm\omega_{t-1} - c\bm\omega_\star
=
-c S_{t-1}^\wedge \cdot \bm\omega_\star
+
c^2 \sqrt{S_{t-1}^\wedge} \cdot \bm e ,
\]
whose squared exterior norm is
$c^2 \lb S_{t-1}^\wedge \rb^2 + c^4 S_{t-1}^\wedge = c^2 S_{t-1}^\wedge$. 
Hence
$$\| c^2\bm\omega_{t-1}-c\bm\omega_\star \|_\wedge = |c| \sqrt{S_{t-1}^\wedge}.$$ 
Moreover, since
$$\| \mc L_{\bm A_t} - \overline{\mc L} \|_{\rm op} \le L + \mu_1 \le 2L,$$
Cauchy--Schwarz gives the second claim in \eqref{eq:one-step-Q}.

Finally, substituting $A_t = \mb E[A_t \mid \mc F_{t-1}] + Q_t$ together with
\eqref{eq:thm3-exterior-drift} into \eqref{eq:one-step-after-denominator} yields
\eqref{eq:exterior-one-step}. The proof is completed. $\hfill \square$

\begin{lemma} \label{lem:eca-sign-data}
Under the assumptions of Theorem~\ref{thm:eca}, the vectors
$\bm{\mc X}_1,\ldots,\bm{\mc X}_n$ are i.i.d.\ and the following hold.
\begin{enumerate}[label=(\roman*)]
\item We have 
$$\displaystyle
\bm{\mc X}_1
\stackrel{d}{=}
\frac{\bm\Sigma^{1/2}\bm U}{\left\| \bm\Sigma^{1/2}\bm U \right\|},
\qquad \bm U \sim \mathrm{Unif}(\mb S^{p-1}) .
$$
In particular the law of $\bm{\mc X}_1$ depends on neither $\bm\mu$ nor the
distribution of $R$.
\item 
$\mb E\bm{\mc X}_1 = \bm 0$ and $\|\bm{\mc X}_1\|=1$. Consequently, the
assumptions of Section~\ref{sec:k-PCA-convergence} hold with $M=1$ and
$L=1+\tilde\lambda_1\le2$.
\item
$\widetilde{\bm\Sigma}=\sum_{i=1}^p\tilde\lambda_i\bm u_i\bm u_i^\top$ with
$\tilde\lambda_i$ as in \eqref{eq:eca-transformed-spectrum}. Thus
$\widetilde{\bm\Sigma}$ and $\bm\Sigma$ have the same eigenvectors. Moreover, it holds that
\[
\sum_{i=1}^p\tilde\lambda_i=1 \qquad \text{and} \qquad \text{$\lambda_i\ge\lambda_j$ implies
$\tilde\lambda_i\ge\tilde\lambda_j$}.
\]
Furthermore, the inequality is strict when
$\lambda_i>\lambda_j$ and $\lambda_i>0$.

\item $\mu_1=\sum_{i\le k}\tilde\lambda_i\le1$,
$\tilde\lambda_1\le1$, and
$\sigma_k^2 = \sum_{i\le k}\sum_{j>k}\tilde\sigma_{ij}^2 \le 1$.
\end{enumerate}
\end{lemma}

\noindent \textbf{Proof of Lemma \ref{lem:eca-sign-data}.}
Independence is obvious because the pairs
$(\bm X_{2t-1},\bm X_{2t})$, $t=1,\ldots,n$, are disjoint.

\emph{(i)} By \eqref{elliptical model}, we have
$$\bm X_{2t}-\bm X_{2t-1} \stackrel{d}{=} \bm\Sigma^{1/2}\bm W$$ 
with
$\bm W := R_{2t}\bm U_{2t}-R_{2t-1}\bm U_{2t-1}$; the mean $\bm\mu$ cancels.

For any orthogonal $\bm O$ we have
$$\bm O\bm W = R_{2t}\bm O\bm U_{2t}-R_{2t-1}\bm O\bm U_{2t-1}
\stackrel{d}{=}\bm W,$$ 
so $\bm W$ is spherically symmetric. 

Hence $\bm W\stackrel{d}{=}\|\bm W\|\cdot\bm U$ with $\bm U\sim\mathrm{Unif}(\mb
S^{p-1})$ independent of $\|\bm W\|$. Since $R$ is continuous,
$\bm W\ne\bm 0$ almost surely. The positive scalar $\|\bm W\|$ therefore
cancels between numerator and denominator of \eqref{eq:eca-preprocessing}.

\emph{(ii)} By \emph{(i)}, $\bm{\mc X}_1\stackrel{d}{=}-\bm{\mc X}_1$, whence
$\mb E\bm{\mc X}_1=\bm0$; and $\|\bm{\mc X}_1\|=1$ by construction. Both
$\bm{\mc A}_1$ and $\widetilde{\bm\Sigma}$ are positive semi-definite with
operator norm at most one, so
$-\widetilde{\bm\Sigma}\preceq\bm{\mc A}_1-\widetilde{\bm\Sigma}\preceq
\bm{\mc A}_1$ gives $\|\bm{\mc A}_1-\widetilde{\bm\Sigma}\|_{\rm op}\le1$.
Hence $M=1$ and $L=\tilde\lambda_1+M$.

\emph{(iii)} Recall $D$ in \eqref{eq:eca-transformed-spectrum}. By \emph{(i)},
$\bm u_i^\top\bm{\mc X}_1=\sqrt{\lambda_i}\,U_i/\sqrt D$, so
$$\bm u_i^\top \, \widetilde{\bm\Sigma} \, \bm u_j=\mb E \left[\sqrt{\lambda_i\lambda_j} \cdot \frac{U_iU_j}{D} \right].$$ 
For $i\ne j$ the map $U_i\mapsto-U_i$ preserves the law of $\bm U$
and fixes the value of $D$, while flipping the sign of $U_iU_j$. Hence the off-diagonal
entries vanish and the diagonal ones are the $\tilde\lambda_i$. 
Summing gives
$$\sum_i\tilde\lambda_i=\mb E\|\bm{\mc X}_1\|^2=1.$$
For the ordering, fix $i\ne j$ and condition on $\la U_\ell:\ell\ne i,j \ra$ and
on $U_i^2+U_j^2$. Conditionally, $(a,b):=(U_i^2,U_j^2)$ is exchangeable.
Define
\[
c:=\sum_{\ell\ne i,j}\lambda_\ell U_\ell^2, \qquad 
Q:=c+\lambda_ib+\lambda_ja.
\]
By the distributional identity $(a,b)\stackrel{d}{=}(b,a)$, we have 
\[
2\lb \tilde\lambda_i-\tilde\lambda_j \rb
=
\mb E\left[
\frac{(\lambda_ia-\lambda_jb)Q+(\lambda_ib-\lambda_ja)D}{DQ}
\right] .
\]
Expanding the numerator,
\[
(\lambda_ia-\lambda_jb)Q+(\lambda_ib-\lambda_ja)D
=
\lb \lambda_i-\lambda_j \rb
\left[ c(a+b)+2\lb \lambda_i+\lambda_j \rb ab \right] ,
\]
which is non-negative when $\lambda_i\ge\lambda_j$.

\emph{(iv)} The first two bounds follow from $\sum_i\tilde\lambda_i=1$. For the
variance, $\bm u_j^\top\widetilde{\bm\Sigma}\bm u_i=0$ for $i\ne j$ by
\emph{(iii)}, so
$$\sigma_k^2=\sum_{i\le k}\sum_{j>k}\mb E[(\bm u_i^\top\bm{\mc X}_1)^2
(\bm u_j^\top\bm{\mc X}_1)^2]=\sum_{i\le k}\sum_{j>k}\tilde\sigma_{ij}^2.$$
Since $\sum_{i=1}^p\lambda_iU_i^2/D=1$ almost surely,
\[
\sum_{i=1}^k\sum_{j=k+1}^p \frac{\lambda_i\lambda_jU_i^2U_j^2}{D^2}
\le
\lb \sum_{i=1}^p \frac{\lambda_iU_i^2}{D} \rb^2
=1 ,
\]
and taking expectations gives $\sigma_k^2\le1$. $\hfill\square$

\section*{\centering{Statement on AI use}}
The  embedding argument, maximal inequality and the strategy of the non-asymptotic error bound belong to the authors. Various wedge product equalities and estimates were suggested by GPT pro 5.6. GPT pro 5.6 was also used to perform various calculations in the proofs. In addition, it was used to perform the simulation studies, to do real-data analysis, and to edit the language. The authors take full responsibility of the correctness of all the results.



\bibliographystyle{plainnat}
\bibliography{paper-ref}

\end{document}